\documentclass[11pt,letterpaper]{amsart}
\usepackage{float}
\usepackage{amsmath, amssymb, amsthm, epsfig, color, mathtools, bbm, dsfont, mathrsfs, bm, mathabx}
\usepackage[latin1]{inputenc}
\usepackage[T1]{fontenc}
\usepackage{indentfirst}
\usepackage[shortlabels]{enumitem}
\usepackage[title]{appendix}
\usepackage{lmodern}       
\usepackage{setspace}
\usepackage{soul}
\usepackage{accents}
\usepackage{hyperref} 
\usepackage{xcolor}
\usepackage{relsize}
\usepackage{scalerel}
\usepackage{tikz}
\usepackage{xr}    
\usepackage{graphicx} 
\usepackage[export]{adjustbox} 
\usepackage{subcaption} 
\usepackage{siunitx}
\usepackage{bm} 
\usepackage{comment}
\usepackage{tikz, tkz-euclide}
\usetikzlibrary{patterns}
\usetikzlibrary{arrows.meta}
\usetikzlibrary{calc}
\usetikzlibrary{decorations.markings}
\usepackage{fancybox}
\newcounter{dummy} \numberwithin{dummy}{section}
  \theoremstyle{plain}
  \newtheorem*{theorem*}        {Theorem}
  \newtheorem{theorem}[dummy]          {Theorem}
  \newtheorem{lemma}[dummy]              {Lemma}
  \newtheorem*{lemma*}          {Lemma}
    
  \newtheorem{corollary}[dummy]           {Corollary}

  \newtheorem{remark}[dummy]           {Remark}
  \newtheorem{notation}[dummy]           {Notation}
  \theoremstyle{remark}
  \theoremstyle{definition}
   \newtheorem{definition}[dummy]          {Definition}
   \newtheorem{hypothesis}{Hypothesis}
   \newtheorem{convention}{Convention}
   
\makeatletter

\def\moverlay{\mathpalette\mov@rlay}
\def\mov@rlay#1#2{\leavevmode\vtop{%
		\baselineskip\z@skip \lineskiplimit-\maxdimen
		\ialign{\hfil$\m@th#1##$\hfil\cr#2\crcr}}}
\newcommand{\charfusion}[3][\mathord]{
	#1{\ifx#1\mathop\vphantom{#2}\fi
		\mathpalette\mov@rlay{#2\cr#3}
	}
	\ifx#1\mathop\expandafter\displaylimits\fi}

\makeatother

\newcommand{\I}{\mathrm{I}}
\newcommand{\E}{\mathrm{E}}
\newcommand{\Pt}{\mathbbold{P}}
\newcommand{\Gr}{\mathbbold{G}}
\newcommand{\Tr}{\mathbbold{T}}

\newcommand{\W}{\mathrm{W}}
\newcommand{\rS}{\mathrm{S}}
\newcommand{\degr}{\mathrm{deg}}
\newcommand{\diam}{\mathrm{diam}}
\newcommand{\dist}{\mathrm{dist}}

\DeclareSymbolFont{bbold}{U}{bbold}{m}{n}
\DeclareSymbolFontAlphabet{\mathbbold}{bbold}
\newcommand{\ind}{\mathbbold{1}}

\allowdisplaybreaks
\begin{document}

\title[Cluster Expansion and Decay of Correlations]{A Cluster Expansion and the Decay of Correlations of the 1d Long-Range Ising Model at low temperatures}
\author[Bissacot]{RODRIGO BISSACOT}
\author[Corsini]{HENRIQUE CORSINI}


\begin{abstract}
      In this work, a convergent low-temperature cluster expansion of the one-dimensional long-range ferromagnetic Ising model with polynomial decay $\alpha\in (1,2]$ is developed; that is, $J(r)=r^{-\alpha}$. As an application, the $n$-point correlations are studied and the two-point correlation is shown to be algebraic with rate of decay exactly $\alpha$. 
\end{abstract}

\maketitle

\tableofcontents

\section{Introduction}\label{sec: introduction}

\subsection{A brief history of the one-dimensional long-range Ising model}\label{subsec: brief-history}

In 1967, Gallavotti and Miracle-Sole \cite{Gallavotti.Miracle-Sole.67} showed the existence of the thermodynamic limit and infinite-volume correlation functions for a large class of lattice models with spin space $\{0,1\}$ and translation-invariant potentials. In particular, they showed the existence of the thermodynamic limit and infinite-volume correlation functions for regular long-range ferromagnetic Ising models. A year later, Ruelle \cite{Ruelle.68-no.P.T.} showed the absence of spontaneous magnetization if the long-range interaction decayed too fast, that is, if $\alpha>2$. In 1969, employing a hierarchical model, Dyson \cite{Dyson.69} gave the first proof of the existence of phase transition for a ferromagnet with long-range polynomial decay $\alpha\in (1,2)$. In the following years, progressively stronger criteria for the lack of long-range order were proved in \cite{Dyson.69-no.P.T.,Rogers1981}. 

In 1980, Bricmont, Lebowitz and Pfister \cite{BLP.1980} showed that, under Ruelle's condition, any equilibrium state is translation invariant. In 1982, Fannes, Vanheuverzwijn and Verbeure \cite{Fannes.Vanheuverzwijn.Verbeure.82} established a criterion for translation invariance of equilibrium states in terms of correlation inequalities, which show, in particular, that any equilibrium state of a regular one-dimensional long-range ferromagnetic Ising model is translation invariant. A similar result (translation invariance of equilibrium states in the one-dimensional case) was obtained independently by Burkov and Sinai in \cite{Burkov.Sinai.83}.

The matter of phase-transition for critical decay remained open until \cite{Frohlich.Spencer.82}. In 1982, Fr\"ohlich and Spencer successfully employed a new contour notion, which had already proved very fruitful in \cite{FS81}, where the first rigorous proof of the BKT transition for the planar XY model was given, to show the existence of long-range order at small temperature for $\alpha=2$. In \cite{Imbrie.82}, building on the work of Fr\"ohlich and Spencer, Imbrie provided a converging cluster expansion for the critical exponent and, as an application, proved the algebraic decay of the two-point function.

In 1988, employing techniques of percolation theory and the hypothesis that the interaction between first neighbors be arbitrarily large (which was absent in \cite{Frohlich.Spencer.82,Imbrie.82}), Aizenman, Chayes, Chayes and Newman \cite{Aizenman1988} proved the non-continuity of the magnetization for the critical polynomial decay (Thouless effect). In the same year and building on the previous work, Imbrie and Newman \cite{Imbrie.Newman.88} proved the existence of an intermediary phase, where correlations decay algebraically with an exponent smaller than $2$; that is, a BKT-like transition.

Nearly twenty years later, interest in the one-dimensional long-range ferromagnetic Ising model was renewed, when Cassandro, Ferrari, Merola and Presutti \cite{Cassandro.05} gave a proof via contours of phase transition in case $\alpha\in (3/2,2]$ (the same proof works for $\alpha\in(3-\log_23,2]$). Similarly to the works cited in the previous paragraph, they assumed that the interaction between first neighbors is sufficiently large. Under this hypothesis and employing the same contours, Cassandro, Orlandi and Picco \cite{Cassandro.Picco.09} proved the stability of phase-transition under the influence of an external random field, provided $\alpha\in(3-\log_23,3/2)$ and the variance $\varepsilon$ be sufficiently small. Recently, this result was extended to $\alpha\in (1,3/2)$ by Ding, Maia and Huang in \cite{Ding} without the introduction of the first-neighbor interaction as a perturbative factor.

In \cite{Cassandro.Merola.Picco.Rozikov.14}, Cassandro, Merola, Picco and Rozikov provided a convergent cluster expansion, and as an application studied the phase separation point in the one-dimensional long-range Ising model (we recall the absence of non-translation invariant equilibrium states in 1D). They also enunciate a bound for the $n$-site correlation function (see Theorem 2.5 in \cite{Cassandro.Merola.Picco.Rozikov.14}). In \cite{Cassandro.Merola.Picco.17}, the authors dealt with the question of phase separation for $\alpha\in (3-\log_23,2)$, once again employing cluster expansion techniques.

A first attempt to eliminate the hypothesis that the interaction between first-neighbors is sufficiently large can be found in \cite{Bissacot-Kimura2018}. Bissacot, Endo, van Enter, Kimura and Ruszel were partially successful: employing the same contours as \cite{Cassandro.05,Cassandro.Merola.Picco.17,Cassandro.Merola.Picco.Rozikov.14,Cassandro.Picco.09}, the authors proved the existence of phase-transition provided $\alpha\in (\alpha^*,2]$, where $\zeta(\alpha^*)=2$. We note that $3-\log_23<\alpha^*$. Another work dealing with the same contours is \cite{LittinPicco2017}, Littin and Picco show that the energy estimates break down for $\alpha<3-\log_23$.

Motivated by the very fruitful adaptation of the Fr\"ohlich-Spencer contours to the multidimensional case \cite{Pereira,Affonso.2021,potts,Johanes,cluster,maia2024phase}, Affonso, Bissacot, Corsini and Welsch \cite{corsini1dpt} obtained, with appropriate modifications to energy estimates, a proof of phase transition via (Fr\"ohlich-Spencer) contours for the region $\alpha\in (1,2]$, without the introduction of the nearest neighbor interaction as a perturbative factor. This fact motivates revisiting, and complementing, whenever possible, many results where, while a contour approach was successful, the introduction of a perturbative first-neighbors interaction was necessary. 

In this work, we provide a converging cluster expansion and, as an application, show that at low temperature the two-point correlation decays polynomially, with the same decay rate as $J(r)=r^{-\alpha}$; the lower-bound was already proved in the seventies by Iagolnitzer and Souillard \cite{Iag77}. In order to study the $n$-point correlation function, we also revisit (abstractly) the technique of cluster expansion in the very particular scenario, where compatibility is (graph)-transitive. 

On the one hand, summing over leaves of trees of polymers is very standard (see, for example, \cite{Dobrushin96,Dobrushin96b,Fernndez2007,Koteck1986,Malyshev1991-yo,procacci2023cluster}). The convergence of the cluster expansion present in this work can be proved with the traditional criteria (Koteck\'y-Preiss \cite{Koteck1986}, Dobrushin \cite{Dobrushin96, Dobrushin96b}, Fern\'andez-Procacci \cite{Fernndez2007}). For a comparison between traditional methods, see \cite{BisFerProc10}. For recent technical improvements, see \cite{Jansen2022,temmelthese,Temmel2014}. On the other hand, given the form of the activity of this particular gas of polymers, we sum over trees of contours; that is, trees where compatibility conditions are global. Hence, it is reasonable (and useful) to introduce more general operations which are represented graphically by removing vertices of an arbitrary degree.

\subsection{Description of the results}

In this work, we develop a convergent low-temperature cluster expansion of the one-dimensional long-range ferromagnetic Ising model with polynomial decay $\alpha\in (1,2]$ in terms of polymer of Fr\"ohlich-Spencer contours. In Section~\ref{sec: preliminaries}, we revisit the definition of the contours first introduced in \cite{Frohlich.Spencer.82} and list the crucial estimates that ensure the convergence of the cluster expansion.

In Section~\ref{sec: polymers}, we derive a representation of the one-dimensional Ising model in terms of a hard-core gas of polymers, which are collections of contours satisfying a "positive" compatibility relation. In this section, we also prove an upper bound to the activity of the polymer gas which is more convenient for computations. 

In Section~\ref{sec: sums-over-trees}, we prove the usual "one-vertex" bound employing a convergence criterion similar to Koteck\'y-Preiss \cite{Koteck1986}; however, we note that our computation is closer to that of Fern\'andez-Procacci \cite{Fernndez2007}. We also prove (many-vertex) estimates for a special type of expansion in terms of trees where compatibility is global. This is useful because of the particular form of the upper bound to the activity of the polymer gas. In Section~\ref{sec: convergence}, we apply the abstract bounds present in Section~\ref{sec: sums-over-trees} to our particular case. Transforming trees of polymers into trees of contours, we prove the convergence of the representation of the pressure given by the cluster expansion. That is, we prove the following theorem:

\begin{theorem}\label{theo: convergence-cluster-expansion}
			If $\alpha\in(1,2]$, then there exist $\beta_0=\beta_0(\alpha)>0$ such that if $\beta\geq\beta_0$, then \begin{align}
				\sup_{\Lambda\Subset\mathbb{Z}}\frac{1}{\lvert\Lambda\rvert}\sum_{n=1}^\infty\frac{1}{n!}\sum_{\mathbf{\Gamma}\in\mathscr{P}^n_{\Lambda}}\lvert{\phi_n(\mathbf{\Gamma})}\rvert\prod_{i=1}^nz_{\beta}(\Gamma_i)<\infty
			\end{align} 
			If $A\Subset\mathbb{Z}$ and $\beta\geq\beta_0$, there exists $r_0\coloneqq r_0(\beta,\lvert A\rvert)>0$ such that 
			\begin{align}
				h\mapsto p_\beta(h)\coloneqq\lim_{\Lambda\uparrow\mathbb{Z}}\frac{1}{\lvert\Lambda\rvert}\sum_{n=1}^\infty\frac{1}{n!}\sum_{\mathbf{\Gamma}\in\mathscr{P}^n_\Lambda}\phi_n(\mathbf{\Gamma})\prod_{i=1}^nz^+_{\beta,h}(\Gamma_i)
			\end{align} 
			is analytic in the polydisk $\mathbb{D}_{r_0}^A\coloneqq\{h\in\mathbb{C}^A:\lVert{h}\rVert_\infty\leq r_0\}$.
		\end{theorem}

In Section~\ref{sec: point-estimates}, we establish some estimates in terms of the lattice sites in the spirit of \cite{Imbrie.82}. Furthermore, we prove a lemma that changes estimates in terms of contours into estimates in terms of sites. In Section~\ref{sec: correlations}, we apply Theorem~\ref{theo: convergence-cluster-expansion} and the contour-site estimates to obtain an upper bound to the decay rate of the correlation functions of the one-dimensional long-range ferromagnetic Ising model with polynomial decay $\alpha\in (1,2]$. Thus, we obtain the following results:

\begin{theorem}\label{theo: two-point-correlations}
    There exist $c_3>0$ and $\beta_0>0$ such that if $\beta\geq\beta_0$ and $x\neq y\in\mathbb{Z}$, then
    \begin{equation}\label{eq: two-point-correlations}
        \rvert\langle\sigma_x;\sigma_y\rangle^+_{\beta}\rvert\leq 2e^{-c_3\beta}\lvert x-y\rvert^{-\alpha}.
    \end{equation}
\end{theorem}

\begin{theorem}\label{theo: many-point-correlations}
    For all $N\geq 3$, there exist $\beta_0\coloneqq\beta_0(N)$ and $c_4\coloneqq c_4(N)>0$ such that if $\beta\geq\beta_0$, $A\Subset \mathbb{Z}$ and $\lvert A\rvert=N$, then \begin{equation}\label{eq: many-point-correlations}
        \lvert\langle:\sigma_A:\rangle^+_{\beta}\rvert\leq 2e^{-c_4\beta}\sum_{T\in\Tr(A)}\prod_{\{x_i,x_j\}\in E(T)}\hspace{-.35cm}\lvert x_i-x_j\rvert^{-\alpha}.
    \end{equation} 
\end{theorem}

We note that combining Theorem~\ref{theo: two-point-correlations} with \cite{Iag77}, we have thus proven that, for $\beta$ large, the two-point correlation decay rate is \textit{exactly} $\alpha$. We recall the existence of an intermediary phase in the case of $\alpha=2$ \cite{Imbrie.Newman.88}.

\section{Preliminaries}\label{sec: preliminaries}

\subsection{The long-range one-dimensional ferromagnetic Ising model}\label{subsec: long_range_1d_Ising_model} 

The configuration space of the one-dimensional Ising model is given by $\Omega=\{-1,+1\}^\mathbb{Z}.$ Given $\omega\in\Omega$ and $\Lambda\Subset\mathbb{Z}$, we denote the set of configurations satisfying the $\omega$-boundary condition outside of $\Lambda$ by \begin{equation*}
    \Omega^\omega_\Lambda\coloneqq\{\sigma\in\Omega:(x\notin\Lambda\Rightarrow\sigma_x=\omega_x)\}.
\end{equation*} We recall that $\Lambda\Subset\mathbb{Z}$ denotes that $\Lambda$ is a finite subset of $\mathbb{Z}$. In a similar fashion, we denote the set of all configurations satisfying the $\omega$-boundary condition by \begin{equation*}
    \Omega^\omega\coloneqq\bigcup_{\Lambda\Subset\mathbb{Z}}\Omega^\omega_\Lambda.
\end{equation*}If $\omega_x=+1$ for all $x\in\mathbb{Z}$, we write $\Omega^+$ instead of $\Omega^\omega$ and say that $\sigma\in\Omega^+$ satisfies the plus boundary condition. If $\omega_x=-1$ for all $x\in\mathbb{Z}$, we proceed in an analogous way. Henceforth, we restrict ourselves to dealing with homogeneous boundary conditions, that is, either the plus boundary condition or the minus boundary condition. 

Let $J:\mathbb{Z}\times\mathbb{Z}\rightarrow\mathbb{R}$ be a coupling function and $\omega\in\Omega$ be a boundary condition, we define the Hamiltonian $H^\omega_J:\Omega^\omega\rightarrow\mathbb{R}$ by \begin{equation}\label{eq: def-hamiltonian}
    H^\omega_J(\sigma)\coloneqq\frac{1}{2}\sum_{(x,y)}J_{xy}(\omega_x\omega_y-\sigma_x\sigma_y).
\end{equation}In this work, we assume that\begin{equation}\label{eq: polynomial-decay-coupling}
    J_{xy}=\begin{cases}
        0 &\textrm{if }x=y,\\
        \lvert x-y\rvert^{-\alpha}&\textrm{otherwise,}
    \end{cases}
\end{equation}where $1<\alpha\leq2$. That is, we assume that the model is translation invariant, ferromagnetic, with a two-body regular \textit{essentially} long-range interaction. Henceforth, we omit $J$ from the notation for better writing. 

Let $h:\mathbb{Z}\rightarrow\mathbb{R}$ be an external field and $\omega\in\Omega$ be a boundary condition, we define the Hamiltonian $H^\omega_h:\Omega^\omega\rightarrow\mathbb{R}$ by \begin{equation*}\label{eq: def-hamiltonian-ext-field}
    H^\omega_h(\sigma)\coloneqq\frac{1}{2}\sum_{(x,y)}J_{xy}(\omega_x\omega_y-\sigma_x\sigma_y)+\sum_{x}h_x(\omega_x-\sigma_x)\eqqcolon H^\omega(\sigma)+E^\omega_h(\sigma).
\end{equation*}
Given $\Lambda\Subset\mathbb{Z}$ and $\beta>0$, we denote the \textit{partition function} at the inverse temperature $\beta$ by \begin{equation*}\label{eq: def-partition-function}
    Z^\omega_{\Lambda,\beta}\coloneqq\sum_{\sigma\in\Omega^\omega_\Lambda}e^{-\beta H^\omega(\sigma)}.
\end{equation*}By the global spin-flip map $\sigma\mapsto-\sigma$, it is easy to see that $Z^-_{\Lambda,\beta}=Z^+_{\Lambda,\beta}$. If the external field $h$ is not identically zero, we define \begin{equation*}
    Z^\omega_{\Lambda,\beta,h}\coloneqq\sum_{\sigma\in\Omega^\omega_\Lambda}e^{-\beta H_h^\omega(\sigma)}.
\end{equation*}We note that if $h\not\equiv0$, then $Z^-_{\Lambda,\beta,h}\neq Z^+_{\Lambda,\beta,h}$. 

\subsection{Configurations as collections of spin flips}\label{subsec: spin-flip-configs}

In this work, following \cite{Frohlich.Spencer.82}, we think of configurations of the one-dimensional Ising model as collections of spin flips. We associate a configuration $\sigma\in\Omega$ to \begin{equation*}
    \partial\sigma\coloneqq\{e=\{x,y\}\in\mathbb{Z}^*:\sigma_x\neq\sigma_y\},
\end{equation*}
where $\mathbb{Z}^*\simeq\mathbb{Z}+1/2$ denotes the dual graph of $\mathbb{Z}$. This application is two-to-one and invariant under the global spin-flip map $\sigma\mapsto -\sigma$. 

Since we shall be dealing with configurations with homogeneous boundary conditions, we define our configuration space as \begin{equation*}\label{eq: configuration-space}
    \Omega^*\coloneqq\{\sigma^*\in\mathrm{P}(\mathbb{Z}^*):\lvert\sigma^*\rvert\textrm{ is even}\},
\end{equation*} where $\mathrm{P}(X)$ denotes the set of all finite parts of $X$. We define the canonical identification $\sigma^+:\Omega^*\rightarrow\Omega^+$ by \begin{equation*}\label{eq: identification-config-spaces}
    \sigma^+(\partial\sigma)=\sigma,
\end{equation*}
for all $\sigma\in\Omega^+$, that is, if $\sigma$ satisfies the plus boundary condition. For the rest of this work, unless noted otherwise, we identify collections of spin flips with configurations satisfying the plus boundary condition. 

Given $\sigma^*\in\Omega^*$, we define its $-$-interior by \begin{equation*}\label{eq: minus-interior}
    \mathrm{I}_-(\sigma^*)\coloneqq\{x\in\mathbb{Z}:\sigma^+(\sigma^*)=-1\}.
\end{equation*}
We note that for \textit{any} coupling function \begin{equation}\label{eq: dual-hamiltonian}
    H(\sigma^*)\coloneqq\hspace{-0.25cm}\sum_{\substack{x\in \mathrm{I}_-(\sigma^*)\\y\in\mathrm{I}_-(\sigma^*)^c}}\hspace{-0.25 cm}2J_{xy}=H^+(\sigma^+),
\end{equation}
where $\sigma^+\equiv\sigma^+(\sigma^*)$. In particular, this holds for the case \eqref{eq: polynomial-decay-coupling}. 

\begin{remark}
    This geometric description (independent of the finite box considered) of the Hamiltonian justifies the definition present in Equation~\eqref{eq: def-hamiltonian}. In any case, without this normalization, which is not detected by Gibbs' measures, there is no hope that a cluster expansion would converge.
\end{remark}

If $\sigma^*_1,\sigma^*_2\in\Omega^*$ satisfy $\I_-(\sigma_1^*)\cap\I_-(\sigma^*_2)=\emptyset$, it is also useful to define \begin{equation}\label{eq: definition-phi}
    \Phi(\sigma^*_1,\sigma^*_2)\coloneqq\hspace{-0.25cm}\sum_{\substack{x\in \I_-(\sigma^*_1)\\y\in\I_-(\sigma^*_2)}}\hspace{-0.25cm}4J_{xy}.
\end{equation}Finally, if $h\not\equiv0$ is an external field, we define \begin{equation}\label{eq: external-fiels-component-definition}
    E_h(\sigma^*)\coloneqq\hspace{-0.25cm}\sum_{x\in\I_-(\sigma^*)}\hspace{-0.25cm}2h_x=E_h^+(\sigma^+)=-E_h^-(\sigma^-)
\end{equation} and $H_h(\sigma^*)=H(\sigma^*)+E_h(\sigma^*)$.

\subsection{Long-range contours}\label{subsec: long-range-contours}

In \cite{Frohlich.Spencer.82}, Fr\"ohlich and Spencer showed the existence of a first-order phase transition for the critical exponent $\alpha=2$ using the notion of a geometric contour that need not be connected, which they first introduced a year earlier in \cite{FS81} to give the first rigorous proof of the BKT transition to the first-neighbor two-dimensional XY model. Such contours are defined by a fixed distancing exponent $a=3/2$ and a linear parameter $M>1$ whose value is chosen during the proof. 

In \cite{Affonso.2021}, the authors managed to prove the existence of a first-order phase transition in the region $\alpha>d\geq 2$, where $d$ denotes the dimension, employing a modification to the definition originally given by Fr\"ohlich and Spencer in which the choice of the distancing exponent itself depended on $(d,\alpha)$. It turns out that, at least in one dimension, this does not need to be the case. In \cite{corsini1dpt} it was shown that the same contours employed by Fr\"ohlich and Spencer are sufficient to prove the existence of a first-order phase transition in the whole region of polynomial decay where such a phenomenon occurs. In fact, \textit{any} distancing exponent $a\in(1,2)$ is sufficient. 

We recall the slight modification to the definition of an irreducible contour introduced by Imbrie in \cite{Imbrie.82}, where, by employing a cluster expansion method, he established sharp estimates of the correlations in the case $\alpha=2$.

\begin{notation}
    Let $n\in\mathbb{N}$, we write $[n]\coloneqq\{1,...,n\}$.
\end{notation}

\begin{notation}\label{not: graphs-trees-partitions}
    Let $V$ be a finite set. We denote by $\Gr(V)$ the set of unordered connected graphs $G$ such that $V(G)=V$. Similarly, we denote by $\Tr(V)$ the set of all unordered connected trees $T$ such that $V(T)=V$. Finally, we denote the set of unordered partitions of $V$ by $\Pt(V)$. If $V=[n]$, we write $\Gr_n$, $\Tr_n$ and $\Pt_n$ instead.
\end{notation}

\begin{convention}
    If $V$ is a finite set and $P\in\Pt(V)$, then $\emptyset\notin P$. If $\lvert P\rvert=1$, we say that $P$ is the trivial partition of $V$.
\end{convention}

\begin{definition}[$M$-contour]\label{def: M-contour}
    Fix $M>1$. We say that $\emptyset\neq\gamma\in\Omega^*$ is an $M$\textit{-irreducible contour} if for any $P\in\Pt(\gamma)\cap\mathrm{P}(\Omega^*)$ not trivial, there exists $\gamma_1\neq\gamma_2\in P$ such that\begin{equation*}
        \mathrm{dist}(\gamma_1,\gamma_2)\leq M(\min\{\mathrm{diam}(\gamma_1),\mathrm{diam}(\gamma_2)\})^{3/2}.
    \end{equation*}We denote the set of all $M$-irreducible contours by $\mathscr{C}_M$. 
    
    If $M>1$ is clear from context, we omit $M$ from the notation and say that $\gamma\in\mathscr{C}$ is a \textit{contour}. If $\gamma_1,\gamma_2\in\mathscr{C}$ are not empty and \begin{equation*}
        \mathrm{dist}(\gamma_1,\gamma_2)>M(\min\{\mathrm{diam}(\gamma_1),\mathrm{diam}(\gamma_2)\})^{3/2},
    \end{equation*}we say that $\gamma_1$ and $\gamma_2$ are \textit{compatible}, denoted by $\gamma_1\sim\gamma_2$. Note that $\gamma\not\sim\gamma$.
\end{definition}

\begin{definition}[$M$-partition]\label{def: M-partition}
    Let $\emptyset\neq\sigma^*\in\Omega^*$. We say that $\Gamma\in \Pt(\sigma^*)\cap\mathrm{P}(\mathscr{C})$ is an $M$\textit{-partition of }$\sigma^*$ if \begin{equation*}
        (\gamma_1,\gamma_2\in\Gamma)\wedge(\gamma_1\neq\gamma_2)\Rightarrow\gamma_1\sim\gamma_2.
    \end{equation*}We denote the set of all $M$-partitions by $\mathscr{E}$. 

    Following \eqref{eq: dual-hamiltonian}, we define $H_h:\mathscr{E}\rightarrow\mathbb{R}$ by \begin{equation*}\label{eq: hamiltonian-partition}
        H_h(\Gamma)\coloneqq H_h\left(\sqcup_{\gamma\in\Gamma}\gamma\right).
    \end{equation*}
\end{definition}

\begin{lemma}\label{lem: M-partition}
    Let $\emptyset\neq \sigma^*\in\Omega^*$, then there is a unique $M$-partition of $\sigma^*$.
\end{lemma}

\begin{proof}
    We refer the reader to Proposition 2.1 in \cite{Imbrie.82}.
\end{proof}

\begin{corollary}\label{cor: M-partition}
    Let $\Gamma_1,\Gamma_2\in\mathscr{E}$ such that $(\gamma_1\in\Gamma_1)\wedge(\gamma_2\in\Gamma_2)\Rightarrow\gamma_1\sim\gamma_2$, then $\Gamma_1\sqcup\Gamma_2\in\mathscr{E}$.
\end{corollary}

\begin{convention}
    We establish $\emptyset\notin\mathscr{C}$ and $\emptyset\notin\mathscr{E}$. The reason for this shall be clear when dealing with collections of compatible polymers.
\end{convention}

\begin{definition}\label{def: positive-partitions}
    If $\gamma_1,\gamma_2\in\mathscr{C}$ satisfy $\gamma_1\sim\gamma_2$ and $\mathrm{I}_-(\gamma_1)\cap\mathrm{I}_-(\gamma_2)=\emptyset,$ we say that they are \textit{positively compatible}, which we denote by $\gamma_1\parallel\gamma_2.$

    If $\Gamma_1,\Gamma_2\in\mathscr{E}$ are such that $(\gamma_1\in\Gamma_1)\wedge(\gamma_2\wedge\Gamma_2)\Rightarrow\gamma_1\parallel\gamma_2$, we define \begin{equation*}\label{eq: Phi-definition-for-collections-of-contours}
        \Phi(\Gamma_1,\Gamma_2)\coloneqq\Phi(\sqcup_{\gamma_1\in\Gamma_1}\gamma_1,\sqcup_{\gamma_2\in\Gamma_2}\gamma_2).
    \end{equation*}We say that $\Gamma\in\mathscr{E}$ is \textit{positive} if \begin{equation*}
        \gamma_1\neq\gamma_2\in\Gamma\Rightarrow\gamma_1\parallel\gamma_2.
    \end{equation*}We denote the set of all positive $M$-partitions by $\mathscr{P}$.
\end{definition}

\begin{figure}[hbt!]
    \centering
    \tikzset{every picture/.style={line width=0.75pt}} 

\begin{tikzpicture}[x=0.75pt,y=0.75pt,yscale=-1,xscale=1]

\draw    (40,110.38) -- (360.17,110.38) ;
\draw [color={rgb, 255:red, 208; green, 2; blue, 27 }  ,draw opacity=1 ][line width=1.5]    (50.25,100.31) -- (50.25,120.6) ;
\draw [color={rgb, 255:red, 245; green, 166; blue, 35 }  ,draw opacity=1 ][line width=1.5]    (120.25,100.31) -- (120.25,120.6) ;
\draw [color={rgb, 255:red, 245; green, 166; blue, 35 }  ,draw opacity=1 ][line width=1.5]    (110.25,100.31) -- (110.25,120.6) ;
\draw [color={rgb, 255:red, 208; green, 2; blue, 27 }  ,draw opacity=1 ][line width=1.5]    (200.25,100.31) -- (200.25,120.6) ;
\draw [color={rgb, 255:red, 74; green, 144; blue, 226 }  ,draw opacity=1 ][line width=1.5]    (240.25,100.31) -- (240.25,120.6) ;
\draw [color={rgb, 255:red, 74; green, 144; blue, 226 }  ,draw opacity=1 ][line width=1.5]    (250.25,100.31) -- (250.25,120.6) ;
\draw [color={rgb, 255:red, 208; green, 2; blue, 27 }  ,draw opacity=1 ][line width=1.5]    (280.25,100.31) -- (280.25,120.6) ;
\draw [color={rgb, 255:red, 208; green, 2; blue, 27 }  ,draw opacity=1 ][line width=1.5]    (350.25,100.31) -- (350.25,120.6) ;
\draw  [dash pattern={on 0.84pt off 2.51pt}] (30,90) -- (390,90) -- (390,140) -- (30,140) -- cycle ;

\draw (372,103.4) node [anchor=north west][inner sep=0.75pt]    {$\Gamma $};
\draw (55,113.4) node [anchor=north west][inner sep=0.75pt]  [color={rgb, 255:red, 208; green, 2; blue, 27 }  ,opacity=1 ]  {$\gamma _{1}$};
\draw (125,113.4) node [anchor=north west][inner sep=0.75pt]  [color={rgb, 255:red, 245; green, 166; blue, 35 }  ,opacity=1 ]  {$\gamma _{2}$};
\draw (255,113.4) node [anchor=north west][inner sep=0.75pt]  [color={rgb, 255:red, 74; green, 144; blue, 226 }  ,opacity=1 ]  {$\gamma _{3}$};

\end{tikzpicture}
    \caption{All three contours are compatible, but $\gamma_1$ is not positively compatible with $\gamma_2$. On the other hand, $\gamma_3$ is positively compatible with both $\gamma_1$ and $\gamma_2$. Hence, $\Gamma=\{\gamma_1,\gamma_2,\gamma_3\}$ is \textit{not} a positive collection.}
    \label{fig: positive-compatibility}
\end{figure}
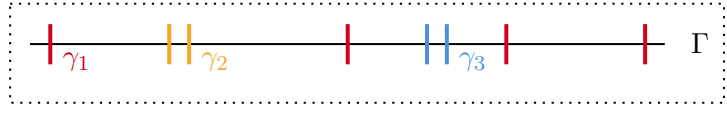

\begin{remark}
    For those who are colorblind, let us give an alternative description of Figure~\ref{fig: positive-compatibility}. Ordering spin flips so that $i<j$ implies that $b_i$ lies to the left of $b_j$, we have \begin{align*}
        &\gamma_1=\{b_1,b_4,b_7,b_8\},
        &\gamma_2=\{b_2,b_3\},\\
        &\gamma_3=\{b_5,b_6\}.
    \end{align*}
\end{remark}

\begin{lemma}\label{lem: positive-partition}
    Let $\Gamma_1,\Gamma_2\in\mathscr{P}$ such that $(\gamma_1\in\Gamma_1)\wedge(\gamma_2\in\Gamma_2)\Rightarrow\gamma_1\parallel\gamma_2$, then $\Gamma_1\sqcup\Gamma_2\in\mathscr{P}$. In this case, we write $\Gamma_1\parallel\Gamma_2$.
\end{lemma}

\begin{proof}
    This follows directly from the fact that compatibility is checked in pairs. That is, $\Gamma\in\mathscr{E}$ if and only if $\gamma_1\neq\gamma_2\in\Gamma\Rightarrow\gamma_1\sim\gamma_2.$
\end{proof}

\begin{notation}\label{not: sets-in-a-box}
    If $\Lambda\Subset\mathbb{Z}$, we write 
    \begin{align*}
        &\mathscr{C}_\Lambda\coloneqq\{\gamma\in\mathscr{C}:\mathrm{I}_-(\gamma)\subset\Lambda\},\\
        &\mathscr{E}_\Lambda\coloneqq\{\Gamma\in\mathscr{E}:\gamma\in\Gamma\Rightarrow\gamma\in\mathscr{C}_\Lambda\},\\
        &\mathscr{P}_\Lambda\coloneqq\mathscr{P}\cap\mathscr{E}_\Lambda.
    \end{align*}
\end{notation}

\begin{definition}\label{def: partial-ordering}
    Let $\gamma_1,\gamma_2\in\mathscr{C}$. If $\I_-(\gamma_1)\subset\I_-(\gamma_2)$, we write $\gamma_1\trianglelefteq\gamma_2$. This defines a partial ordering in $\mathscr{C}$ such that \begin{equation*}
        \gamma_1\sim\gamma_2\Rightarrow(\gamma_1\parallel\gamma_2)\vee(\gamma_1\trianglelefteq\gamma_2)\vee(\gamma_2\trianglelefteq\gamma_1).
    \end{equation*}
\end{definition}

The partial ordering $\trianglelefteq$ defines a natural projection $\partial:\mathscr{E}\rightarrow\mathscr{P}$ in the following way: \begin{equation*}\label{def: positive-contours-of-Gamma}
    \partial\Gamma\coloneqq\{\gamma\in\Gamma:(\gamma'\in\Gamma)\wedge(\gamma\neq\gamma')\Rightarrow(\gamma'\trianglelefteq\gamma)\vee(\gamma\parallel\gamma')\}.
\end{equation*}Note that \begin{equation*}
    \I_-(\Gamma)\coloneqq\I_-\left(\sqcup_{\gamma\in\Gamma}\gamma\right)\subset\bigsqcup_{\gamma\in\partial\Gamma}\I_-(\gamma).
\end{equation*}Furthermore, $\partial\Gamma$ is the smallest subset of $\Gamma$ with this property. Equality ($\Gamma=\partial\Gamma$) holds if and only if $\Gamma\in\mathscr{P}$.

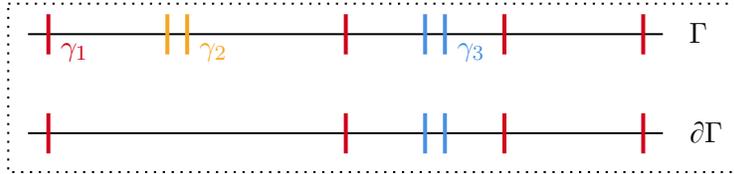
\begin{figure}[hbt!]
    \centering
    \tikzset{every picture/.style={line width=0.75pt}} 

\begin{tikzpicture}[x=0.75pt,y=0.75pt,yscale=-1,xscale=1]

\draw    (20,90.38) -- (340.17,90.38) ;
\draw [color={rgb, 255:red, 208; green, 2; blue, 27 }  ,draw opacity=1 ][line width=1.5]    (30.25,80.31) -- (30.25,100.6) ;
\draw [color={rgb, 255:red, 245; green, 166; blue, 35 }  ,draw opacity=1 ][line width=1.5]    (100.25,80.31) -- (100.25,100.6) ;
\draw [color={rgb, 255:red, 245; green, 166; blue, 35 }  ,draw opacity=1 ][line width=1.5]    (90.25,80.31) -- (90.25,100.6) ;
\draw [color={rgb, 255:red, 208; green, 2; blue, 27 }  ,draw opacity=1 ][line width=1.5]    (180.25,80.31) -- (180.25,100.6) ;
\draw [color={rgb, 255:red, 74; green, 144; blue, 226 }  ,draw opacity=1 ][line width=1.5]    (220.25,80.31) -- (220.25,100.6) ;
\draw [color={rgb, 255:red, 74; green, 144; blue, 226 }  ,draw opacity=1 ][line width=1.5]    (230.25,80.31) -- (230.25,100.6) ;
\draw [color={rgb, 255:red, 208; green, 2; blue, 27 }  ,draw opacity=1 ][line width=1.5]    (260.25,80.31) -- (260.25,100.6) ;
\draw [color={rgb, 255:red, 208; green, 2; blue, 27 }  ,draw opacity=1 ][line width=1.5]    (330.25,80.31) -- (330.25,100.6) ;
\draw    (20,140.37) -- (340.17,140.37) ;
\draw [color={rgb, 255:red, 208; green, 2; blue, 27 }  ,draw opacity=1 ][line width=1.5]    (30.25,130.3) -- (30.25,150.59) ;
\draw [color={rgb, 255:red, 208; green, 2; blue, 27 }  ,draw opacity=1 ][line width=1.5]    (180.25,130.3) -- (180.25,150.59) ;
\draw [color={rgb, 255:red, 74; green, 144; blue, 226 }  ,draw opacity=1 ][line width=1.5]    (220.25,130.3) -- (220.25,150.59) ;
\draw [color={rgb, 255:red, 74; green, 144; blue, 226 }  ,draw opacity=1 ][line width=1.5]    (230.25,130.3) -- (230.25,150.59) ;
\draw [color={rgb, 255:red, 208; green, 2; blue, 27 }  ,draw opacity=1 ][line width=1.5]    (260.25,130.3) -- (260.25,150.59) ;
\draw [color={rgb, 255:red, 208; green, 2; blue, 27 }  ,draw opacity=1 ][line width=1.5]    (330.25,130.3) -- (330.25,150.59) ;
\draw  [dash pattern={on 0.84pt off 2.51pt}] (10,75) -- (380,75) -- (380,160) -- (10,160) -- cycle ;

\draw (352,83.4) node [anchor=north west][inner sep=0.75pt]    {$\Gamma $};
\draw (352,133.4) node [anchor=north west][inner sep=0.75pt]    {$\partial \Gamma $};
\draw (35,93.4) node [anchor=north west][inner sep=0.75pt]  [color={rgb, 255:red, 208; green, 2; blue, 27 }  ,opacity=1 ]  {$\gamma _{1}$};
\draw (105,93.4) node [anchor=north west][inner sep=0.75pt]  [color={rgb, 255:red, 245; green, 166; blue, 35 }  ,opacity=1 ]  {$\gamma _{2}$};
\draw (235,93.4) node [anchor=north west][inner sep=0.75pt]  [color={rgb, 255:red, 74; green, 144; blue, 226 }  ,opacity=1 ]  {$\gamma _{3}$};

\end{tikzpicture}
    \caption{Take $\Gamma$ as in Figure~\ref{fig: positive-compatibility}, then $\partial\Gamma=\{\gamma_1,\gamma_3\}$.}
    \label{fig: boundary-of-a-partition}
\end{figure}

\begin{remark}
    Recalling Definition 2.7 in \cite{cluster}, $\partial\Gamma$ denotes the set of external contours of $\Gamma$.
\end{remark}

\begin{definition}\label{def: internal-collection}
    We say that $\Gamma\in\mathscr{E}$ is \textit{internal to }$\gamma_0\in\mathscr{C}$, if $\gamma\in\Gamma\Rightarrow(\gamma\sim\gamma_0)\wedge(\gamma\trianglelefteq\gamma_0)$, which we denote by $\Gamma\trianglelefteq\gamma_0$. We denote the collection of all $M$-partitions internal to $\gamma_0$ by \begin{equation*}\label{eq: internal-collection}
        \mathscr{I}(\gamma_0)\coloneqq\{\Gamma\in\mathscr{E}:\Gamma\trianglelefteq\gamma_0\}.
    \end{equation*}Since $\gamma_0\parallel\gamma_1\Rightarrow\mathscr{I}(\gamma_0)\cap\mathscr{I}(\gamma_1)=\emptyset$, this notion extends to $\Gamma\in\mathscr{P}$ by \begin{equation*}\label{eq: internal-collection-2}
        \mathscr{I}(\Gamma)\coloneqq\bigsqcup_{\gamma\in\Gamma}\mathscr{I}(\gamma).
    \end{equation*}
\end{definition}

\begin{remark}
    If $\Gamma\in\mathscr{P}$ and $\Gamma_1\in\mathscr{I}(\Gamma)$, then there is a unique $\gamma\in\Gamma$ such that $\Gamma_1\in\mathscr{I}(\gamma)$.
\end{remark}

Given $\gamma_0\in\mathscr{C}$, we define $\iota^*_{\gamma_0}:\mathscr{E}\rightarrow\mathscr{I}(\gamma_0)\sqcup\{\emptyset\}$ and $\iota_{\gamma_0}:\mathscr{E}\rightarrow\mathscr{E}$ by\begin{align*}
        &\iota^*_{\gamma_0}(\Gamma)\coloneqq\{\gamma\in\Gamma:(\gamma\sim\gamma_0)\wedge(\gamma\trianglelefteq\gamma_0)\}, &\iota_{\gamma_0}(\Gamma)\coloneqq\{\gamma_0\}\sqcup\iota_{\gamma_0}^*(\Gamma).
    \end{align*} These definitions naturally extend to positive collections $\Gamma_0\in\mathscr{P}$. We write \begin{align*}
        &\iota^*_{\Gamma_0}(\Gamma)\coloneqq\bigsqcup_{\gamma_0\in\Gamma_0}\iota^*_{\gamma_0}(\Gamma), &\iota_{\Gamma_0}(\Gamma)\coloneqq\bigsqcup_{\gamma_0\in\Gamma_0}\iota_{\gamma_0}(\Gamma).
    \end{align*}Finally, if $\Gamma\in\mathscr{E}$, then \begin{equation}\label{eq: relation-minus-interior-Gamma-interior-gamma}
        \I_-(\Gamma)=\bigsqcup_{\gamma\in\partial\Gamma}\I_-(\iota_\gamma(\Gamma))=\I_-(\partial\Gamma)\cap \bigcap_{\gamma\in\partial\Gamma}\I_-(\iota^*_\gamma(\Gamma))^c.
    \end{equation}

\begin{notation}\label{not: internal-collection}
    Henceforth, we write (recall Equation~\eqref{eq: definition-phi}) \begin{equation}
        \begin{split}
            &H_h(\iota_{\gamma_0}(\Gamma))\equiv H_{h,\gamma_0}(\Gamma),\\
            &H_h(\iota^*_{\gamma_0}(\Gamma))\equiv H^*_{h,\gamma_0}(\Gamma),\\
            &\Phi\left(\iota_{\gamma_1}(\Gamma),\iota_{\gamma_2}(\Gamma)\right)\equiv\Phi_{\gamma_1,\gamma_2}(\Gamma).
        \end{split}
    \end{equation}
    If $h\equiv0$, we omit it from the notation.
\end{notation}
    
\begin{lemma}\label{lem: energy-positive-partition-1}
    If $\Gamma\in\mathscr{E}$, then \begin{equation*}\label{eq: energy_positive-partition}
        H_h(\Gamma)=\sum_{\gamma\in\partial\Gamma}H_{\gamma,h}(\Gamma)-\hspace{-.25cm}\sum_{\substack{\gamma_1,\gamma_2\in\partial\Gamma\\\gamma_1\neq\gamma_2}}\frac{1}{2}\Phi_{\gamma_1,\gamma_2}(\Gamma).
    \end{equation*} In particular, if $\Gamma\in\mathscr{P}$, then \begin{equation*}\label{eq: energy-positive-partition-2}
        H_h(\Gamma)=\sum_{\gamma\in\Gamma}H_h(\gamma)-\hspace{-.25cm}\sum_{\substack{\gamma_1,\gamma_2\in\Gamma\\\gamma_1\neq\gamma_2}}\frac{1}{2}\Phi(\gamma_1,\gamma_2).
    \end{equation*}
\end{lemma}

\begin{proof}
    Let $\{A_1,...,A_n\}$ be a finite collection of non-empty disjoint sets. Then, \begin{align*}
        \left(\sqcup_{i=1}^nA_i\right)^c&=\cap_{i=1}^nA_i^c=A_i^c\setminus\left(\sqcup_{j\neq i}A_j\right),
    \end{align*} for all $i\in[n]$. Hence, \begin{equation}\label{eq: set-relation}
        \begin{split}
            \left(\sqcup_{i=1}^nA_i\right)\times\left[\left(\sqcup_{i=1}^nA_i\right)^c\right]&=\sqcup_{i=1}^n\left[A_i\times\left(A_i^c\setminus\sqcup_{j\neq i}A_j\right)\right]\\
        &=\left[\sqcup_{i=1}^n(A_i\times A_i^c)\right]\setminus\left[\sqcup_{i=1}^n\sqcup_{j\neq i}A_i\times A_j\right].
        \end{split}
    \end{equation} 
    Let $A_\gamma=\I_-(\iota_\gamma(\Gamma))$, where $\gamma\in\partial\Gamma$. Combining Equations~\eqref{eq: relation-minus-interior-Gamma-interior-gamma} and \eqref{eq: set-relation}, we obtain \begin{align*}
        H_h(\Gamma)=\hspace{-0.30cm}\sum_{\substack{x\in\I_-(\Gamma)\\y\in\I_-(\Gamma)^c}}\hspace{-0.20cm}2J_{xy}=\sum_{\gamma\in\partial\Gamma}\sum_{\substack{x\in\I_-(\iota_\gamma(\Gamma))\\y\in\I_-(\iota_\gamma(\Gamma))^c}}\hspace{-0.50cm}2J_{xy}-\sum_{\substack{\substack{\gamma,\gamma'\in\partial\Gamma\\\gamma\neq\gamma'}}}\sum_{\substack{x\in\I_-(\iota_\gamma(\Gamma))\\y\in\I_-(\iota_{\gamma'}(\Gamma))}}\hspace{-0.45cm}2J_{xy},
    \end{align*}The lemma has thus been proved.
\end{proof}

Finally, note that $\mathscr{C}$ is invariant under integer translations. If $\gamma_0\in\mathscr{C}$, we define \begin{align*}\label{eq: translation-invariance}
    &[\gamma_0]\coloneqq\{\gamma_0+k:k\in\mathbb{Z}\}, &\overline{\mathscr{C}}\coloneqq\{[\gamma_0]:\gamma_0\in\mathscr{C}\}.
\end{align*}

\subsection{Crucial estimates}\label{subsec: crucial-estimates} 

More important than the geometric objects themselves are the relevant estimates that they satisfy. We list the three that are key to proving the convergence of the cluster expansion. 

\begin{hypothesis}\label{hyp: M-choice}
    There exists $M=M(\alpha)>1$ such that \begin{equation*}
        H(\Gamma\setminus\gamma)+(7/8)H(\gamma)\leq H(\Gamma)\leq H(\Gamma\setminus\gamma)+H(\gamma),
    \end{equation*} for any $\Gamma\in\mathscr{E}$ and $\gamma\in\partial\Gamma$. Furthermore, there exists $\varepsilon>0$ such that $H(\gamma)\geq \varepsilon$ for all $\gamma\in\mathscr{C}$. In our case, we may take $\varepsilon=2$.
\end{hypothesis}

\begin{hypothesis}\label{hyp: diameter-hypothesis}
    There exist $c_0,c_1>0$, depending on $(M,\alpha)$, and a function $N:\mathscr{C}\rightarrow\mathbb{N}$, invariant under integer translations, such that \begin{equation*}\label{eq: cover-size}
        1+\log_2\mathrm{diam}(\gamma)\leq N(\gamma)\leq c_0(\log 2)^{-1}H(\gamma)
    \end{equation*} and, for all $n\in\mathbb{N}$, \begin{equation*}
        \lvert\{\gamma\in\mathscr{C}:(0\in\I_-(\gamma))\wedge(N(\gamma)=n)\}\rvert\leq 2^{c_1n}.
    \end{equation*}We note that    
    \begin{equation*}\label{eq: diameter-hypothesis}
    \gamma\in\mathscr{C}\Rightarrow \mathrm{diam}(\gamma)\leq e^{c_0 H(\gamma)}.
\end{equation*}
\end{hypothesis}

\begin{remark}
    In the one-dimensional case, one can choose the cover size introduced by Fr\"ohlich and Spencer in \cite{Frohlich.Spencer.82} as the function $N$.
\end{remark}

\begin{hypothesis}\label{hyp: phase-transition-estimate}
    There exist $\beta_0=\beta_0(M,\alpha)>0$ and $c_2=c_2(M,\alpha)>0$ such that \begin{equation*}\label{eq: pressure-upper-bound}
        \beta\geq\beta_0\Rightarrow\sum_{\gamma\in\mathscr{C}}\ind_{0\in\I_-(\gamma)}e^{-\beta H(\gamma)}\leq e^{-c_2\beta}.
    \end{equation*}
\end{hypothesis}

These hypotheses are \textit{the} crucial ingredients needed for the energy-entropy argument used to prove the existence of a first-order phase transition via contours. If $\alpha=2$, all were shown to hold in \cite{Frohlich.Spencer.82}. They were sufficient to obtain a convergent cluster expansion in this case; see \cite{Imbrie.82}. For the region $\alpha\in(1,2]$, Hypotheses~\ref{hyp: M-choice}, \ref{hyp: diameter-hypothesis} and \ref{hyp: phase-transition-estimate} were shown to hold in \cite{corsini1dpt}. We note that a Peierls' argument is possible with different types of energy estimates, i.e, modifications to Hypothesis~\ref{hyp: M-choice}, see for instance \cite{Bissacot-Kimura2018,Cassandro.05}. Unfortunately, in that case, a restriction on $\alpha$ appeared, see \cite{LittinPicco2017}.

Not only are these hypotheses sufficient to prove phase transition, they guaranty the absolute convergence of the cluster expansion. We show that this is the case following a similar approach to that of \cite{Cassandro.Merola.Picco.17,Cassandro.Merola.Picco.Rozikov.14}, which deals with a smaller region of polynomial decay $\alpha$. The same approach was successfully extended to the multidimensional case in \cite{cluster}.

Although Hypotheses~\ref{hyp: M-choice} and \ref{hyp: diameter-hypothesis} imply Hypothesis~\ref{hyp: phase-transition-estimate}, it is convenient to have it as a separate hypothesis. As a final observation before describing the appropriate polymer decomposition, we note that \begin{align}\label{eq: hypothesis-3-modified}
    \sum_{[\gamma\hspace{+0.025cm}]\in\overline{\mathscr{C}}}e^{-\beta H([\gamma\hspace{+0.025cm}])}\leq \sum_{\gamma\in\mathscr{C}}\ind_{0\in\I_-(\gamma)}e^{-\beta H(\gamma)}=\sum_{[\gamma\hspace{+0.025cm}]\in\overline{\mathscr{C}}}\lvert\I_-(\gamma)\rvert e^{-\beta H([\gamma\hspace{+0.025cm}])}.
\end{align}We will often use this alternative form of Hypothesis~\ref{hyp: phase-transition-estimate}.

\section{Polymers}\label{sec: polymers}

In this Section, we describe a polymer construction that is appropriate to the model at hand. We follow very closely the approach developed in \cite{cluster, Cassandro.Merola.Picco.17,Cassandro.Merola.Picco.Rozikov.14}: a polymer is a positive collection of contours, as in Definition~\ref{def: positive-partitions}. For this reason, we refer henceforth to positive collections of contours as \textit{polymers.} Nevertheless, a different type of compatibility must be introduced, for reasons that become clear in the proof of the Lemma~\ref{lem: polymer-gas-partition-function}. 

\begin{definition}\label{def: polymers}
    Let $\Gamma_1,\Gamma_2\in\mathscr{P}$. They are \textit{polymerally compatible}, which we denote by $\Gamma_1\sim\Gamma_2$, if one of the following mutually exclusive conditions is satisfied: \begin{enumerate}
        \item $\Gamma_1\parallel\Gamma_2$;
        \item $\Gamma_2\in\mathscr{I}(\Gamma_1)$, i.e., there exists $\gamma_1\in\Gamma_1\textrm{ such that }\Gamma_2\in\mathscr{I}(\gamma_1)$;
        \item $\Gamma_1\in\mathscr{I}(\Gamma_2)$.
    \end{enumerate}
    We say that $X\in\mathrm{P}(\mathscr{P})$ is a \textit{collection of compatible polymers} if \begin{equation*}
        (\Gamma_1,\Gamma_2\in X)\wedge(\Gamma_1\neq\Gamma_2)\Rightarrow\Gamma_1\sim\Gamma_2.
    \end{equation*} We denote the set of all collections of compatible polymers by $\mathscr{X}$. Similarly to Notation~\ref{not: sets-in-a-box}, if $\Lambda\Subset\mathbb{Z}$, we write \begin{equation*}
        \mathscr{X}_\Lambda\coloneqq\{X\in\mathscr{X}:\Gamma\in X\Rightarrow\Gamma\in\mathscr{P}_\Lambda\}.
    \end{equation*}
\end{definition}

\begin{remark}
    We note that if $\Gamma_1=\{\gamma_1\}$ and $\Gamma_2=\{\gamma_2\}$, then $\Gamma_1\sim\Gamma_2$ if and only if $\gamma_1\sim\gamma_2$.
\end{remark}

\begin{remark}
    For each $\Gamma\in\mathscr{E}$, there exists at least one decomposition $X$ of $\Gamma$ into compatible polymers. Indeed, just consider $X=\{\{\gamma\}:\gamma\in\Gamma\}$.
\end{remark}

\begin{convention}
    We establish $\emptyset\in\mathscr{X}_\Lambda$, for all $\Lambda\Subset\mathbb{Z}$. In particular, $\emptyset\in\mathscr{X}$.
\end{convention}

\begin{remark}\label{rem: characteristic-function-polymer}
    Let $\Gamma_1,\Gamma_2\in\mathscr{P}$. Then \begin{equation*}\label{eq: characteristic-function-polymer}
        \ind_{\Gamma_1\sim\Gamma_2}=\left(\ind_{\Gamma_1\parallel\Gamma_2}+\sum_{\gamma'_1\in\Gamma_1}\ind_{\Gamma_2\in\mathscr{I}(\gamma'_1)}+\sum_{\gamma'_2\in\Gamma_2}\ind_{\Gamma_1\in\mathscr{I}(\gamma'_2)}\right)\prod_{\substack{\gamma_1\in\Gamma_1\\\gamma_2\in\Gamma_2}}\ind_{\gamma_1\sim\gamma_2}.
    \end{equation*} Since all indicative functions inside parentheses are mutually exclusive, we have \begin{equation*}\label{eq: characteristic-function-not-polymer}
        \ind_{\Gamma_1\not\sim\Gamma_2}=\ind_{\left\{\substack{\exists\gamma_1\in\Gamma_1,\gamma_2\in\Gamma_2\\\textrm{s.t. }\gamma_1\not\sim\gamma_2}\right\}}+\prod_{\substack{\gamma_1,\gamma'_1\in\Gamma_1\\\gamma_2,\gamma'_2\in\Gamma_2}}\ind_{\gamma_1\sim\gamma_2}\ind_{\Gamma_1\not\parallel\Gamma_2}\ind_{\Gamma_2\not\in\mathscr{I}(\gamma'_1)}\ind_{\Gamma_1\not\in\mathscr{I}(\gamma'_2)}.
    \end{equation*}
\end{remark}

Since we may see collections of compatible polymers as partitions of elements of $\mathscr{E}$, the notion of compatible polymers implies that \begin{equation*}
    (X\in\mathscr{X})\wedge(Y\preceq X)\Rightarrow Y\in\mathscr{X}.
\end{equation*} That is, if $X$ is a collection of compatible polymers and $Y$ is finer than $X$ (if $A\in Y$, then there is $B\in X$ such that $A\subset B$), then $Y$ is also a collection of compatible polymers. 

There is also a subjacent natural order to collections of compatible polymers. More precisely, if $X\in\mathscr{X}$, it is reasonable to define \begin{equation*}
    \partial X\coloneqq\{\Gamma_1\in X:[(\Gamma_2\in X)\wedge(\I_-(\Gamma_1)\subset\I_-(\Gamma_2))]\Rightarrow\Gamma_1=\Gamma_2\}.
\end{equation*} In this case, we define recursively $X_1=\partial X$ and $X_{k+1}=\partial(X\setminus\sqcup_{i=1}^k X_i)$. It is easy to see that \begin{equation*}
    X=\bigsqcup_{k=1}^\infty X_k.
\end{equation*} This allows us to see $X$ as a (graph) forest whose set of vertices is $X$, and whose set of edges is \begin{align*}
    E(X)\coloneqq\bigsqcup_{k=1}^\infty\left\{\{\Gamma_k,\Gamma_{k+1}\}:(\exists\gamma_k\in\Gamma_k\in X_k:\Gamma_{k+1}\in\mathscr{I}(\gamma_k)\cap X_{k+1})\right\}.
\end{align*}In fact, we can attribute the "color" $\gamma_k$ to the edge $\{\Gamma_k,\Gamma_{k+1}\}$, if $\gamma_k\in\Gamma_k$ and $\Gamma_{k+1}\in\mathscr{I}(\gamma_k)$.

This very same representation hints at the existence of a single \textit{coarsest} collection of compatible polymers $X_\Gamma$ representing the same configuration $\Gamma\in\mathscr{E}$ as $X$. If the edges $\{\Gamma_k,\Gamma_{k+1,1}\}\neq\{\Gamma_k,\Gamma_{k+1,2}\}$ both have the color $\gamma_k\in\Gamma_k$, it must be true that $\Gamma_{k+1,1}\parallel\Gamma_{k+1,2}$, that is, $\Gamma_{k+1,1}\sqcup\Gamma_{k+1,2}$ is also a polymer. It is easy to see that \begin{equation*}
    X'=X\sqcup\{\Gamma_{k+1,1}\sqcup\Gamma_{k+1,2}\}\setminus\{\Gamma_{k+1,1},\Gamma_{k+1,2}\}
\end{equation*} is also a collection of compatible polymers. Furthermore, this change does not affect the "grading" described earlier.

One may repeat the procedure described in the previous paragraph as long as there remain two edges with the same color. At this point, since we have not yet made any changes to $\partial X$, we may not yet have obtained $X_\Gamma$. The last step is simply to combine all polymers in $\partial X$ into a single one. Thus, we obtain the coarsest collection of compatible polymers $X_\Gamma$ representing $\Gamma\in\mathscr{E}$. Note that the forest representation of $X_\Gamma$ is a tree.

As a matter of concreteness, let us give two examples.

\begin{figure}[hbt!]
    \centering
    \input{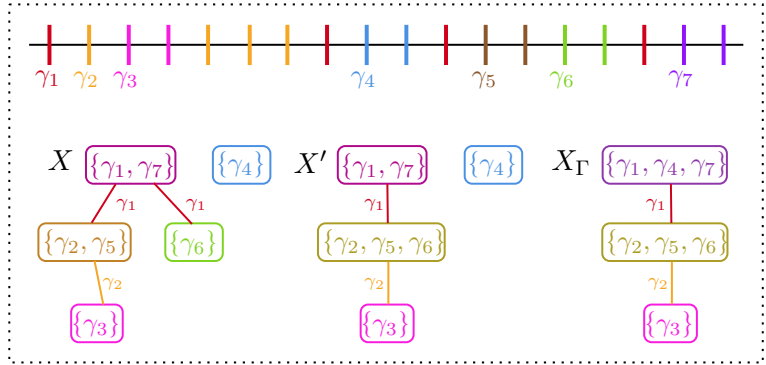}
    \caption{An example of the procedure whereof one obtains the coarsest decomposition of $\Gamma\in\mathscr{E}$ into compatible polymers.}
    \label{fig: forest-of-polymers}
\end{figure}

Let us explain the content of Figure~\ref{fig: forest-of-polymers}. Firstly, for those who are colorblind, let us give an alternative description of the contours. Ordering spin flips so that $i<j$ implies that $b_i$ lies to the left of $b_j$, we have\begin{align*}
        &\gamma_1=\{b_1,b_8,b_{11},b_{16}\} &\gamma_2=\{b_2,b_5,b_6,b_7\}\\
        &\gamma_3=\{b_3,b_4\} &\gamma_4=\{b_9,b_{10}\}\\
        &\gamma_5=\{b_{12},b_{13}\} &\gamma_6=\{b_{14},b_{15}\}\\
        &\gamma_7=\{b_{17},b_{18}\}.
    \end{align*}
    Let $\Gamma=\{\gamma_1,\gamma_2,\gamma_3,\gamma_4,\gamma_5,\gamma_6,\gamma_7\}$ and $X=\{\Gamma_1,\Gamma_2,\Gamma_3,\Gamma_4,\Gamma_5\}$, where \begin{align*}
        &\Gamma_1=\{\gamma_1,\gamma_7\}, &\Gamma_2=\{\gamma_2,\gamma_5\},\\
        &\Gamma_3=\{\gamma_3\}, &\Gamma_4=\{\gamma_4\},\\
        &\Gamma_5=\{\gamma_6\}.
    \end{align*}Then $X$ is a decomposition of $\Gamma$ into compatible polymers. We may then fuse $\Gamma_2=\{\gamma_2,\gamma_5\}$ and $\Gamma_5=\{\gamma_6\}$, and thus obtain $X'$, which is a decomposition of $\Gamma$ into compatible polymers, coarser than $X$. Finally, we may fuse $\Gamma_1=\{\gamma_1,\gamma_7\}$ and $\Gamma_4=\{\gamma_4\}$, and thus obtain $X_\Gamma$, the \textit{coarsest} decomposition of $\Gamma$ into compatible polymers.

\begin{figure}[hbt!]
    \centering
    \input{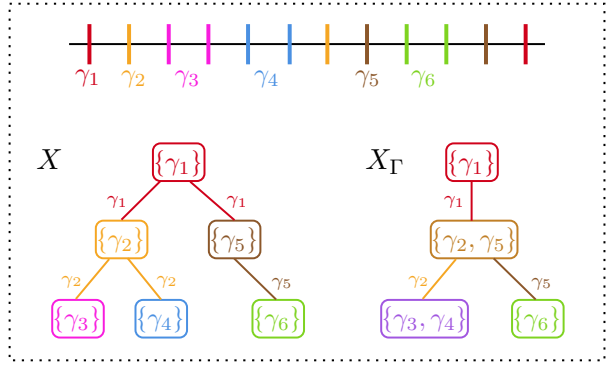}
    \caption{Another example of the procedure whereof one obtains the coarsest decomposition of $\Gamma\in\mathscr{E}$ into compatible polymers.}
    \label{fig: forest-of-polymers-2}
\end{figure}

Once again, let us explain the content of Figure~\eqref{fig: forest-of-polymers-2}. For those who are colorblind, let us give an alternative description of the contours. Ordering spin flips so that $i<j$ implies that $b_i$ lies to the left of $b_j$, we have \begin{align*}
    &\gamma_1 = \{b_1,b_{12}\},&\gamma_2=\{b_2,b_7\},\\
    &\gamma_3=\{b_3,b_4\},&\gamma_4=\{b_5,b_6\},\\
    &\gamma_5=\{b_8,b_{11}\},&\gamma_6=\{b_9,b_{10}\}.
\end{align*}Let $\Gamma=\{\gamma_1,\gamma_2,\gamma_3,\gamma_4,\gamma_5,\gamma_6\}$ and $X=\{\Gamma_i:i\in[6]\}$, where $\Gamma_i=\{\gamma_i\}$. We may fuse $\Gamma_3$ and $\Gamma_4$; similarly, we may fuse $\Gamma_2$ and $\Gamma_5$. Note that $\Gamma_6$ cannot be glued to $\Gamma'\coloneqq\Gamma_3\sqcup\Gamma_4$, for that would break the property of compatibility of the collection of polymers (recall Definition~\ref{def: polymers}). Thus, we obtain the coarsest decomposition of $\Gamma$.

\begin{lemma}
    Let $\Gamma\in\mathscr{E}$. Define \begin{equation*}
        \mathscr{X}_\Gamma\coloneqq\{X\in\Pt(\Gamma):X\preceq X_\Gamma\},
    \end{equation*} where $X_\Gamma$ denotes the coarsest decomposition of $\Gamma$ into compatible polymers. Then, \begin{equation*}
        \mathscr{X}=\{\emptyset\}\sqcup\bigsqcup_{\Gamma\in\mathscr{E}}\mathscr{X}_\Gamma.
    \end{equation*}Similarly, for any $\Lambda\Subset\mathbb{Z}$, \begin{equation*}
        \mathscr{X}_\Lambda=\{\emptyset\}\sqcup\bigsqcup_{\Gamma\in\mathscr{E}_\Lambda}\mathscr{X}_\Gamma.
    \end{equation*}
\end{lemma}

\subsection{A polymer gas associated to the long-range Ising model}\label{subsec: polymer-gas-decomposition}

In this subsection, we show that, for all $\Lambda\Subset\mathbb{Z}$, the partition function of the long-range ferromagnetic Ising model at inverse temperature $\beta>0$ and external field $h$ is identical to that of a hard-core polymer gas. This decomposition is similar to the one employed in \cite{cluster,Cassandro.Merola.Picco.17,Cassandro.Merola.Picco.Rozikov.14}. 

Let $\Gamma\in\mathscr{P}$, we define the activity of $\Gamma$ at temperature $\beta$ by \begin{equation}\label{eq: z-activity}
    z^+_{\beta,h}(\Gamma)\coloneqq\sum_{G\in\Gr(\Gamma)}\left\langle\prod_{e\in E(G)}\left[e^{\beta\Phi_e(\cdot)}-1\right]\right\rangle_{\Gamma,\beta,h}\frac{Z_{\beta,h}(\Gamma)}{Z^*_{\beta,h}(\Gamma)},
\end{equation}where (recall Definition~\ref{def: internal-collection} and Notation~\ref{not: internal-collection}) $\langle\cdot\rangle_{\Gamma,\beta,h}$ is the probability measure on $\mathscr{I}(\Gamma)$ such that if $\overline{\Gamma}\in\mathscr{I}(\Gamma)$, then
\[
    \left\langle\left\{\overline{\Gamma}\right\}\right\rangle_{\Gamma,\beta,h}\coloneqq \prod_{\gamma\in\Gamma}e^{-\beta H_{h,\gamma}\left(\overline{\Gamma}\right)}\cdot\left(Z_{\beta,h}(\Gamma)\right)^{-1}.
\]
Similarly, $\langle\cdot\rangle^*_{\Gamma,\beta,h}$ is the probability measure on $\mathscr{I}(\Gamma)$ such that if $\overline{\Gamma}\in\mathscr{I}(\Gamma)$, then
\[
    \left\langle\left\{\overline{\Gamma}\right\}\right\rangle^*_{\Gamma,\beta,h}\coloneqq \prod_{\gamma\in\Gamma}e^{-\beta H^*_{h,\gamma}\left(\overline{\Gamma}\right)}\cdot\left(Z^*_{\beta,h}(\Gamma)\right)^{-1}.
\]

\begin{lemma}\label{lem: polymer-gas-partition-function}
    Let $\Lambda\Subset\mathbb{Z}$. Then \begin{equation*}\label{eq: polymer-gas-partition-function}
        Z^+_{\beta,h,\Lambda}=\sum_{X\in\mathscr{X}_\Lambda}\prod_{\Gamma\in X}z^+_{\beta,h}(\Gamma)\eqqcolon\sum_{X\in\mathrm{P}(\mathscr{P})}z^+_{\beta,h}(X)\ind_{X\in\mathscr{X}_\Lambda}.
    \end{equation*}
\end{lemma}

\begin{proof}
Given any finite set $A$, we write \begin{equation*}
    \mathrm{P}_2(A)\coloneqq\{e\in\mathrm{P}(A):\lvert e\rvert=2\}.
\end{equation*}That is, we denote by $\mathrm{P}_2(A)$ the collection of all subsets of $A$ that contain exactly two elements. Let $\{x_e:e\in \mathrm{P}_2(A)\}\subset\mathbb{R}$, then (recall Notation~\ref{not: graphs-trees-partitions})\begin{align}\label{eq: mayer-trick}
    \prod_{e\in\mathrm{P}_2(A)}(1+x_e)=\sum_{P\in\Pt(A)}\prod_{Q\in P}\left[\sum_{G\in\Gr(Q)}\prod_{e\in E(G)}x_e\right].
\end{align}
Let $\Lambda\Subset\mathbb{Z}$. Then, \begin{align*}
    Z^+_{\Lambda,\beta,h}\hspace{-0.05cm}=\hspace{-0.05cm}\sum_{\sigma^+\in\Omega^+_\Lambda}\hspace{-0.05cm}e^{-\beta H_h^+(\sigma^+)}\hspace{-0.05cm}=\hspace{-0.05cm}1\hspace{-0.05cm}+\hspace{-0.05cm}\sum_{\overline{\Gamma}\in\mathscr{E}_\Lambda}\hspace{-0.05cm}e^{-\beta H_h(\overline{\Gamma})}\hspace{-0.05cm}=\hspace{-0.05cm}1\hspace{-0.05cm}+\hspace{-0.05cm}\sum_{\Gamma'_1\in\mathscr{P}_\Lambda}\hspace{-0.05cm}\sum_{\substack{\overline{\Gamma}_1\in\mathscr{E}\\\partial\overline{\Gamma}_1=\Gamma'_1}}\hspace{-0.05cm}e^{-\beta H_h(\overline{\Gamma}_1)}.
\end{align*}
By Lemma~\ref{lem: energy-positive-partition-1} and Mayer's trick (Equation~\eqref{eq: mayer-trick}), we have \begin{equation*}
    e^{-\beta H_h(\overline{\Gamma})}=\prod_{\gamma\in\partial\overline{\Gamma}}e^{-\beta H_{h,\gamma}(\overline{\Gamma})}\sum_{P\in\Pt(\partial\overline{\Gamma})}\prod_{Q\in P}\sum_{G\in\Gr(Q)}\prod_{e\in E(G)}\left(e^{\beta\Phi_e(\overline{\Gamma})}-1\right),
\end{equation*} for every $\Gamma\in\mathscr{E}$. Hence,\begin{align*}
    \sum_{\substack{\overline{\Gamma}_1\in\mathscr{E}\\\partial\overline{\Gamma}_1=\Gamma'_1}}e^{-\beta H_h(\overline{\Gamma}_1)}=\sum_{\substack{\overline{\Gamma}_1\in\mathscr{E}\\\partial\overline{\Gamma}_1=\Gamma'_1}}\prod_{\gamma_1\in\Gamma'_1}e^{-\beta H_{h,\gamma_1}(\overline{\Gamma}_1)}\sum_{X_1\in\Pt(\Gamma'_1)}\prod_{\Gamma_1\in X_1}\sum_{G_1\in\Gr(\Gamma_1)}\prod_{e_1\in E(G_1)}\left(e^{\beta\Phi_{e_1}(\overline{\Gamma}_1)}-1\right).\end{align*}
    
Let us comment on the chosen notation. We write $X_1$ in the hope that the reader identifies it with the same notation used earlier, where a graphical description of compatible collections of polymers is provided. 

Back to the decomposition algorithm: note that, for each $X_1\in\Pt(\Gamma'_1)$, there is a natural bijection between $\{\overline{\Gamma}_1\in\mathscr{E}:\partial\overline{\Gamma}_1=\Gamma'_1\}$ and $\times_{\Gamma_1\in X_1}\{\overline{\Gamma}(\Gamma_1)\in\mathscr{E}:\partial(\overline{\Gamma}(\Gamma_1))=\Gamma_1\}$; the canonical identification being \begin{equation*}
    (\overline\Gamma(\Gamma_1):\Gamma_1\in X_1)\mapsto\overline{\Gamma}_1=\sqcup_{\Gamma_1\in X_1}\overline{\Gamma}(\Gamma_1),
\end{equation*} whose inverse is given by \begin{equation*}
    \overline{\Gamma}_1\mapsto\left(\iota_{\Gamma_1}(\overline{\Gamma}_1):\Gamma_1\in X_1\right).
\end{equation*} This bijection leaves the values of $H_{h,\gamma_1}(\overline{\Gamma}_1)$ and $\Phi_{e_1}(\overline{\Gamma}_1)$ invariant. That is, if $\gamma_1\in\Gamma_1$, then \begin{align*}
    H_{h,\gamma_1}(\overline{\Gamma}_1)=H_{h,\gamma_1}(\overline{\Gamma}(\Gamma_1)).
\end{align*}Similarly, if $e_1\in\mathrm{P}_2(\Gamma_1)$, then \begin{align*}
    \Phi_{e_1}(\overline{\Gamma}_1)=\Phi_{e_1}(\overline{\Gamma}(\Gamma_1)).
\end{align*}
In this case, we arrive at  
\begin{equation}\label{eq: iteration-polymer-factorization}
    \begin{split}
        \sum_{\substack{\overline{\Gamma}_1\in\mathscr{E}\\\partial\overline{\Gamma}_1=\Gamma'_1}}\hspace{-0.10cm}e^{-\beta H_h(\overline{\Gamma}_1)}\hspace{-0.05cm}&=\hspace{-0.25cm}\sum_{X_1\in\Pt(\Gamma'_1)}\prod_{\Gamma_1\in X_1}\hspace{-0.10cm}\left(\sum_{G\in\Gr(\Gamma)}\left\langle\prod_{e\in E(G)}\left[e^{\beta\Phi_e(\cdot)}-1\right]\right\rangle
        \right)Z_{\beta,h}(\Gamma_1)\hspace{-0.05cm}\\
        &=\hspace{-0.25cm}\sum_{X_1\in\Pt(\Gamma'_1)}\prod_{\Gamma_1\in X_1}\hspace{-0.10cm}z^+_{\beta,h}(\Gamma_1 )Z^*_{\beta,h}(\Gamma_1).
    \end{split}
\end{equation}On the other hand, for each $\gamma_1\in\Gamma_1$, we have \begin{align*}
    Z^*_{\beta,h}(\gamma_1)&=\sum_{\substack{\overline{\Gamma}_1\in\mathscr{E}\\\partial\overline{\Gamma}_1=\{\gamma_1\}}}e^{-\beta H^*_{h,\gamma_1}(\overline{\Gamma}_1)}=1+\sum_{\substack{\Gamma'_2\in\mathscr{P}\\\Gamma'_2\in\mathscr{I}(\gamma_1)}}\sum_{\substack{\overline{\Gamma}_2\in\mathscr{E}\\\partial\overline{\Gamma}_2=\Gamma'_2}}e^{-\beta H_{h}(\overline{\Gamma}_2)}.
\end{align*}Since $\mathscr{I}(\gamma_1)$ might be empty, the sum can be $0$. This happens if, for example, $\mathrm{diam}(\gamma_1)\leq 2M$. If not, we proceed with the decomposition by applying Equation~\eqref{eq: iteration-polymer-factorization}. We have thus provided an algorithm that writes the initial partition function in terms of a gas of polymers.\end{proof}

\begin{remark}\label{rem: decomposition-algorithm}
    Firstly, note that\begin{equation*}
        \gamma_{k+1}\in\Gamma'_{k+1}\in\mathscr{I}(\gamma_k)\Rightarrow\mathrm{diam(\gamma_{k+1})}\leq \frac{1}{2M}\mathrm{diam}(\gamma_k)<\frac{1}{2}\mathrm{diam}(\gamma_k).
    \end{equation*} This means that this algorithm ends after at most $\log_2\mathrm{diam}(\Lambda)$ steps. In this case, what \textit{does} such an algorithm produce? Certainly, it cannot produce $z^+_{\beta,h}(X)$, lest $X\in\mathscr{X}_\Lambda$. That is, it does not produce a collection of polymers that are \textit{not} compatible. Hence, we obtain an expression of the type \begin{equation*}
        Z^+_{\Lambda,\beta,h}=1+\sum_{X\in\mathscr{X}_\Lambda}c_\Lambda(X)\prod_{\Gamma\in X}z^+_{\beta,h}(\Gamma).
    \end{equation*}Since $c_\Lambda(\emptyset)=0$, we may naturally identify "$1$" outside the summation with $z_{\beta,h}(\emptyset)$. 

Consider now $\emptyset\neq X\in\mathscr{X}_\Lambda$. Is $z_{\beta,h}(X)$ produced by the algorithm described earlier? How many times? The subjacent order to any collection of compatible polymers described after Remark~\ref{rem: characteristic-function-polymer} guaranties that should $z_{\beta,h}(X)$ appear, it does so only once. The fact that it \textit{does} appear is a consequence of the term "$+1$" in $Z_{\beta,h}^*(\gamma)$, which gives us an arbitrary stopping step as we run the decomposition algorithm described in the proof.
\end{remark}

\subsection{The formal representation of the pressure}

The representation of the model in terms of a gas of polymers is useful as long as the series computing the pressure of the system is absolutely convergent. That is, if the following formal series converges absolutely:
\begin{equation}\label{eq: pressure-polymer-gas}
    P_{\beta,h}(\Lambda)\coloneqq\frac{1}{\lvert\Lambda\rvert}\log Z^+_{\Lambda,\beta,h}=\frac{1}{\lvert\Lambda\rvert}\sum_{n=1}^\infty\frac{1}{n!}\sum_{\mathbf{\Gamma}\in\mathscr{P}^n_\Lambda}\phi_n(\mathbf{\Gamma})\prod_{i=1}^nz^+_{\beta, h}(\Gamma_i),
\end{equation} where $\phi_n:\mathscr{P}_\Lambda^n\rightarrow\mathbb{R}$ denotes a Ursell's function\begin{equation*}\label{eq: ursell_function}
    \phi_n(\mathbf{\Gamma})\coloneqq\sum_{G\in\Gr_n}\prod_{\{i,j\}\in E(G)}\left(-\ind_{\Gamma_i\not\sim\Gamma_j}\right)=\sum_{\substack{G\in\Gr_n\\E(G)\subset E(\mathbf{\Gamma})}}(-1)^{\lvert E(G)\rvert},
\end{equation*}and \begin{equation*}\label{eq: edges-vector-Gamma}
    E(\mathbf{\Gamma})=\{\{i,j\}\in\mathrm{P}_2([n]):\Gamma_i\not\sim\Gamma_j\}.
\end{equation*}We note that $\phi_n$ is invariant under the action of the symmetric group $\mathrm{Sym}(n)$, and that $\phi_n(\mathbf{\Gamma})=0$ if the graph with vertex set $[n]$ and edge set $E(\mathbf{\Gamma})$ is not connected. One may consult \cite{procacci2023cluster}, for example, for a derivation of Equation~\eqref{eq: pressure-polymer-gas}; see Section 3.1.2.

\subsection{Upper bounds to the activity} Although the decomposition derived in  Section~\ref{subsec: polymer-gas-decomposition} is exact, there are caveats. First, proving that the pressure is analytic implies that such decomposition is true in the case of a complex-valued external field. If $A\Subset\mathbb{Z}$, then for every $\Gamma\in\mathscr{P}$ there exists $\delta=\delta(\beta,A,\Gamma)>0$ such that 
\[
    \mathbb{D}_{\delta}^A\coloneqq\{h\in\mathbb{C}^A:\lVert h\rVert_\infty\leq \delta(\Gamma)\}\ni h\mapsto z^+_{\beta,h}(\Gamma)
\]
is holomorphic. Thus, if $\Lambda\Subset\mathbb{Z}$, then there exists $\delta=\delta(\beta,A,\Lambda)>0$ such that
\[
    \mathbb{D}^A_{\delta}\ni h\mapsto P_{\beta,h}(\Lambda)
\]
is holomorphic. However, it is necessary to prove that $\delta(\beta,A)=\inf_{\Lambda}\delta(\beta,A,\Lambda)>0$.

The second pertains to the \emph{ansatz} of $z^+_{\beta,h}$, it is quite complicated: it is much more reasonable to prove the absolute convergence of the cluster expansion by studying the convergence of the pressure of another hard-core gas of polymers with a simpler activity $\bar{z}_{\beta}(\cdot)$ that bounds $z^+_{\beta,h}(\cdot)$ uniformly (provided $\lVert h\rVert_\infty$ is sufficiently small).

We first deal with the matter of holomorphism. If $\Gamma\in\mathscr{P}$ and $F:\mathscr{I}(\Gamma)\rightarrow\mathbb{C}$, then
\[
    w\mapsto f_{\beta,h}(w)\coloneqq
    \left\langle e^{wF(\cdot)-\beta E_h(\cdot)}\right\rangle_{\Gamma,\beta,0}
\]
is an entire function such that $f_{\beta,h}(0)\neq 0$ provided $\lVert h\rVert_{\infty}$ is sufficiently small. Therefore, there exists $\epsilon>0$ such that 
\[
    \mathbb{D}_\epsilon\ni w\mapsto \log f_{\beta,h}(w)
\]
is holomorphic (we also consider the principal value logarithm). Note that
\[
    \frac{\mathrm{d}(\log f_{\beta,h})}{\mathrm{d}w}|_{w=0}=\left\langle F(\cdot)\right\rangle_{\Gamma,\beta,h}.
\]
Furthermore, the same argument works when replacing $\langle\cdot\rangle_{\Gamma,\beta,0}$ with $\langle\cdot\rangle^*_{\Gamma,\beta,0}$. 

There are types of such functions that we must define. Let $\Gamma\in\mathscr{P}$ and $G\in\Gr(\Gamma)$. Define
\[
    f_{G,\beta,h}(w)\coloneqq\left\langle\exp\left(w\cdot\prod_{e\in E(G)}\left[e^{\beta\Phi_e(\cdot)}-1\right]-\beta E_h(\cdot)\right)\right\rangle_{\Gamma,\beta,0}\eqqcolon \left\langle e^{wF_{G,\beta}(\cdot)-\beta E_h(\cdot)}\right\rangle_{\Gamma,\beta,0}.
\]
The second type arises from writing $\frac{Z_{\beta,h}(\Gamma)}{Z^*_{\beta,h}(\Gamma)}$ as an expectation with respect to $\langle\cdot\rangle^*_{\Gamma,\beta,h}$. On the one hand, if $\overline{\Gamma}\in\mathscr{I}(\gamma)$, then
\[
    E_h\left(\iota_\gamma^*\left(\overline{\Gamma}\right)\right)=E_h(\gamma)-E\left(\iota_\gamma\left(\overline{\Gamma}\right)\right).
\]
On the other hand, 
\[
    (7/8)H(\gamma)\leq \sum_{x,y}4J_{xy}\ind_{x\in\I_-(\iota^*_\gamma(\overline{\Gamma}))}\ind_{y\in\I_-(\gamma)^c}=H_\gamma\left(\overline{\Gamma}\right)-H^*_\gamma\left(\overline{\Gamma}\right)\leq H(\gamma).
\]
Thus,
\[
    \frac{Z_{\beta,h}(\Gamma)}{Z^*_{\beta,h}(\Gamma)}=\left\langle\prod_{\gamma\in\Gamma}e^{\beta H^*_\gamma(\cdot)-\beta H_\gamma(\cdot)+2\beta E_h(\iota^*_\gamma(\cdot))-\beta E_h(\gamma)}\right\rangle^*_{\Gamma,\beta,h}\eqqcolon\langle G_{\Gamma,\beta,h}(\cdot)\rangle^*_{\Gamma,\beta,h},
\]
which motives the definition of
\[
    g_{\Gamma,\beta, h}(w)\coloneqq\left\langle e^{w G_{\Gamma,\beta,h}(\cdot)-\beta E_h(\cdot\setminus\Gamma)}\right\rangle^*_{\Gamma,\beta,0}.
\]
Finally,
\begin{equation}\label{eq: activity-as-derivative-of-complex-logarithms}
    z^+_{\beta,h}(\Gamma)=\sum_{G\in\Gr(\Gamma)}\frac{\mathrm{d}\log f_{G,\beta,h}}{\mathrm{d}w}|_{w=0}\frac{\mathrm{d}\log g_{\Gamma,\beta,h}}{\mathrm{d}w}|_{w=0}.
\end{equation}

\begin{lemma}\label{lem: upper-activity}
    Let $A\Subset\mathbb{Z}$, $\Gamma\in\mathscr{P}$ and $\beta\geq (1/3)+2\log 288$. If $h\in\mathbb{D}^A_\delta$, where $\delta=(12\beta\cdot \# A)^{-1}$, then 
	\begin{equation}\label{eq: upper-activity-no-field}
		\left\lvert{z^+_{\beta,h}(\Gamma)}\right\rvert\leq\prod_{\gamma\in\Gamma}e^{-\frac{1}{2}\beta H(\gamma)}\sum_{T\in\Tr(\Gamma)}\prod_{e\in E(T)}\Phi(e)\eqqcolon\bar{z}_\beta(\Gamma).
	\end{equation}
\end{lemma}

\begin{proof}
    We first show that if $\lVert h\rVert_\infty\leq \delta$ and $\Gamma\in\mathscr{P}$, then $h\mapsto z^+_{\beta,h}(\Gamma)$ is holomorphic.
    
    By Equation~\eqref{eq: activity-as-derivative-of-complex-logarithms} and Cauchy's estimate for the derivative of holomorphic functions,
    \[
        \left\lvert{z^+_{\beta,h}(\Gamma)}\right\rvert\leq\sum_{G\in\Gr(\Gamma)}\frac{1}{R_1}\sup_{\lvert w\rvert\leq R_1}\left\lvert \log f_{G,\beta,h}(w)\right\rvert\cdot\frac{1}{R_2}\sup_{\lvert w\rvert\leq R_2}\lvert\log g_{\Gamma,\beta,h}(w)\rvert,
    \]
    where
    \[
        R_1\coloneqq \left(6\cdot\sup_{\overline{\Gamma}\in\mathscr{I}(\Gamma)}\left\lvert F_{G,\beta}\left(\overline{\Gamma}\right)\right\rvert\right)^{-1}=(6\lVert F_{G,\beta}\rVert)^{-1}
    \]
    and
    \[
        R_2\coloneqq\left(6\cdot \sup_{\overline{\Gamma}\in\mathscr{I}(\Gamma)}\left\lvert G_{\Gamma,\beta,h}\left(\overline{\Gamma}\right)\right\rvert\right)^{-1}=(6\lVert G_{\Gamma,\beta,h}\rVert)^{-1}.
    \]
    If $\lvert w\rvert\leq 1/2$, then $\lvert \log(1+w)\rvert\leq 2\lvert w\rvert$; and $\lvert e^w-1\lvert\leq \lvert w\rvert e^{\lvert w\rvert}$. These two bounds allows us to write
    \[
        \sup_{\lvert w\rvert\leq R_1}\left\lvert\log f_{G,\beta,h}(w)\right\rvert\leq2\cdot \sup_{\lvert w\rvert\leq R_1} (\lvert w\vert\lVert F_{G,\beta}\rVert+2\beta\lvert A\rvert\delta)e^{(\lvert w\vert\lVert F_{G,\beta}\rVert+2\beta\lvert A\rvert\delta)}\leq 2(2/6)e^{2/6}\leq 1.
    \]
    The same result holds for $\sup_{\lvert w\rvert\leq R_1}\left\lvert\log g_{\Gamma,\beta,h}(w)\right\rvert$. Thus,
    \[
         \left\lvert{z^+_{\beta,h}(\Gamma)}\right\rvert\leq\sum_{G\in\Gr(\Gamma)}36\cdot \lVert F_{G,\beta,h}\rVert\cdot \lVert G_{\Gamma,\beta,h}\lVert.
    \]
    We first bound $\lVert G_{\Gamma,\beta,h}\rVert$. By the choice of $\delta>0$, if $\overline{\Gamma}\in\mathscr{I}(\Gamma)$, then
    \[
        \left\lvert 2\beta E_h\left(\iota^*_\gamma\left(\overline{\Gamma}\right)\right)-\beta E_h(\gamma)\right\rvert\leq 2\beta\lvert A\rvert \delta\leq 1/6.
    \]
    Since $H_\gamma\left(\overline{\Gamma}\right)-H^*_\gamma\left(\overline{\Gamma}\right)\leq -(7/8)H(\gamma)$,
    \[
        \lVert G_{\Gamma,\beta,h}\rVert\leq e^{\frac{1}{6}\lvert\Gamma\rvert}\cdot\prod_{\gamma\in\Gamma}e^{-\frac{7}{8}\beta H(\gamma)}.
    \]
    We now bound $\sum_{G\in\Gr(\Gamma)}\lVert F_{G,\beta,h}\rVert$. If $\Gamma\in\mathscr{P}$, $\overline{\Gamma}\in\mathscr{I}(\Gamma)$ and $e\in\mathrm{P}_2(\Gamma)$, then $\Phi_e(\overline{\Gamma})\leq \Phi(e)$. Furthermore, every connected graph $G\in\Gr(\Gamma)$ contains at least one spanning tree $T\in\Tr(\Gamma)$, therefore, 
    \begin{equation*}
        \sum_{G\in\Gr(\Gamma)}\lVert F_{G,\beta,h}\rVert\leq\sum_{T\in\Tr(\Gamma)}\prod_{e\in E(T)}\left(e^{\beta\Phi(e)}-1\right)\sum_{\substack{G\in\Gr(\Gamma)\\E(T)\subset E(G)}}\prod_{e\in E(G)\setminus E(T)}\left(e^{\beta \Phi(e)}-1\right).
    \end{equation*}The non-negativity of $\Phi(e)$ implies that\begin{equation*}
        \sum_{\substack{G\in\Gr(\Gamma)\\E(T)\subset E(G)}}\hspace{-0.25cm}\prod_{e\in E(G)\setminus E(T)}\hspace{-0.50cm}(e^{\beta\Phi(e)}-1)\leq\hspace{-0.50cm}\sum_{\substack{E\subset\mathrm{P}_2(\Gamma)\\E\cap E(T)=\emptyset}}\hspace{-0.15cm}\prod_{e\in E}(e^{\beta\Phi(e)}-1)=\hspace{-0.50cm}\prod_{e\in\mathrm{P_2}(\Gamma)\setminus E(T)}\hspace{-0.60cm}e^{\beta\Phi(e)}.
    \end{equation*}On the other hand, since $x\geq 0\Rightarrow e^x-1\leq xe^x$, we have \begin{equation*}
        \prod_{e\in E(T)}\left(e^{\beta\Phi(e)}-1\right)\leq \left(\prod_{e\in E(T)}\hspace{-0.15cm}\beta\Phi(e)\right)\exp\left(\sum_{e\in E(T)}\beta\Phi(e)\right).
    \end{equation*}Hence \begin{equation*}
        \sum_{G\in\Gr(\Gamma)}\lVert F_{G,\beta,h}\rVert\leq \sum_{T\in\Tr(\Gamma)}\prod_{e\in E(T)}\beta\Phi(e)\cdot \prod_{e\in P_2(\Gamma)}e^{\beta\Phi(e)}.
    \end{equation*}By Lemma~\ref{lem: energy-positive-partition-1} and Hypothesis~\ref{hyp: M-choice}, \begin{align*}
    \sum_{e\in\mathrm{P}_2(\Gamma)}\Phi(e)=-H(\Gamma)+\sum_{\gamma\in\Gamma}H(\gamma)\leq\frac{1}{8}\sum_{\gamma\in\Gamma}H(\gamma).
\end{align*}
Since $\inf_{\gamma\in\mathscr{C}}H(\gamma)\geq 2$, $\lvert E(T)\rvert=\lvert \Gamma\rvert-1$ and $x\geq 0\Rightarrow x\leq e^x$, we have 
\begin{equation*}
    36\cdot e^{\frac{1}{6}\lvert\Gamma\rvert}\cdot\beta^{\lvert\Gamma\rvert-1}\leq (288\cdot e^{1/6})^{\lvert \Gamma\rvert}\cdot(\beta/8)^{\lvert\Gamma\rvert}\leq \prod_{\gamma\in\Gamma}e^{\frac{1}{4}\beta H(\gamma)},
\end{equation*}
where the inequality follows from the choice $\beta
\geq 1/3+2\log(288)>8$. 

Since $-(7/8)+1/4+1/8
=1/2$, the lemma follows. This concludes the proof.
\end{proof}

\section{Sums over trees}\label{sec: sums-over-trees}

This Section should be understood as a systematization and extension of the methods employed in \cite{cluster, Cassandro.Merola.Picco.17, Cassandro.Merola.Picco.Rozikov.14} to prove the absolute convergence of the cluster expansion. Although it is the traditional method of summing over leaves of trees, in the application to the long-range Ising model, we transform trees of polymers into trees of contours. This implies a change from local, as in for each edge of the tree, (in)compatibility conditions to global compatibility conditions. We can, under this hypothesis and another appropriate hypothesis (see Hypothesis~\ref{hyp: kth-contracting-function}), remove vertices of arbitrary degree and prove a rather useful result dealing with trees of contours (see Theorem~\ref{theo: m-point-bound}). We deal with trees whose vertices \textit{and} edges are assigned weights.

Let $\mathscr{V}$ be a countable set. Consider a family of functions $v_\beta:\mathscr{V}\rightarrow\mathbb{R}_{>0}$ and another (single) symmetric function $e:\mathscr{V}^2\rightarrow\mathbb{R}_{\geq 0}$, where $\beta>1$. We call $\mathfrak{V}\coloneqq\{v_\beta:\beta> 1\}$ the \textit{vertex function family} and $e(\cdot,\cdot)$ the \textit{edge weight function}. We define the set of compatible pairs of vertices by \begin{equation*}
    \mathscr{V}_2\coloneqq\{(\vartheta_1,\vartheta_2)\in\mathscr{V}^2:e(\vartheta_1,\vartheta_2)\neq 0\}.
\end{equation*}It is a symmetric set. We demand that $\pi_1\mathscr{V}_2=\mathscr{V}$, where $\pi_1$ denotes the first projection.

Since the function $\vartheta\mapsto \log v_{\beta/2}(\vartheta)-\log v_\beta(\vartheta)$ will appear often, let us give it a name. We define the \textit{modified vertex function family} $\mathfrak{V}'\coloneqq\{v'_\beta:\beta>1\}$, where $v'_\beta:\mathscr{V}\rightarrow\mathbb{R}$ is given by \begin{equation}\label{eq: modified-vertex-function}
    v'_\beta(\vartheta)\coloneqq\log v_{\beta/2}(\vartheta)-\log v_\beta(\vartheta).
\end{equation} For $k\geq 3,$ we also define \begin{equation*}\label{eq: global-vertex-compatibility}
    \mathscr{V}_k\coloneqq\{\mathbf{v}=(\vartheta_1,...,\vartheta_k)\in\mathscr{V}^k:i\neq j\in[k]\Rightarrow(\vartheta_i,\vartheta_j)\in \mathscr{V}_2\}.
\end{equation*}
We demand that the (modified) vertex function family $\mathfrak{V}'$ satisfies the following hypothesis:

\begin{hypothesis}\label{hyp: vertex-edge-function-hypotheses}
    The function $\beta\mapsto v'_\beta(\vartheta)$ is an increasing function for all $\vartheta\in\mathscr{V}$. Furthermore, \begin{equation*}
        \lim_{\beta\rightarrow\infty}\inf_{\vartheta\in\mathscr{V}}v'_\beta(\vartheta)=\infty.
    \end{equation*}
\end{hypothesis}

\begin{remark}\label{rem: VnVn}
    In general, $\mathscr{V}_n\subsetneq\mathscr{V}^n$. Nevertheless, it is natural to write $\mathscr{V}=\mathscr{V}^1=\mathscr{V}_1$. If $\mathcal{I}$ is a finite set of indices, we may naturally define $\mathscr{V}^\mathcal{I}$ and $\mathscr{V}_\mathcal{I}$, accordingly.
\end{remark}

\begin{remark}
    In our concrete applications, the modified vertex function $v'_\beta(\cdot)$ is proportional to $\beta H(\cdot)$, see Equations~\eqref{eq: v-prime-function-polymers} and \eqref{eq: modified-vertex-function-contours}.
\end{remark}

Before proceeding, let us fix some useful notation regarding trees. 

\begin{notation}\label{not: specific-vertex-sets}
    Let $V$ be a finite set and $T\in\Tr(V)$. Let $v\in V$, we define \begin{align*}
    &E_T(v)\coloneqq\{e\in E(T):v\in e\},\\
    &N_T(v)\coloneqq\{w\in V:\{v,w\}\in E(T)\}.
\end{align*}
\end{notation}

\begin{definition}\label{def: degree-of-a-vertex}
    Let $T\in\Tr(V)$ and $v\in V$, we define $\degr_T:V\rightarrow\mathbb{N}$ by \begin{equation*}\label{eq: degree-definition}
        \degr_T(v)\coloneqq\lvert E_T(v)\rvert=\lvert N_T(v)\rvert.
    \end{equation*}That is, we denote the function that maps a vertex $v$ to its \textit{degree} (in the tree $T$) by $\degr_T$.
\end{definition}

\begin{notation}\label{not: trees-with-set-of-leaves-fixed}
    Let $V$ be a finite set. We denote the set of leaves of $T\in\Tr(V)$ by \begin{equation*}
        L(T)\coloneqq\{v\in V:\degr_T(v)=1\}.
    \end{equation*}Let $L\subset V$, we define \begin{equation*}
        \Tr(V,L)\coloneqq\{T\in\Tr(V):L(T)\subset L\}
    \end{equation*}If $V=[n]$, we write $\Tr_n(L)$ instead of $\Tr([n],L)$. If $B=[n]$ and $A=[m]$, we write $\Tr_{n,m}$ instead of $\Tr_n([m])$. Finally, we define \begin{equation*}
        \overline{\Tr}_{n,m}\coloneqq \bigsqcup_{A\subset [m]\subset B\subset[n]}\Tr(B,A)\hspace{1cm}\textrm{and}\hspace{1cm}
        \overline{\Tr}_m\coloneqq\bigcup_{n:n\geq m}\overline{\Tr}_{n,m}.
    \end{equation*}Note that $m\leq n_1\leq n_2\Rightarrow\overline{\Tr}_{n_1,m}\subset\overline{\Tr}_{n_2,m}$ and $m_1\leq m_2\Rightarrow\overline{\Tr}_{m_1}\subset\overline{\Tr}_{m_2}$.
\end{notation}


Let $A\subset [n]$ and $T\in\Tr_n$. We denote by $T[A]$ the smallest sub-tree of $T$ whose set of vertices contains $A$. If $A=[m]$, we write $T[m]$ instead of $T[[m]]$. Note that \begin{equation}\label{eq: leaves-of-smallest-subtree}
    L(T[A])\subseteq A. 
\end{equation} Let $V\Subset\mathbb{N}$ and $T\in\Tr(V)$. Define $w_\beta[T]:\mathscr{V}^V\rightarrow\mathbb{R}$ by \begin{equation*}
    w_\beta[T](\mathbf{v})\coloneqq\prod_{i\in V}v_\beta(\vartheta_i)\hspace{-0.25cm}\prod_{\{j,k\}\in E(T)}\hspace{-0.4cm}e(\vartheta_j,\vartheta_k).
\end{equation*}Note that if $\{j,k\}\in E(T)$ and $\ind_{\mathscr{V}_2}(\vartheta_j,\vartheta_k)=0$, then $w_\beta[T](\mathbf{v})=0$.

Let $V_1$ and $V_2$ be finite sets and $\phi:V_1\rightarrow V_2$ be a bijection. Define $\phi_\star:\Tr(V_1)\rightarrow\Tr(V_2)$ by
\[
    E(\phi_\star T_1)\coloneqq\{\{\phi(v_1),\phi(w_1)\}:\{v_1,w_1\}\in E(T_1)\}.
\]
Note that $\phi_\star$ is a bijection. Similarly, define $\phi_\star:\mathscr{V}^{V_1}\rightarrow\mathscr{V}^{v_2}$ by
\[
    (\phi_\star\mathbf{v})_{i_2}\coloneqq\vartheta_{\phi^{-1}(i_2)}.
\]
Thus, if $T\in\Tr(V_1)$ and $\mathbf{v}\in\mathscr{V}^{V_1}$, then
\begin{equation}\label{eq: w-transforms-with-relabeling}
    w_\beta[\phi_\star T](\phi_\star\mathbf{v})=w_\beta[T](\mathbf{v}).
\end{equation}
Let $\mathbf{w}=(\varOmega_1,...,\varOmega_m)\in\mathscr{V}^m$. Let $T\in\Tr_n$ such that $m\leq n$, we define the \textit{weight of} $(T,\mathbf{w})$ by\begin{align*}
    \overline{\iota}_\beta[\mathbf{w}](T)\coloneqq\sum_{\varphi\in \mathrm{Inj}(m,n)}\frac{1}{n!}\sum_{\mathbf{v}\in\mathscr{V}^n}\prod_{l=1}^m\ind_{\vartheta_{\varphi(l)}=\varOmega_l}w_\beta(T,\mathbf{v}),
\end{align*} where $\mathrm{Inj}(m,n)$ denotes the set of all injections from $[m]$ to $[n]$.

We have just defined the usual weight seen in a cluster expansion. Let us now define another one, which takes into account the global compatibility condition announced earlier, which holds for the trees of contours, our ultimate interest. We define the \textit{modified weight of} $(T,\mathbf{w})$ by \begin{equation*}
    \iota_\beta[\mathbf{w}](T)\coloneqq\sum_{\varphi\in \mathrm{Inj}(m,n)}\frac{1}{n!}\sum_{\mathbf{v}\in\mathscr{V}_n}\prod_{l=1}^m\ind_{\vartheta_{\varphi(l)}=\varOmega_l}w_\beta(T,\mathbf{v}).
\end{equation*}Note the change in the sum from $\mathbf{v}\in\mathscr{V}^n$ to $\mathbf{v}\in\mathscr{V}_n$.

\begin{figure}[hbt!]
    \centering
    \tikzset{every picture/.style={line width=0.75pt}} 

\begin{tikzpicture}[x=0.75pt,y=0.75pt,yscale=-1,xscale=1]

\draw   (55,60) .. controls (55,54.48) and (59.48,50) .. (65,50) .. controls (70.52,50) and (75,54.48) .. (75,60) .. controls (75,65.52) and (70.52,70) .. (65,70) .. controls (59.48,70) and (55,65.52) .. (55,60) -- cycle ;
\draw   (30,90) .. controls (30,84.48) and (34.48,80) .. (40,80) .. controls (45.52,80) and (50,84.48) .. (50,90) .. controls (50,95.52) and (45.52,100) .. (40,100) .. controls (34.48,100) and (30,95.52) .. (30,90) -- cycle ;
\draw   (80,90) .. controls (80,84.48) and (84.48,80) .. (90,80) .. controls (95.52,80) and (100,84.48) .. (100,90) .. controls (100,95.52) and (95.52,100) .. (90,100) .. controls (84.48,100) and (80,95.52) .. (80,90) -- cycle ;
\draw   (55,120) .. controls (55,114.48) and (59.48,110) .. (65,110) .. controls (70.52,110) and (75,114.48) .. (75,120) .. controls (75,125.52) and (70.52,130) .. (65,130) .. controls (59.48,130) and (55,125.52) .. (55,120) -- cycle ;
\draw   (30,150) .. controls (30,144.48) and (34.48,140) .. (40,140) .. controls (45.52,140) and (50,144.48) .. (50,150) .. controls (50,155.52) and (45.52,160) .. (40,160) .. controls (34.48,160) and (30,155.52) .. (30,150) -- cycle ;
\draw   (80,150) .. controls (80,144.48) and (84.48,140) .. (90,140) .. controls (95.52,140) and (100,144.48) .. (100,150) .. controls (100,155.52) and (95.52,160) .. (90,160) .. controls (84.48,160) and (80,155.52) .. (80,150) -- cycle ;
\draw    (40,100) -- (65,110) ;
\draw    (65,130) -- (90,140) ;
\draw    (65,130) -- (40,140) ;
\draw    (65,70) -- (90,80) ;
\draw    (65,70) -- (40,80) ;
\draw    (120,90) .. controls (159.6,60.3) and (179.6,118.81) .. (218.81,90.87) ;
\draw [shift={(220,90)}, rotate = 143.13] [color={rgb, 255:red, 0; green, 0; blue, 0 }  ][line width=0.75]    (10.93,-3.29) .. controls (6.95,-1.4) and (3.31,-0.3) .. (0,0) .. controls (3.31,0.3) and (6.95,1.4) .. (10.93,3.29)   ;
\draw  [color={rgb, 255:red, 208; green, 2; blue, 27 }  ,draw opacity=1 ] (265.56,60.33) .. controls (265.56,54.81) and (270.04,50.33) .. (275.56,50.33) .. controls (281.08,50.33) and (285.56,54.81) .. (285.56,60.33) .. controls (285.56,65.86) and (281.08,70.33) .. (275.56,70.33) .. controls (270.04,70.33) and (265.56,65.86) .. (265.56,60.33) -- cycle ;
\draw  [color={rgb, 255:red, 208; green, 2; blue, 27 }  ,draw opacity=1 ] (240.56,90.33) .. controls (240.56,84.81) and (245.04,80.33) .. (250.56,80.33) .. controls (256.08,80.33) and (260.56,84.81) .. (260.56,90.33) .. controls (260.56,95.86) and (256.08,100.33) .. (250.56,100.33) .. controls (245.04,100.33) and (240.56,95.86) .. (240.56,90.33) -- cycle ;
\draw  [color={rgb, 255:red, 208; green, 2; blue, 27 }  ,draw opacity=1 ] (290.56,90.33) .. controls (290.56,84.81) and (295.04,80.33) .. (300.56,80.33) .. controls (306.08,80.33) and (310.56,84.81) .. (310.56,90.33) .. controls (310.56,95.86) and (306.08,100.33) .. (300.56,100.33) .. controls (295.04,100.33) and (290.56,95.86) .. (290.56,90.33) -- cycle ;
\draw  [color={rgb, 255:red, 208; green, 2; blue, 27 }  ,draw opacity=1 ] (265.56,120.33) .. controls (265.56,114.81) and (270.04,110.33) .. (275.56,110.33) .. controls (281.08,110.33) and (285.56,114.81) .. (285.56,120.33) .. controls (285.56,125.86) and (281.08,130.33) .. (275.56,130.33) .. controls (270.04,130.33) and (265.56,125.86) .. (265.56,120.33) -- cycle ;
\draw   (240.56,150.33) .. controls (240.56,144.81) and (245.04,140.33) .. (250.56,140.33) .. controls (256.08,140.33) and (260.56,144.81) .. (260.56,150.33) .. controls (260.56,155.86) and (256.08,160.33) .. (250.56,160.33) .. controls (245.04,160.33) and (240.56,155.86) .. (240.56,150.33) -- cycle ;
\draw   (290.56,150.33) .. controls (290.56,144.81) and (295.04,140.33) .. (300.56,140.33) .. controls (306.08,140.33) and (310.56,144.81) .. (310.56,150.33) .. controls (310.56,155.86) and (306.08,160.33) .. (300.56,160.33) .. controls (295.04,160.33) and (290.56,155.86) .. (290.56,150.33) -- cycle ;
\draw [color={rgb, 255:red, 208; green, 2; blue, 27 }  ,draw opacity=1 ]   (250.56,100.33) -- (275.56,110.33) ;
\draw    (275.56,130.33) -- (300.56,140.33) ;
\draw    (275.56,130.33) -- (250.56,140.33) ;
\draw [color={rgb, 255:red, 208; green, 2; blue, 27 }  ,draw opacity=1 ]   (275.56,70.33) -- (300.56,80.33) ;
\draw [color={rgb, 255:red, 208; green, 2; blue, 27 }  ,draw opacity=1 ]   (275.56,70.33) -- (250.56,80.33) ;
\draw  [dash pattern={on 0.84pt off 2.51pt}] (20,40) -- (320,40) -- (320,170) -- (20,170) -- cycle ;

\draw (134,106.4) node [anchor=north west][inner sep=0.75pt]  [font=\scriptsize]  {$\mathbf{w} =( \varOmega _{1} ,\varOmega _{2})$};
\draw (97,52.4) node [anchor=north west][inner sep=0.75pt]    {$T$};
\draw (60,53.4) node [anchor=north west][inner sep=0.75pt]  [font=\small]  {$1$};
\draw (85,83.4) node [anchor=north west][inner sep=0.75pt]  [font=\small]  {$3$};
\draw (35,83.4) node [anchor=north west][inner sep=0.75pt]  [font=\small]  {$2$};
\draw (60,113.4) node [anchor=north west][inner sep=0.75pt]  [font=\small]  {$4$};
\draw (35,143.4) node [anchor=north west][inner sep=0.75pt]  [font=\small]  {$5$};
\draw (85,143.4) node [anchor=north west][inner sep=0.75pt]  [font=\small]  {$6$};
\draw (147,124.4) node [anchor=north west][inner sep=0.75pt]  [font=\scriptsize]  {$ \begin{array}{l}
\varphi ( 1) =4\\
\varphi ( 2) =3
\end{array}$};
\draw (270.56,53.73) node [anchor=north west][inner sep=0.75pt]  [font=\small,color={rgb, 255:red, 208; green, 2; blue, 27 }  ,opacity=1 ]  {$1$};
\draw (245.56,83.73) node [anchor=north west][inner sep=0.75pt]  [font=\small,color={rgb, 255:red, 208; green, 2; blue, 27 }  ,opacity=1 ]  {$2$};
\draw (267.56,115.73) node [anchor=north west][inner sep=0.75pt]  [font=\scriptsize,color={rgb, 255:red, 208; green, 2; blue, 27 }  ,opacity=1 ]  {$\varOmega _{1}$};
\draw (245.56,143.73) node [anchor=north west][inner sep=0.75pt]  [font=\small]  {$5$};
\draw (295.56,143.73) node [anchor=north west][inner sep=0.75pt]  [font=\small]  {$6$};
\draw (292.56,85.23) node [anchor=north west][inner sep=0.75pt]  [font=\scriptsize,color={rgb, 255:red, 208; green, 2; blue, 27 }  ,opacity=1 ]  {$\varOmega _{2}$};
\draw (214,54.4) node [anchor=north west][inner sep=0.75pt]  [font=\footnotesize,color={rgb, 255:red, 208; green, 2; blue, 27 }  ,opacity=1 ]  {$T[ \mathrm{Im}\varphi ]$};

\end{tikzpicture}
    \caption{Visual description of the injection of vertices on a tree.}
    \label{fig: injecting-vertices-on-a-tree}
\end{figure}
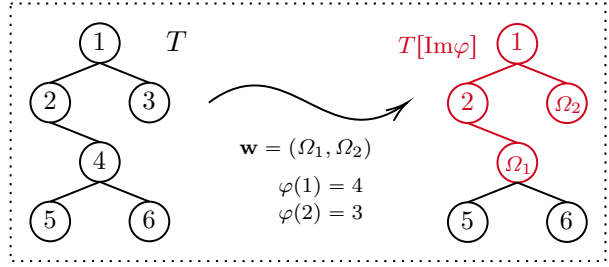

The term $(n!)^{-1}$ comes from the labeling introduced in the set of trees, and without which there is no hope of computing sums of trees. In which case, fixing the value of specific vertices can be described by injections. For a graphical description, see Figure~\ref{fig: injecting-vertices-on-a-tree}. 

\begin{lemma}\label{lem: mute-indices-cluster-weights}
    Let $\mathbf{w}\in\mathscr{V}^m$ and $m\leq n$. Then \begin{equation}\label{eq: bar-w-normal-w-mute-indices}
    \begin{split}
        \sum_{T\in\Tr_n}\overline{\iota}_\beta[\mathbf{w}](T)&=\sum_{T\in\Tr_n}\frac{1}{(n-m)!}\sum_{\mathbf{v}\in\mathscr{V}^n}\ind_{\mathbf{v}|_{[m]}=\mathbf{w}}w_\beta[T](\mathbf{v})\eqqcolon\sum_{T\in\Tr_n}\overline{\W}_\beta[\mathbf{w}](T),\\
        \sum_{T\in\Tr_n}\iota_\beta[\mathbf{w}](T)&=\sum_{T\in\Tr_n}\frac{1}{(n-m)!}\sum_{\mathbf{v}\in\mathscr{V}_n}\ind_{\mathbf{v}|_{[m]}=\mathbf{w}}w_\beta[T](\mathbf{v})\eqqcolon\sum_{T\in\Tr_n}\W_\beta[\mathbf{w}](T).
    \end{split}
    \end{equation}
\end{lemma}

\begin{proof}
    Let $\mathrm{Inj}(m,n)\ni\varphi\mapsto\tilde{\varphi}\in\mathrm{Sym}(n)$ such that $\tilde{\varphi}|_{[m]}=\varphi$. Hence,
    \[
        \begin{split}
            \sum_{T\in\Tr_n}\overline{\iota}_\beta[\mathbf{w}](T)&=\sum_{\varphi\in\mathrm{Inj}(m,n)}\frac{1}{n!}\sum_{T\in\Tr_n}\sum_{\mathbf{v}\in\mathscr{V}^n}\prod_{l\in[m]}\ind_{(\tilde{\varphi}^{-1}_\star\mathbf{v})_l=\varOmega_l}w_\beta[\tilde{\varphi}_\star(\tilde{\varphi}_\star^{-1}T)](\tilde{\varphi}_\star(\tilde{\varphi}_\star^{-1}\mathbf{v}))\\
            &=\frac{\lvert\mathrm{Inj}(m,n)\rvert}{n!}\sum_{T\in\Tr_n}\sum_{\mathbf{v}\in\mathscr{V}^n}\prod_{l\in[m]}\ind_{(\tilde{\varphi}^{-1}_\star\mathbf{v})_l=\varOmega_l}w_\beta[T](\mathbf{v}),
        \end{split}
    \]
    where the second equality from the change of variables $(T,\mathbf{v})\mapsto(\tilde\varphi^{-1}_\star T,\tilde{\varphi}_\star^{-1}\mathbf{v})$ and Equation~\eqref{eq: w-transforms-with-relabeling}. The same proof holds for $\sum_{T\in\Tr_n}\iota[\mathbf
    w](T)$ because $\mathscr{V}_n$ is invariant under index permutations. The lemma has thus been proved.
    
\end{proof}

Given $\mathbf{w}\in\mathscr{V}^m$, our objective is to estimate\begin{align}\label{eq: mute-indices-cluster-formula}
    \overline{\W}_\beta(\mathbf{w})\coloneqq\sum_{n=m}^\infty\overline{\W}_{\beta,n}(\mathbf{w})\coloneqq\sum_{n=m}^\infty\sum_{T\in\Tr_n}\overline{\W}_\beta[\mathbf{w}](T)
\end{align}
and
\[
    {\W}_\beta(\mathbf{w})\coloneqq\sum_{n=m}^\infty{\W}_{\beta,n}(\mathbf{w})\coloneqq\sum_{n=m}^\infty\sum_{T\in\Tr_n}{\W}_\beta[\mathbf{w}](T)
\]

\begin{definition}\label{def: families-of-sums}
    We call $\mathfrak{\overline{W}}\coloneqq\{\overline{\W}_\beta(\mathbf{w}):\beta>1, m\geq 1, \mathbf{w}\in\mathscr{V}^m\}$ the \textit{family of locally compatible sums} associated to $(\mathscr{V},\mathfrak{V},e)$.

    Similarly, we call $\mathfrak{W}\coloneqq\{\W_\beta(\mathbf{w}):\beta>1, m\geq 1, \mathbf{w}\in\mathscr{V}_m\}$ the \textit{family of globally compatible sums} associated to $(\mathscr{V},\mathfrak{V},e)$.
\end{definition}

\subsection{One-vertex estimates}\label{subsec: one-vertex-estimates}

Let us begin by tackling the case $m=1$. Let $\mathbf{w}=\varOmega\in\mathscr{V}$. Since we can see any $T\in\Tr_n$ as a tree rooted in $1$, there is a one-to-one correspondence between $T\in\Tr_n$ and \begin{enumerate}
    \item a number $t\in [n-1]$, that is, the generation of the rooted tree $T$;
    \item an \textit{ordered} partition $(P_1,...,P_t)$ of $[2,n]$;
    \item a family of functions $\varphi_p:P_p\rightarrow P_{p-1}$, with the convention that $P_0\equiv\{1\}$. That is, the edges of the tree $T$ are given by $\{k_p,\varphi_p(k_p)\}$, where $p\in [t]$ and $k_p\in P_p$.
\end{enumerate}
Hence, by separating the sum over $T\in\Tr_1$ and making the change of indices $n\mapsto n-1$ (this implies the convention $P_0\equiv\{0\}$ and $\vartheta_0\equiv\varOmega$), we obtain \begin{align*}
    \frac{\overline{\W}_\beta(\varOmega)}{v_\beta(\varOmega)}=1+\sum_{n=1}^\infty\frac{1}{n!}\sum_{t=1}^n\sum_{P_1,...,P_t}\sum_{\varphi_1,...,\varphi_t}\sum_{\mathbf{v}\in\mathscr{V}^n}\prod_{p=1}^t\prod_{k_p\in P_p}v_\beta(\vartheta_{k_p})e(\vartheta_{k_p},\vartheta_{\varphi_p(k_p)}).
\end{align*}
Since every index belonging to $[n]$ is mute, this sum is much easier to compute than it might look at first glance. One must simply count the number of ordered partitions of $[n]$ and note that \begin{equation*}
    \sum_{\varphi\in B^A}\prod_{a\in A}x_{a,\varphi(a)}=\prod_{a\in A}\sum_{b\in B} x_{a,b}.
\end{equation*} Thus, the series is equal to\begin{equation*}
    \sum_{n=1}^\infty\sum_{t=1}^n\sum_{\substack{n_1+...+n_t=n\\n_i\geq 1}}\prod_{p=1}^t\frac{1}{n_p!}\prod_{k_p=1}^{n_p}\sum_{l_{p-1}=1}^{n_{p-1}}\sum_{\vartheta_{p,k_p}}v_\beta(\vartheta_{p,k_p})e(\vartheta_{p,k_p},\vartheta_{p-1,l_{p-1}}).
\end{equation*} The first product should be understood as compute everything involving the index $p=t$, then everything involving the index $p=t-1$, etc. Exchanging the order of summation and summing over the superfluous index $n$, we obtain \begin{align*}\label{eq: bar-S-definition-leaf-op}
    \sum_{t=1}^\infty\prod_{p=1}^t\sum_{n_p=1}^\infty\frac{1}{n_p!}\prod_{k_p=1}^{n_p}\sum_{l_{p-1}=1}^{n_{p-1}}\sum_{\vartheta_{p,k_p}}v_\beta(\vartheta_{p,k_p})e(\vartheta_{p,k_p},\vartheta_{p-1,l_{p-1}})\eqqcolon\sum^\infty_{t=1}\rS_{\beta,t}(\varOmega),
\end{align*}with the conventions $n_0\equiv1$, and $\vartheta_{0,1}\equiv\varOmega$. 

Define the \textit{leaf pruning function family} $\mathfrak{R}^1=\{\mathcal{R}^1_\beta:\beta>1\}$, where $\mathcal{R}^1_\beta:\mathscr{V}\rightarrow\mathbb{R}_{\geq 0}$ is given by\begin{equation}\label{eq: leaf-pruning-operator}
    \mathcal{R}^1_\beta (\vartheta_1)\coloneqq\sum_{\vartheta\in\mathscr{V}}v_\beta
    (\vartheta)e(\vartheta,\vartheta_1).
\end{equation}Suppose that this function satisfies the following hypothesis.

\begin{hypothesis}\label{hyp: leaf-pruning-operator}
    The family function $\mathfrak{V}'$ dominates the family function $\mathfrak{R}^1$. That is,\begin{equation*}
        \lim_{\beta\rightarrow\infty}\sup_{\vartheta\in\mathscr{V}}\left\lvert\frac{\mathcal{R}^1_\beta(\vartheta)}{v'_\beta(\vartheta)}\right\rvert=0,
    \end{equation*}where we recall the definition of the modified vertex function present in Equation~\ref{eq: modified-vertex-function}.
\end{hypothesis}

\begin{remark}
    This is very similar to the convergence criterion given by Koteck\'y-Preiss in \cite{Koteck1986}. On the other hand, our computation is much closer to that of Fernand\'ez-Procacci in \cite{Fernndez2007}.
\end{remark}

\begin{lemma}\label{lem: reducing-sums-over-trees}
    For all $\eta>0$, there exists $\beta_0\coloneqq\beta_0(\eta)>0$ such that if $\beta\geq\beta_0$ and $\varOmega\in\mathscr{V}$, then \begin{equation}\label{eq: reducing-sums-over-trees}
    \overline{\W}_\beta(\varOmega)\leq \eta v_{\beta/2}(\varOmega).
\end{equation}
\end{lemma}

\begin{remark}
    The natural graphical representation of Lemma~\ref{lem: reducing-sums-over-trees} is the contraction of trees into a single vertex (with fixed value $\vartheta=\varOmega$).
\end{remark}

\begin{proof}
    By Hypothesis~\ref{hyp: vertex-edge-function-hypotheses}, we have  \begin{align*}
    \rS_{\beta,t}(\varOmega)\leq&\prod_{p=1}^{t-1}\sum_{n_p=1}^\infty\frac{1}{n_p!}\prod_{k_p=1}^{n_p}\left(\sum_{l_{p-1}=1}^{n_{p-1}}\sum_{\vartheta_{p,k_p}}v_\beta(\vartheta_{p,k_p})e(\vartheta_{p,k_p},\vartheta_{p-1,l_{p-1}})\right)\\&\times \sum_{n_t=1}^\infty\frac{1}{n_t!}\prod_{k_t=1}^{n_t}\sum_{l_{t-1}=1}^{n_{t-1}}\sum_{\vartheta_{t,k_t}}v_{\beta/2}(\vartheta_{t,k_t})e(\vartheta_{t,k_t},\vartheta_{t-1,l_{t-1}})\eqqcolon\bar{\rS}_{\beta,t}(\varOmega).
\end{align*} We prove the validity of Equation~\ref{eq: reducing-sums-over-trees} by showing that there exists $\beta_0>0$ such that if $\beta\geq \beta_0$, then\begin{align}\label{eq: S-bar-recurrence}
    \bar{\rS}_{\beta,t+1}(\varOmega)\leq 2^{-1}\bar{\rS}_{\beta,t}(\varOmega),
\end{align} for all $t\in\mathbb{N}$ and all $\varOmega\in\mathscr{V}$. In this case, \begin{align*}
    \frac{\overline{\W}_\beta(\varOmega)}{v_\beta(\varOmega)}\leq 1+\sum_{t=1}^\infty \bar{\rS}_{\beta,t}(\varOmega)\leq 1+\sum_{t=1}^\infty2^{-(t-1)}\bar{\rS}_{\beta,1}(\varOmega)\leq 2(1+\bar{\rS}_{\beta,1}(\varOmega)).
\end{align*}By the definition of $\bar{\rS}_{\beta,1}(\varOmega)$ and Equation~\eqref{eq: leaf-pruning-operator}, we have \begin{equation*}
\begin{split}
    1+\bar{\rS}_{\beta,1}(\varOmega)&=1+\sum_{n=1}^\infty\frac{1}{n!}\prod_{k=1}^n\sum_{\vartheta_k}v_{\beta/2}(\vartheta_k)e(\vartheta_k,\varOmega)=\sum_{n=0}^\infty\frac{1}{n!}(\mathcal{R}^1_{\beta/2}(\varOmega))^n=\exp\mathcal{R}^1_{\beta/2}(\varOmega).
    \end{split}
\end{equation*}By Hypothesis~\ref{hyp: leaf-pruning-operator}, if $\beta\geq\beta_0$ is sufficiently large, then, for every $\varOmega\in\mathscr{V}$, we have \begin{equation*}
    \mathcal{R}^1_{\beta/2}(\varOmega)\leq\frac{1}{2}v'_{\beta/2}(\varOmega)\leq\frac{1}{2} v'_\beta(\varOmega),
\end{equation*}where the second inequality is a consequence of Hypothesis~\ref{hyp: vertex-edge-function-hypotheses}. In this case, $\beta_0>0$ must be chosen so that \begin{equation*}
    2v_\beta(\varOmega)\exp(2^{-1} v'_\beta(\varOmega))\leq\eta v_{\beta/2}(\varOmega)\iff 4\eta^{-2}\leq\exp(v'_\beta(\varOmega))
\end{equation*}holds for $\beta\geq\beta_0$ and  all $\varOmega\in\mathscr{V}$. This can certainly be done by Hypothesis~\ref{hyp: vertex-edge-function-hypotheses}.

Let us now prove that if $\beta_0$ is sufficiently large, then $\bar{\rS}_{\beta,t+1}(\varOmega)\leq 2^{-1}\bar{\rS}_{\beta,t}(\varOmega)$ for all $t\geq 1$ and all $\varOmega\in\mathscr{V}$. By the definition of $\mathcal{R}^1_{\beta/2}(\cdot)$, we have 
\begin{equation*}
    \sum_{n_{t+1}=1}^{\infty}\frac{1}{n_{t+1}!}\prod_{k_{t+1}=1}^{n_{t+1}}\sum_{l_t=1}^{n_t}\sum_{\vartheta_{t+1,k_{t+1}}}v_{\beta/2}(\vartheta_{t+1,k_{t+1}})e(\vartheta_{t+1,k_{t+1}},\vartheta_{t,l_t})= e^{\sum_{l_t=1}^{n_t}\mathcal{R}^1_{\beta/2}(\vartheta_{t,l_t})}-1.
\end{equation*} Recalling that $x>0\Rightarrow(e^x-1<xe^x)\wedge(x<e^x)$, we arrive at 
\begin{equation}\label{eq: generation-eating-operation}
    \sum_{n_{t+1}=1}^{\infty}\frac{1}{n_{t+1}!}\prod_{k_{t+1}=1}^{n_{t+1}}\sum_{l_t=1}^{n_t}\sum_{\vartheta_{t+1,k_{t+1}}}v_{\beta/2}(\vartheta_{t+1,k_{t+1}})e(\vartheta_{t+1,k_{t+1}},\vartheta_{t,l_t})\leq \frac{1}{2}\prod_{l_t=1}^{n_t}e^{3\mathcal{R}^1_{\beta/2}(\vartheta_{t,l_t})}.
\end{equation}By Hypotheses~\ref{hyp: vertex-edge-function-hypotheses} and \ref{hyp: leaf-pruning-operator}, there exists $\beta_0>0$ such that if $\beta\geq\beta_0$, then \begin{align*}
    e^{3\mathcal{R}^1_{\beta/2}(\vartheta_1)}v_\beta(\vartheta_1)e(\vartheta_1,\vartheta_2)\leq v_{\beta/2}(\vartheta_1)e(\vartheta_1,\vartheta_2),
\end{align*}for all $\vartheta_1,\vartheta_2\in\mathscr{V}$. Combining this inequality with Equation~\ref{eq: generation-eating-operation}, we obtain Equation~\eqref{eq: S-bar-recurrence}. Thus, the lemma has been proved.
\end{proof}

\subsection{Two-vertex estimates}\label{subsec: two-vertex-estimates}

We now tackle the case $m=2$. Recalling Notation~\ref{not: trees-with-set-of-leaves-fixed} and Equation~\eqref{eq: leaves-of-smallest-subtree}, if $2\leq m\leq n$, we may define a  canonical application $\Psi_{n,m}:\Tr_n\mapsto\overline{\Tr}_{n,m}$ by
\begin{align*}
    \Psi_{n,m}(T)\coloneqq T[m].
\end{align*}
In particular, if $m=2$, then this application maps $\Tr_n$ to linear trees with leaves $\{1,2\}$. The matter is understanding how such an application interacts with the sums over trees.

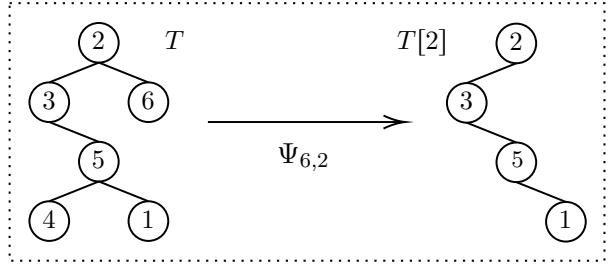
\begin{figure}[hbt!]
    \centering
    \tikzset{every picture/.style={line width=0.75pt}} 

\begin{tikzpicture}[x=0.75pt,y=0.75pt,yscale=-1,xscale=1]

\draw   (55,60) .. controls (55,54.48) and (59.48,50) .. (65,50) .. controls (70.52,50) and (75,54.48) .. (75,60) .. controls (75,65.52) and (70.52,70) .. (65,70) .. controls (59.48,70) and (55,65.52) .. (55,60) -- cycle ;
\draw   (30,90) .. controls (30,84.48) and (34.48,80) .. (40,80) .. controls (45.52,80) and (50,84.48) .. (50,90) .. controls (50,95.52) and (45.52,100) .. (40,100) .. controls (34.48,100) and (30,95.52) .. (30,90) -- cycle ;
\draw   (80,90) .. controls (80,84.48) and (84.48,80) .. (90,80) .. controls (95.52,80) and (100,84.48) .. (100,90) .. controls (100,95.52) and (95.52,100) .. (90,100) .. controls (84.48,100) and (80,95.52) .. (80,90) -- cycle ;
\draw   (55,120) .. controls (55,114.48) and (59.48,110) .. (65,110) .. controls (70.52,110) and (75,114.48) .. (75,120) .. controls (75,125.52) and (70.52,130) .. (65,130) .. controls (59.48,130) and (55,125.52) .. (55,120) -- cycle ;
\draw   (30,150) .. controls (30,144.48) and (34.48,140) .. (40,140) .. controls (45.52,140) and (50,144.48) .. (50,150) .. controls (50,155.52) and (45.52,160) .. (40,160) .. controls (34.48,160) and (30,155.52) .. (30,150) -- cycle ;
\draw   (80,150) .. controls (80,144.48) and (84.48,140) .. (90,140) .. controls (95.52,140) and (100,144.48) .. (100,150) .. controls (100,155.52) and (95.52,160) .. (90,160) .. controls (84.48,160) and (80,155.52) .. (80,150) -- cycle ;
\draw    (40,100) -- (65,110) ;
\draw    (65,130) -- (90,140) ;
\draw    (65,130) -- (40,140) ;
\draw    (65,70) -- (90,80) ;
\draw    (65,70) -- (40,80) ;
\draw  [color={rgb, 255:red, 0; green, 0; blue, 0 }  ,draw opacity=1 ] (265.56,60.33) .. controls (265.56,54.81) and (270.04,50.33) .. (275.56,50.33) .. controls (281.08,50.33) and (285.56,54.81) .. (285.56,60.33) .. controls (285.56,65.86) and (281.08,70.33) .. (275.56,70.33) .. controls (270.04,70.33) and (265.56,65.86) .. (265.56,60.33) -- cycle ;
\draw  [color={rgb, 255:red, 0; green, 0; blue, 0 }  ,draw opacity=1 ] (240.56,90.33) .. controls (240.56,84.81) and (245.04,80.33) .. (250.56,80.33) .. controls (256.08,80.33) and (260.56,84.81) .. (260.56,90.33) .. controls (260.56,95.86) and (256.08,100.33) .. (250.56,100.33) .. controls (245.04,100.33) and (240.56,95.86) .. (240.56,90.33) -- cycle ;
\draw  [color={rgb, 255:red, 0; green, 0; blue, 0 }  ,draw opacity=1 ] (265.56,120.33) .. controls (265.56,114.81) and (270.04,110.33) .. (275.56,110.33) .. controls (281.08,110.33) and (285.56,114.81) .. (285.56,120.33) .. controls (285.56,125.86) and (281.08,130.33) .. (275.56,130.33) .. controls (270.04,130.33) and (265.56,125.86) .. (265.56,120.33) -- cycle ;
\draw   (290.56,150.33) .. controls (290.56,144.81) and (295.04,140.33) .. (300.56,140.33) .. controls (306.08,140.33) and (310.56,144.81) .. (310.56,150.33) .. controls (310.56,155.86) and (306.08,160.33) .. (300.56,160.33) .. controls (295.04,160.33) and (290.56,155.86) .. (290.56,150.33) -- cycle ;
\draw [color={rgb, 255:red, 0; green, 0; blue, 0 }  ,draw opacity=1 ]   (250.56,100.33) -- (275.56,110.33) ;
\draw    (275.56,130.33) -- (300.56,140.33) ;
\draw [color={rgb, 255:red, 0; green, 0; blue, 0 }  ,draw opacity=1 ]   (275.56,70.33) -- (250.56,80.33) ;
\draw  [dash pattern={on 0.84pt off 2.51pt}] (20,40) -- (320,40) -- (320,170) -- (20,170) -- cycle ;
\draw    (120,100) -- (193,100) -- (217.5,100) ;
\draw [shift={(219.5,100)}, rotate = 180] [color={rgb, 255:red, 0; green, 0; blue, 0 }  ][line width=0.75]    (10.93,-3.29) .. controls (6.95,-1.4) and (3.31,-0.3) .. (0,0) .. controls (3.31,0.3) and (6.95,1.4) .. (10.93,3.29)   ;

\draw (97,52.4) node [anchor=north west][inner sep=0.75pt]  [font=\small]  {$T$};
\draw (60,53.4) node [anchor=north west][inner sep=0.75pt]  [font=\small]  {$2$};
\draw (85,83.4) node [anchor=north west][inner sep=0.75pt]  [font=\small]  {$6$};
\draw (35,83.4) node [anchor=north west][inner sep=0.75pt]  [font=\small]  {$3$};
\draw (60,113.4) node [anchor=north west][inner sep=0.75pt]  [font=\small]  {$5$};
\draw (35,143.4) node [anchor=north west][inner sep=0.75pt]  [font=\small]  {$4$};
\draw (85,143.4) node [anchor=north west][inner sep=0.75pt]  [font=\small]  {$1$};
\draw (153.5,108.9) node [anchor=north west][inner sep=0.75pt]  [font=\normalsize]  {$\Psi _{6,2}$};
\draw (270.56,53.73) node [anchor=north west][inner sep=0.75pt]  [font=\small,color={rgb, 255:red, 0; green, 0; blue, 0 }  ,opacity=1 ]  {$2$};
\draw (245.56,83.73) node [anchor=north west][inner sep=0.75pt]  [font=\small,color={rgb, 255:red, 0; green, 0; blue, 0 }  ,opacity=1 ]  {$3$};
\draw (271.56,113.73) node [anchor=north west][inner sep=0.75pt]  [font=\footnotesize,color={rgb, 255:red, 0; green, 0; blue, 0 }  ,opacity=1 ]  {$5$};
\draw (295.56,143.73) node [anchor=north west][inner sep=0.75pt]  [font=\small]  {$1$};
\draw (214,51.4) node [anchor=north west][inner sep=0.75pt]  [font=\small,color={rgb, 255:red, 0; green, 0; blue, 0 }  ,opacity=1 ]  {$T[ 2]$};

\end{tikzpicture}
    \caption{Note that there might not exist $n'\in\mathbb{N}$ such that $\Psi_{n,2}(T)\in\Tr_n$.}
    \label{fig: Psi-n-m-example}
\end{figure}

We define the \textit{total edge function family} $\mathfrak{\overline{E}}=\{\overline{\E}_\beta:\beta>1\}$, where $\overline{\E}_\beta:\mathscr{V}^2\rightarrow\mathbb{R}_{\geq0}$ is given by\begin{equation}\label{eq: total-edge-function}
    \overline{\E}_\beta(\varOmega_1,\varOmega_2)\coloneqq\sum_{n=2}^\infty\sum_{\mathbf{v}\in\mathscr{V}^n}\ind_{\vartheta_1=\varOmega_1}\ind_{\vartheta_n=\varOmega_2}\prod_{i=2}^{n-1}\overline{\W}_\beta(\vartheta_i)\prod_{j=1}^{n-1}e(\vartheta_j,\vartheta_{j+1}).
\end{equation}It is also useful to define \begin{equation*}
    \overline{\E}_{\beta,n}(\varOmega_1,\varOmega_2)\coloneqq\hspace{-0.25cm}\sum_{\mathbf{v}\in\mathscr{V}^{n+2}}\hspace{-0.25cm}\ind_{\vartheta_1=\varOmega_1}\ind_{\vartheta_{n+2}=\varOmega_2}\prod_{i=2}^{n+1}\overline{\W}_\beta(\vartheta_i)\prod_{j=1}^{n+1}e(\vartheta_j,\vartheta_{j+1}),
\end{equation*}in which case, we have \begin{equation*}
    \overline{\E}_{\beta}(\varOmega_1,\varOmega_2)\coloneqq e(\varOmega_1,\varOmega_2)+\sum_{n=1}^\infty\overline{\E}_{\beta,n}(\varOmega_1,\varOmega_2).
\end{equation*}Note that $e(\varOmega_1,\varOmega_2)=0$ does \textit{not} imply that $\overline{\E}_\beta(\varOmega_1,\varOmega_2)=0$.

We also define the \textit{modified total edge function family}  $\mathfrak{E}\coloneqq\{\E_\beta:\beta>1\}$, where $\E_\beta:\mathscr{V}^2\rightarrow\mathbb{R}_{\geq0}$ is given by\begin{equation*}
    \E_\beta(\varOmega_1,\varOmega_2)\coloneqq\sum_{n=2}^\infty\sum_{\mathbf{v}\in\mathscr{V}_n}\ind_{\vartheta_1=\varOmega_1}\ind_{\vartheta_n=\varOmega_2}\prod_{i=2}^{n-1}\W_\beta(\vartheta_i)\prod_{j=1}^{n-1}e(\vartheta_j,\vartheta_{j+1}).
\end{equation*}Once again, it is useful to define \begin{equation*}
    \E_{\beta,n}(\varOmega_1,\varOmega_2)\coloneqq\hspace{-0.25cm}\sum_{\mathbf{v}\in\mathscr{V}_{n+2}}\hspace{-0.25cm}\ind_{\vartheta_1=\varOmega_1}\ind_{\vartheta_{n+2}=\varOmega_2}\prod_{i=2}^{n+1}\W_\beta(\vartheta_i)\prod_{j=1}^{n+1}e(\vartheta_j,\vartheta_{j+1}),
\end{equation*}in which case, we have \begin{equation*}
    \E_{\beta}(\varOmega_1,\varOmega_2)\coloneqq e(\varOmega_1,\varOmega_2)+\sum_{n=1}^\infty\E_{\beta,n}(\varOmega_1,\varOmega_2).
\end{equation*}Note that, in this case, $e(\varOmega_1,\varOmega_2)=0$ does imply that $\E_\beta(\varOmega_1,\varOmega_2)=0$.

\begin{lemma}\label{lem: two-point-arbitrary sums}
    Let $\mathbf{w}=(\varOmega_1,\varOmega_2)\in\mathscr{V}^2$. Then, \begin{equation*}
        \begin{split}
            \overline{\W}_\beta(\mathbf{w})&=\overline{\W}_\beta(\varOmega_1)\overline{\E}_\beta(\varOmega_1,\varOmega_2)\overline{\W}_\beta(\varOmega_2),\\
            \W_\beta(\mathbf{w})&\leq\W_\beta(\varOmega_1)\mathrm{E}_\beta(\varOmega_1,\varOmega_2)\W_\beta(\varOmega_2).
        \end{split}
    \end{equation*}
\end{lemma}

\begin{remark}
    The inequality in the case of global compatibility conditions is natural: it is not necessary that gluing globally compatible trees generates another globally compatible tree.
\end{remark}

\begin{proof}
    By Equation~\eqref{eq: mute-indices-cluster-formula}, we have \begin{equation*}
    \overline{\W}_\beta(\mathbf{w})=\sum_{n=2}^\infty\frac{1}{(n-2)!}\sum_{T\in\Tr_n}\sum_{\mathbf{v}\in\mathscr{V}^n}\ind_{\mathbf{v}|_{[2]}=\mathbf{w}}w_\beta(T,\mathbf{v}).
\end{equation*}Introducing the application $\Psi_{n,2}$, we obtain \begin{equation*}
    \overline{\W}_\beta(\mathbf{w})=\sum_{n=2}^\infty\frac{1}{(n-2)!}\sum_{T\in\Tr_n}\sum_{\underline{T}\in\overline{\Tr}_{n,2}}\ind_{\Psi_{n,2}(T)=\underline{T}}\sum_{\mathbf{v}\in\mathscr{V}^n}\ind_{\mathbf{v}|_{[2]}=\mathbf{w}}w_\beta(T,\mathbf{v}).
\end{equation*}Inverting the order of the sums and recalling that $\Psi_{n,2}(T)=T[2]$, we have \begin{equation*}
    \overline{\W}_\beta(\mathbf{w})=\sum_{\underline{T}\in\overline{\Tr}_2}\sum_{n=2}^\infty\frac{1}{(n-2)!}\sum_{T\in\Tr_n}\ind_{T[2]=\underline{T}}\sum_{\mathbf{v}\in\mathscr{V}^n}\ind_{\mathbf{v}|_{[2]}=\mathbf{w}}w_\beta(T,\mathbf{v}).
\end{equation*}We note that $\overline{\Tr}_2$ is in one-to-one correspondence with pairs $(\underline{n},\underline{\varphi})$, where $\underline{n}\geq 2$ is a natural number and $\underline{\varphi}$ is an injection from $[\underline{n}]$ to $\mathbb{N}$ such that $\underline{\varphi}(1)=1$ and $\underline{\varphi}(\underline{n})=2$. This bijection is given by \begin{enumerate}
    \item $V(\underline{T})=\underline{\varphi}([\underline{n}])$,
    \item $E(\underline{T})=\{\{\underline{\varphi}(i),\underline{\varphi}(i+1)\}:i\in [\underline{n}-1]\}$.
\end{enumerate}We also note that if $T[2]=\underline{T}$, then $\textrm{Im}(\underline{\varphi})\subset [n]$. That is, if $\underline{N}\coloneqq \max\textrm{Im}(\underline{\varphi})\leq n$.

\begin{figure}[hbt!]
    \centering
    \input{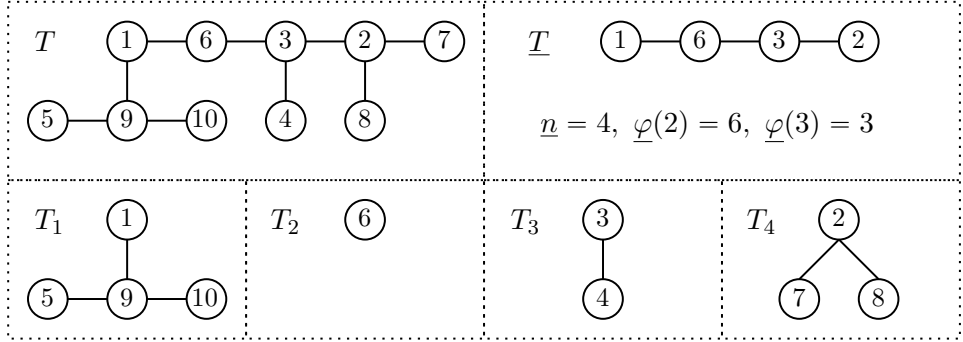}
    \caption{An example of the correspondences between $\underline{T}\in\overline{\Tr}_2$ and $(\underline{n},\underline{\varphi})$, and between $T\in\Psi_{n,2}^{-1}(\underline{T})$ and $(\mathbf{P,\mathbf{T}})$.}
    \label{fig: T-underlineT-correspondence}
\end{figure}

For each $\underline{T}\in\overline{\Tr}_2$, provided $n\geq\underline{N}$, the set $\Psi_{n,2}^{-1}(\underline{T})$ is in a one-to-one correspondence with \begin{enumerate}
    \item an \textit{ordered} partition $\mathbf{P}=(P_1,...,P_{\underline{n}})$ of $[n]$ such that $i\in[\underline{n}]\Rightarrow\underline{\varphi}(i)\in P_i$;
    \item a vector of trees $\mathbf{T}=(T_1,...,T_{\underline{n}})$, such that $T_i\in\Tr(P_i)$.
\end{enumerate}The correspondence is given by \begin{equation*}
    E(T)=E(\underline{T})\sqcup\bigsqcup_{i=1}^{\underline{n}}E(T_i).
\end{equation*}In this case, if $T\in\Psi_{n,2}^{-1}(\underline{T})$, then there exists $(\mathbf{P},\mathbf{T})$ such that \begin{equation}\label{eq: weight-usual-cluster-2}
    w_\beta(T,\mathbf{v})=\prod_{i=1}^{\underline{n}}w_\beta(T_i,\mathbf{v}|_{P_i})\prod_{j=1}^{\underline{n}-1}e(\vartheta_{\underline{\varphi}(j)},\vartheta_{\underline{\varphi}(j+1)}).
\end{equation}
By Equation~\eqref{eq: weight-usual-cluster-2}, we have \begin{equation*}
    \overline{\W}_\beta(\mathbf{w})=\sum_{\underline{n}=2}^\infty\sum_{\underline{\varphi}}\sum_{n=\underline{N}}^\infty\frac{1}{(n-2)!}\sum_{\mathbf{P}}\sum_{\mathbf{T}}\sum_{\mathbf{v}\in\mathscr{V}^n}\ind_{\mathbf{v}|_{[2]}=\mathbf{w}}\prod_{i=1}^{\underline{n}}w_\beta(T_i,\mathbf{v}|_{P_i})\prod_{j=1}^{\underline{n}-1}e(\vartheta_{\underline{\varphi}(j)},\vartheta_{\underline{\varphi}(j+1)}),
\end{equation*}and an analogous formula for $\W_\beta(\mathbf{w})$. 

It is useful to separate the sum term indexed by $\underline{n}=2$. Note that in this case there is a single injection $\underline{\varphi}$ and $\underline{N}=2$. Furthermore, the term involving edge functions simplifies to $e(\varOmega_1,\varOmega_2)$. Since the indices are mute, the matter becomes counting the set of partitions $\mathbf{P}=(P_1,P_2)$ of $[n]$ such that $1\in P_1$ and $2\in P_2$. One readily finds that the term indexed by $\underline{n}=2$ is equal to \begin{equation*}
    \overline{\W}_\beta(\varOmega_1)e(\varOmega_1,\varOmega_2)\overline{\W}_\beta(\varOmega_2).
\end{equation*}Since $\mathscr{V}_{n_1+n_2}\subset\mathscr{V}_{n_1}\times\mathscr{V}_{n_2}$, where equality is not necessary, in the case of $\W$, the same argument yields an upper bound: the term indexed by $\underline{n}=2$ is bounded above by \begin{equation*}
    \W_\beta(\varOmega_1)e(\varOmega_1,\varOmega_2)\W_\beta(\varOmega_2).
\end{equation*}

Henceforth, we suppose that $\underline{n}\geq 3$. Let us begin by eliminating the sum over $\underline{\varphi}$ and transform it into a sum over $\underline{N}\geq\underline{n}$. For each $\underline{\varphi}$, this will be done by relabeling $i\in[n]$ (not $\underline{n}$ or $\underline{N}$) so that for $i\in[\underline{n}]$, $\underline{\varphi}(i)=i$. This operation will introduce a combinatorial factor which we compute in the following. Given $\underline{N}$, one must simply count the amount of injections from $[\underline{n}]$ to $\mathbb{N}$ such that $\underline{\varphi}(1)=1$, $\underline{\varphi}(\underline{n})=2$ and $\max\textrm{Im}(\underline{\varphi})=\underline{N}$. There are $(\underline{n}-2)$ possible indices such that $\underline{\varphi}(i)=\underline{N}$, in which case one must simply distribute $(\underline{N}-3)$ values over the remaining $(\underline{n}-3)$ indices. That is, one obtains the following combinatorial factor \begin{equation*}
    (\underline{n}-2)\times(\underline{N}-3)...(\underline{N}-\underline{n}+1)=(\underline{n}-2)\frac{(\underline{N}-3)!}{(\underline{N}-\underline{n})!}.
\end{equation*} The sum over $\underline{n}\geq3$ becomes \begin{equation*}
    \sum_{\underline{n}=3}^\infty\sum_{\underline{N}=\underline{n}}^\infty\sum_{n=\underline{N}}^\infty\frac{(\underline{N}-3)!(\underline{n}-2)}{(\underline{N}-\underline{n})!(n-2)!}\sum_{\mathbf{P}}\sum_{\mathbf{T}}\sum_{\mathbf{v}\in\mathscr{V}^n}\ind_{\vartheta_1=\varOmega_1}\ind_{\vartheta_{\underline{n}}=\varOmega_2}\prod_{i=1}^{\underline{n}}w_\beta(T_i,\mathbf{v}|_{P_i})\prod_{j=1}^{\underline{n}-1}e(\vartheta_{j},\vartheta_{j+1}).
\end{equation*}We note that the ordered partition now satisfies $i\in P_i$, for $i\in [\underline{n}]$.

The next step is to relabel the ordered partitions. Consider a vector $\mathbf{n}=(n_1,...,n_{\underline{n}})$ such that $n_i\geq 1$ and $n_1+...+n_{\underline{n}}=n$. Since $i\in P_i$, there are \begin{equation*}
    \frac{(n-\underline{n})!}{\prod_{i=1}^{\underline{n}}(n_i-1)!}
\end{equation*}ordered partitions $\mathbf{P}$ such that $\lvert P_i\rvert=n_i$. In this case, relabeling them in terms of $\mathbf{n}$ and recalling Equation~\eqref{eq: mute-indices-cluster-formula}, the sum becomes \begin{equation*}
    \sum_{\underline{n}=3}^\infty\sum_{\underline{N}=\underline{n}}^\infty\sum_{n=\underline{N}}^\infty C(\underline{n},\underline{N},n)\sum_{\mathbf{n}}\sum_{\mathbf{v}\in\mathscr{V}^{\underline{n}}}\ind_{\vartheta_1=\varOmega_1}\ind_{\vartheta_{\underline{n}}=\varOmega_2}\prod_{i=1}^{\underline{n}}\overline{\W}_{\beta,n_i}(\vartheta_i)\prod_{j=1}^{\underline{n}-1}e(\vartheta_j,\vartheta_{j+1}),
\end{equation*}where \begin{equation*}
    C(\underline{n},\underline{N},n)\coloneqq \frac{(\underline{n}-2)(\underline{N}-3)!(n-\underline{n})!}{(\underline{N}-\underline{n})!(n-2)!}=\binom{\underline{N}-3}{\underline{n}-3}\binom{n-2}{\underline{n}-2}^{-1}.
\end{equation*}Once again, in the case of $\W_\beta(\cdot)$, the analogous formula is an upper bound.

The final step consists of showing that, for all $\mathbf{v}\in\mathscr{V}^{\underline{n}}$, \begin{equation*}
        \sum_{\underline{N}=\underline{n}}^\infty\sum_{n=\underline{N}}^\infty C(\underline{n},\underline{N},n)\sum_{\mathbf{n}}\prod_{i=1}^{\underline{n}}\overline{\W}_{\beta,n_i}(\vartheta_i)=\prod_{i=1}^{\underline{n}}\sum_{n_i=1}^\infty\overline{\W}_{\beta, n_i}(\vartheta_i)=\prod_{i=1}^{\underline{n}}\overline{\W}_\beta(\vartheta_i),
\end{equation*}where we recall Equation~\eqref{eq: mute-indices-cluster-formula}. Noting that \begin{equation*}
    \sum_{\underline{N}=\underline{n}}^\infty\sum_{n=\underline{N}}^\infty\sum_{\mathbf{n}}=\sum_{n_1=...=n_{\underline{n}}=1}^\infty\sum_{n=\underline{n}}^\infty\sum_{\underline{N}=\underline{n}}^n\ind_{n=n_1+...+n_{\underline{n}}},
\end{equation*}it is sufficient to show that \begin{equation}\label{eq: computation-mute-indices}
    \sum_{n=\underline{n}}^\infty\sum_{\underline{N}=\underline{n}}^n\ind_{n=n_1+...+n_{\underline{n}}}C(\underline{n},\underline{N},n)=1.
\end{equation}We recall that if $r,s\geq 0$, then \begin{equation*}
    \sum_{t=0}^r\binom{s+t}{t}=\binom{r+s+1}{r+1}.
\end{equation*}Hence, \begin{equation*}
    \sum_{\underline{N}=\underline{n}}^n\binom{\underline{N}-3}{\underline{n}-3}=\sum_{\underline{N}=0}^{n-\underline{n}}\binom{\underline{n}-3+\underline{N}}{\underline{N}}=\binom{n-2}{\underline{n}-2}.
\end{equation*}We have thus shown that Equation~\eqref{eq: computation-mute-indices} holds. The Lemma has thus been proved.
\end{proof}

By Lemma~\ref{lem: two-point-arbitrary sums}, it is clear that we must understand the (modified) total edge function. In general, we want to prove a bound of the type \begin{equation}\label{eq: total-edge-function-bound}
    \overline{\E}_\beta(\varOmega_1,\varOmega_2)\leq \eta v_{\beta_1}(\varOmega_1)^{-1}\tilde{e}(\varOmega_1,\varOmega_2)v_{\beta_1}(\varOmega_2)^{-1},
\end{equation}where $\beta_1\in\mathbb{R}$ is fixed, $\tilde{e}:\mathscr{V}^2\rightarrow\mathbb{R}_+$ is a \textit{good} edge function, and $\eta>0$ can be made arbitrarily small, provided $\beta\geq\beta_0=\beta_0(\eta)$ is sufficiently large. Although we have not found a reasonable abstract assumption that proves this bound, we are able to prove it in our case. On the other hand, in the case of the modified total edge function, we can prove, under reasonable hypotheses, that \begin{equation}\label{eq: modified-tota-edge-bound}
    \E_\beta(\varOmega_1,\varOmega_2)\leq \eta v'_{\beta/2}(\varOmega_1)^2e(\varOmega_1,\varOmega_2)v'_{\beta/2}(\varOmega_2)^2.
\end{equation}

Choosing $\eta=1$ in Lemma~\ref{lem: reducing-sums-over-trees}, if $\beta_0$ is sufficiently large, then $\W_\beta(\vartheta)\leq v_{\beta/2}(\vartheta)$, for all $\beta\geq\beta_0$ and $\vartheta\in\mathscr{V}$. In this case, if $\beta\geq\beta_0$, then \begin{equation}\label{eq: E-E-prime-bound}
    \E_{\beta,n}(\varOmega_1,\varOmega_2)\leq \sum_{\mathbf{v}\in\mathscr{V}_{n+2}}\ind_{\vartheta_1=\varOmega_1}\ind_{\vartheta_{n+2}=\varOmega_2}\prod_{i=2}^{n+1}v_{\beta/2}(\vartheta_i)\prod_{j=1}^{n+1}e(\vartheta_j,\vartheta_{j+1})\eqqcolon \E'_{\beta/2,n}(\varOmega_1,\varOmega_2).
\end{equation}
This motivates us to introduce the \textit{two-vertex contracting function family} $\mathfrak{R}^2\coloneqq\{\mathcal{R}_\beta^2:\beta>1\}$, where $\mathcal{R}_{\beta}^2:\mathscr{V}^2\rightarrow \mathbb{R}_{\geq 0}$ is given by \begin{equation}\label{eq: reducing-operator-definition}
    \mathcal{R}_{\beta}^2 (\vartheta_1,\vartheta_2)\coloneqq\sum_{\vartheta\in\mathscr{V}}e(\vartheta_1,\vartheta)v_{\beta}(\vartheta)e(\vartheta,\vartheta_2).
\end{equation}By Hypothesis~\ref{hyp: vertex-edge-function-hypotheses}, the function $\beta\mapsto \mathcal{R}^2_ \beta(\vartheta_1,\vartheta_2)$ is increasing for any $(\vartheta_1,\vartheta_2)\in\mathscr{V}^2$.

\begin{hypothesis}\label{hyp: reduction-operator}
    The family function $\mathfrak{V}'_2$ dominates the family function $\mathfrak{R}^2$. That is, \begin{equation*}
        \lim_{\beta\rightarrow\infty}\sup_{(\vartheta_1,\vartheta_2)\in\mathscr{V}_2}\left\lvert\frac{\mathcal{R}_{\beta}^2(\vartheta_1,\vartheta_2)}{e(\vartheta_1,\vartheta_2)[v'_\beta(\vartheta_1)+v'_{\beta}(\vartheta_2)]}\right\rvert=0.
    \end{equation*}
\end{hypothesis}

\begin{remark}\label{rem: second-reduction-function}
    We will employ this object only in the case of contours. By Notation~\ref{not: calC_k} and Equation~\eqref{eq: modified-vertex-function-contours}, Hypothesis~\ref{hyp: reduction-operator} means that for all $\eta>0$, there exists $\beta_0$ such that if $\beta\geq \beta_0$ and $(\gamma_1,\gamma_2)\in\mathscr{C}_2$, then \begin{equation*}
        \mathcal{R}^2_\beta(\gamma_1,\gamma_2)\leq \eta[\beta H(\gamma_1)+\beta H(\gamma_2)] \Phi(\gamma_1,\gamma_2).
    \end{equation*}
\end{remark}

\begin{lemma}\label{lem: reduction-function-sums}
    Let $\eta\in (0,1)$. There exists $\beta_0=\beta_0(\eta)$ such that if $\beta\geq\beta_0$, $n\in\mathbb{N}$ and $(\varOmega_1,\varOmega_2)\in\mathscr{V}_2$, then \begin{equation*}
        \E_{\beta,n}(\varOmega_1,\varOmega_2)\leq \eta^nv'_{\beta/2}(\varOmega_1)^2e(\varOmega_1,\varOmega_2)v'_{\beta/2}(\varOmega_2)^2.
    \end{equation*}
\end{lemma}

\begin{proof}
    First, employing an idea very similar to that in the proof of Proposition 3.7 in \cite{cluster}, we obtain a bound for $\E'_{\beta,n}(\varOmega_1,\varOmega_2)$. Then, we conclude by Equation~\eqref{eq: E-E-prime-bound}. We provide the algorithm operated on $\E'_{\beta,n}(\varOmega_1,\varOmega_2)$ below. Its crucial feature is the fact that no vertex is visited more than twice.

    \begin{figure}[hbt!]
    \centering
    \input{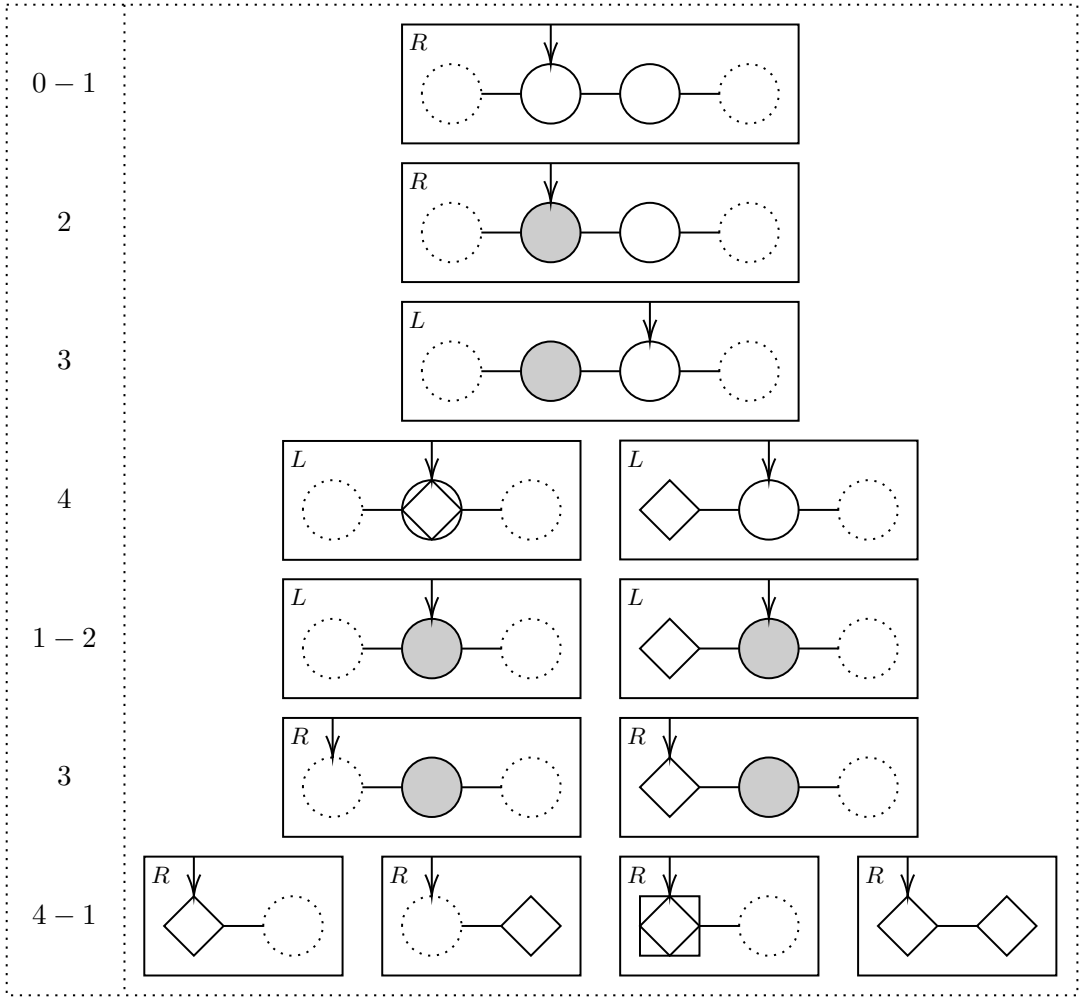}
    \caption{A graphical description of the algorithm in case $n=2$. The number to the left is the current step ($0$ denotes the starting position). A dotted circle represents a weight of $1$, a blank circle represents a weight of $v_\beta(\cdot)$ and a gray circle represents a weight of $v_{\beta/2}(\cdot)$. A blank lozenge or square represents a weight of $v'_\beta(\cdot)$. Superposition of figures denotes product of weights.}
    \label{fig: two-vertices-algorithm}
\end{figure}

    There is a \textit{state} $s$ that takes values in $\{L,R\}$, which indicates in which direction the pointer will move next: $L$ is for \textit{left} and $R$ is for \textit{right}. The starting position of the pointer is $n_0=\lfloor n/2\rfloor+1$ (recall that there are $n+2$ vertices). If $n$ is odd, the starting state is $s_0=L$. If $n$ is even, the starting state is $s_0=R$.

    After setting up the starting position, we run the following \textit{monotone} algorithm. More precisely, we run an algorithm that can only increase the value of the estimated sum.

    In the first step, we check whether the pointer is over an extremal vertex. If it is, then the algorithm ends. Otherwise, we proceed to the second step.

    In the second step, we update the weight of the vertex indicated by the pointer. More precisely, its current weight, which is equal to $v_\beta(\cdot)v'_\beta(\cdot)^{\xi}$, where $\xi\in\{0,1,2\}$, is updated to $2v_{\beta/2}(\cdot)$. This is possible due to Hypothesis~\ref{hyp: vertex-edge-function-hypotheses} and $x^2\leq 2e^x$, if $x>0$.

    In the third step, the pointer moves in the direction indicated by the state, then it updates ($L\rightarrow R$ or $R\rightarrow L$). 
    
    In the fourth step, if $s=L$ ($s=R$), then we remove the vertex to the left (right) of the pointer employing Hypothesis~\ref{hyp: reduction-operator}. More precisely, if $\beta\geq \beta_0(\eta)$, then \begin{equation*}
        \mathcal{R}^2_{\beta/2}(\vartheta_1,\vartheta_2)\leq (\eta/4)[v'_{\beta/2}(\vartheta_1)+v'_{\beta/2}(\vartheta_2)]e(\vartheta_1,\vartheta_2)\leq (\eta/4)[v'_\beta(\vartheta_1)+v'_\beta(\vartheta_2)]e(\vartheta_1,\vartheta_2),
    \end{equation*}where the second inequality follows from Hypothesis~\ref{hyp: vertex-edge-function-hypotheses}. In this way, a single sum is split into two: in the first, the weight of the vertex indicated by the pointer has been multiplied by $v'_\beta(\cdot)$; in the second, the vertex to the direction indicated by the state of the pointer has been multiplied by $v'_\beta(\cdot)$. We move to the first step.

    For a detailed graphical description of the algorithm, see Figure~\ref{fig: two-vertices-algorithm}. By the end, we obtain \begin{align*}
        \E'_{\beta,n}(\varOmega_1,\varOmega_2)&\leq  2^n2^n(\eta/4)^n\max\left\{v'_\beta(\varOmega_1)^{\xi_1}v'_\beta(\varOmega_2)^{\xi_2}:\xi_1,\xi_2\in\{0,1,2\}\right\}e(\Omega_1,\Omega_2)\\
        &=\eta^nv'_\beta(\varOmega_1)^2e(\varOmega_1,\varOmega_2)v'_\beta(\varOmega_2)^2,
    \end{align*}where the equality is a consequence of Hypothesis~\ref{hyp: vertex-edge-function-hypotheses}, as long as $\beta\geq\beta_0$ is sufficiently large. The first term $2^n$ comes from the number of sums generated as the algorithm runs. The second term $2^n$ comes from the second step.  Recalling Equation~\eqref{eq: E-E-prime-bound}, this concludes the proof of the lemma.
\end{proof}

\begin{corollary}\label{cor: reduction-operator-sums}
    There exists $\beta_0>0$ such that if $\beta\geq \beta_0$ and $(\varOmega_1,\varOmega_2)\in\mathscr{V}_2$, then \begin{equation*}\label{eq: reduction-operator-sums}
        \W_\beta
        (\varOmega_{1},\varOmega_2)\leq v_{\beta/4}(\varOmega_{1})e(\varOmega_{1},\varOmega_2)v_{\beta/4}(\varOmega_2).
    \end{equation*}
\end{corollary}

\begin{proof}
    The proof is a direct application of Lemmas~\ref{lem: reducing-sums-over-trees} and \ref{lem: reduction-function-sums}. We have\begin{align*} 
        \W_\beta(\varOmega_1,\varOmega_2)&\leq v_{\beta/2}(\varOmega_1)e(\varOmega_1,\varOmega_2)v_{\beta/2}(\varOmega_2)+\sum_{n=1}^\infty v_{\beta/2}(\varOmega_1)\E'_{\beta/2,n}(\varOmega_1,\varOmega_2)v_{\beta/2}(\varOmega_2).
    \end{align*}By Hypothesis~\ref{hyp: vertex-edge-function-hypotheses}, we can bound the first term by \begin{align*}
        v_{\beta/2}(\varOmega_1)e(\varOmega_1,\varOmega_2)v_{\beta/2}(\varOmega_2)\leq 2^{-1}v_{\beta/4}(\varOmega_1)e(\varOmega_1,\varOmega_2)v_{\beta/4}(\varOmega_2).
    \end{align*}Let us now bound the second term. Choose $\eta=9^{-1}$ in Lemma~\ref{lem: reduction-function-sums}. Recalling that \begin{equation*}
        v'_{\beta/2}(\vartheta)^2v_{\beta/2}(\vartheta)\leq 2!v_{\beta/4}(\vartheta)=2v_{\beta/4}(\vartheta),
    \end{equation*} we obtain \begin{equation*}
        \sum_{n=1}^\infty v_{\beta/2}(\varOmega_1)\E'_{\beta/2,n}(\varOmega_1,\varOmega_2)v_{\beta/2}(\varOmega_2)\leq 2^{-1}v_{\beta/4}(\varOmega_1)e(\varOmega_1,\varOmega_2)v_{\beta/4}(\varOmega_2).
    \end{equation*}The corollary has thus been proved.
\end{proof}

\subsection{Many-vertex estimates}\label{subsec: many-vertex-estimates}

We now tackle the case $m\geq 3$. It is simpler than it seems, we just need a generalization of Lemma~\ref{lem: two-point-arbitrary sums} and a generalization of the contraction function to an arbitrary (finite) number of vertices. Let us first prove a simple result on tree sizes that is very useful in accomplishing the first task.

\begin{lemma}\label{lem: tree-size}
    Let $m,n\in\mathbb{N}$ such that $2\leq m\leq n$. Let $T\in\Tr_n$ such that\begin{enumerate}
        \item $L(T)\subset [m]$;
        \item $v\in [m+1,n]\Rightarrow \degr_T(v)\geq 3$.
    \end{enumerate} Then, $n\leq 2m-2$.
\end{lemma}

\begin{proof}
    On the one hand, since $T$ is a tree, we have $\lvert E(T)\rvert=n-1$. On the other hand, we have \begin{align*}
        \lvert E(T)\rvert&=\frac{1}{2}\sum_{v\in [n]}\degr_T(v)=\frac{1}{2}\sum_{v\in [m]}\degr_T(v)+\frac{1}{2}\sum_{v=m+1}^n\degr_T(v)\geq \frac{1}{2}m+\frac{3}{2}(n-m)=\frac{1}{2}(3n-2m).
    \end{align*}We thus obtain $n\leq 2m-2$.
\end{proof}

\begin{remark}
    It is not difficult to construct a tree that satisfies the hypotheses of Lemma~\ref{lem: tree-size} such that $n=2m-2$. That is, $n=2m-2$ is the optimal bound.
\end{remark}

\begin{notation}\label{not: special-trees}
    Let $m\leq n\leq 2m-2$, we denote by $\Tr^*_{n,m}$ the set of elements of $\Tr_{n,m}$ (recall Notation~\ref{not: trees-with-set-of-leaves-fixed}) satisfying the conditions of Lemma~\ref{lem: tree-size}. 
    
    Let $[m]\subset B$ such that $\lvert B\rvert\leq 2m-2$, we denote by $\Tr^*_m(B)$ the set of elements $T\in \Tr(B,[m])$ such that $b\in B\setminus[m]\Rightarrow \degr_T(b)\geq 3$. Finally, we also write \begin{equation*}
        \Tr^*_m\coloneqq\bigsqcup_{n=m}^{2m-2}\Tr^*_{n,m}\textrm{ and }\overline{\Tr}^*_m\coloneqq\bigsqcup_{\substack{[m]\subset B\\\lvert B\rvert\leq 2m-2}}\Tr^*_m(B).
    \end{equation*}For a graphical representation of trees in $\Tr^*_m$ for $m$ small, see Figure~\ref{fig: star-trees} below.
\end{notation}

\begin{figure}[hbt!]
    \centering
    \input{Fig_star-trees}
    \caption{All examples of trees in $\Tr^*_{n,m}$ for $2\leq m\leq 5$. White vertices correspond to those labeled by $[m]$ and black vertices correspond to those labeled by $[m+1,n]$.}
    \label{fig: star-trees}
\end{figure}

\begin{lemma}\label{lem: many-point-arbitrary-sums}
    Let $\mathbf{w}\in(\varOmega_1,...,\varOmega_m)\in\mathscr{V}^m$, where $m\geq 2$. Then, \begin{equation}\label{eq: many-point-arbitrary-sums}
        \begin{split}
            \overline{\W}_\beta(\mathbf{w})&=\sum_{n=m}^{2m-2}\frac{1}{(n-m)!}\sum_{T\in\Tr^*_{n,m}}\sum_{\mathbf{v}\in\mathscr{V}^n}\ind_{\mathbf{v}|_{[m]}=\mathbf{w}}\prod_{i=1}^n\overline{\W}_\beta(\vartheta_i)\hspace{-0.35cm}\prod_{\{j,k\}\in E(T)}\hspace{-0.35cm}\overline{\E}_\beta(\vartheta_j,\vartheta_k),\\
            \W_\beta(\mathbf{w})&\leq\sum_{n=m}^{2m-2}\frac{1}{(n-m)!}\sum_{T\in\Tr^*_{n,m}}\sum_{\mathbf{v}\in\mathscr{V}_n}\ind_{\mathbf{v}|_{[m]}=\mathbf{w}}\prod_{i=1}^n\W_\beta(\vartheta_i)\hspace{-0.35cm}\prod_{\{j,k\}\in E(T)}\hspace{-0.35cm}\E_\beta(\vartheta_j,\vartheta_k).
        \end{split}
    \end{equation}
\end{lemma}

\begin{remark}
    Note that, due to Lemma~\ref{lem: tree-size}, the combinatorial factor in Equation~\eqref{eq: many-point-arbitrary-sums} is immaterial, when it comes to perturbative results.
\end{remark}

\begin{proof}
    Since the combinatorial part of the argument is identical to that of the proof of Lemma~\ref{lem: two-point-arbitrary sums}, we just describe the argument in terms of operations of trees.

    In the first step, we contract every tree to ensure that the set of leaves is a subset of $[m]$. Then, we relabel them to ensure that their vertex set is given by $[n]$, for some $n\geq m$. That is, recalling Notation~\ref{not: trees-with-set-of-leaves-fixed}, we obtain \begin{equation}\label{eq: m-points-abstract-sum-step-one}
        \overline{\W}_\beta(\mathbf{w})=\sum_{n=m}^\infty\frac{1}{(n-m)!}\sum_{T\in\Tr_{n,m}}\sum_{\mathbf{v}\in\mathscr{V}^n}\ind_{\mathbf{v}|_{[m]}=\mathbf{w}}\prod_{i=1}^n\overline{\W}_{\beta}(\vartheta_i)\hspace{-0.35cm}\prod_{\{j,k\}\in E(T)}\hspace{-0.35cm}e(\vartheta_j,\vartheta_k).
    \end{equation}
    
    In the second step, we contract lines to guaranty that if a vertex is not labeled by an element of $[m]$, then it has degree at least $3$. Let $\Psi^*_m:\Tr_m\rightarrow\Tr^*_m$ be the application defined by the composition of the following two maps. First, let $T\in\Tr_m$, define \begin{equation*}
        V(T^*)\coloneqq[m]\sqcup\{i\in V(T): (i>m)\wedge (\delta_T(i)\geq 3)\}.
    \end{equation*}Let $i,j\in V(T^*)$, we denote the (unordered) path in $T$ connecting $i$ and $j$ by $\mathfrak{p}_T(i,j)$. In this case, we also define \begin{equation*}
        E(T^*)=\{\{i,j\}\in\mathrm{P}_2(V(T^*)):V(\mathfrak{p}_T(i,j))\cap V(T^*)=\{i,j\}\}.
    \end{equation*}Since $T^*\coloneqq (V(T^*),E(T^*))\in\Tr^*_m$, we have just constructed an application $\overline{\Psi}^*_m:\Tr_m\rightarrow\Tr^*_m$.  
    
    \begin{figure}[hbt!]
    \centering
    \input{Fig_theo-m-points-1}
    \caption{A graphic description of the procedure where $m=4$.}
    \label{fig: special-trees-application}
\end{figure}
    
    Let $[m]\subset B$ such that $\lvert B\rvert\leq 2m-2$. We denote by $\varphi_B:B\rightarrow[\lvert B\rvert]$ the canonical bijection satisfying \begin{equation*}
        b_1<b_2\in B\Rightarrow\varphi_B(b_1)<\varphi_B(b_2).
    \end{equation*}The family $\{\varphi_B:([m]\subset B\in\mathrm{P}(\mathbb{N}))\wedge(\lvert B\rvert\leq 2m-2)\}$ induces a canonical relabeling application $\overline{\varphi}_m:\Tr^*_m\rightarrow\Tr^*_m$. Finally, we define $\Psi^*_m\coloneqq\overline{\varphi}_m\circ\overline{\Psi}^*_m$. For a visual description, see Figure~\ref{fig: special-trees-application}.
    
    Proceeding as in the proof of Lemma~\ref{lem: two-point-arbitrary sums}, we transform Equation~\eqref{eq: m-points-abstract-sum-step-one} into Equation~\eqref{eq: many-point-arbitrary-sums}. Once again, we note that in the case of the global compatibility condition, we obtain upper bounds.    
\end{proof}

Fix $m\geq 3$. By Hypothesis~\ref{hyp: vertex-edge-function-hypotheses}, Lemma~\eqref{lem: reducing-sums-over-trees} and Equation~\eqref{eq: modified-tota-edge-bound} (consequence of Lemma~\ref{lem: reduction-function-sums}), if $\beta\geq\beta_0$ is sufficiently large, we may suppose that \begin{align*}
    &\W_\beta(\vartheta)\leq [(8m-12)!]^{-1/m}v_{\beta/2}(\vartheta),\\
    &\E_\beta(\vartheta_1,\vartheta_2)\leq v'_{\beta/2}(\vartheta_1)^2e(\vartheta_1,\vartheta_2)v'_{\beta/2}(\vartheta_2)^2,
\end{align*}for all $\vartheta,\vartheta_1,\vartheta_2\in\mathscr{V}$. By Lemma~\ref{lem: many-point-arbitrary-sums}, for all $\mathbf{w}\in\mathscr{V}_m$, we obtain \begin{equation*}
    \W_\beta(\mathbf{w})\leq\sum_{n=m}^{2m-2}\frac{[8m-12)!]^{-1}}{(n-m)!}\sum_{T\in\Tr^*_{n,m}}\sum_{\mathbf{v}\in\mathscr{V}_n}\ind_{\mathbf{v}|_{[m]}=\mathbf{w}}\prod_{i=1}^n v'_{\beta/2}(\vartheta_i)^{2\degr_{T}(i)}v_{\beta/2}(\vartheta_i)\hspace{-0.35cm}\prod_{\{j,k\}\in E(T)}\hspace{-0.35cm}e(\vartheta_j,\vartheta_k).
\end{equation*}On the other hand, since $x>0\Rightarrow x^n\leq n!e^x$, we have \begin{equation*}
    v'_{\beta/2}(\vartheta_i)^{2\degr_{T}(i)}v_{\beta/2}(\vartheta_i)\leq (2\degr_{T}(i))!v_{\beta/4}(\vartheta_i),
\end{equation*}and \begin{equation*}
    \sum_{i=1}^n2\cdot\degr_{T}(i)=4(n-1)\leq 8m-12\Rightarrow\prod_{i=1}^n(2\cdot\degr_{T}(i))!\leq(8m-12)!. 
\end{equation*}Therefore, \begin{equation}\label{eq: W'-definition}
    \begin{split}
        \W_\beta(\mathbf{w})&\leq\sum_{n=m}^{2m-2}\sum_{T\in\Tr^*_{n,m}}\frac{1}{(n-m)!}\sum_{\mathbf{v}\in\mathscr{V}_n}\ind_{\mathbf{v}|_{[m]}=\mathbf{w}}\prod_{i=1}^nv_{\beta/4}(\vartheta_i)\hspace{-0.35cm}\prod_{\{j,k\}\in E(T)}\hspace{-0.35cm}e(\vartheta_j,\vartheta_k)\\
        &=\sum_{n=m}^{2m-2}\sum_{T\in\Tr^*_{n,m}}\W_{\beta/4}[\mathbf{w}](T)\eqqcolon\sum_{n=m}^{2m-2}\W^*_{\beta/4,n}(\mathbf{w}).
    \end{split}
\end{equation}This motivates us to define the following generalization of the contracting function introduced in Equation~\eqref{eq: reducing-operator-definition}. Let $k\geq 1$, we define the $k$\textit{-th contracting function family} $\mathfrak{R}^k=\{\mathcal{R}^k_\beta:\beta>1\}$, where $\mathcal{R}^k_{\beta}:\mathscr{V}^k\rightarrow \mathbb{R}_{\geq 0}$ is given by \begin{equation}\label{eq: kth-contracting-function-definition}
    \mathcal{R}^k_{\beta}(\vartheta_1,...,\vartheta_k)\coloneqq\sum_{\vartheta\in\mathscr{V}}v_\beta(\vartheta)\prod_{i=1}^ke(\vartheta,\vartheta_i).
\end{equation} We suppose that it satisfies the following generalization of Hypothesis~\ref{hyp: reduction-operator}. 

\begin{hypothesis}\label{hyp: kth-contracting-function}
The family function $\mathfrak{V}'_k$ dominates the family function $\mathfrak{R}^k$. That is, \begin{equation*}\label{eq: kth-contracting-function-hypothesis}
    \lim_{\beta\rightarrow\infty}\sup_{\mathbf{v}\in\mathscr{V}_k}\left\lvert\frac{\mathcal{R}_{\beta}^k(\mathbf{v})}{\sum_{\varphi\in\mathrm{Sym}(k)}v'_\beta(\vartheta_{\varphi(k)})\prod_{i=1}^{k-1}e(\vartheta_{\varphi(i)},\vartheta_{\varphi(i+1)})}\right\rvert=0,
\end{equation*}where we recall that $\mathrm{Sym}(k)$ denotes the $k$-th symmetric group.
\end{hypothesis}

We note that Remark~\ref{rem: second-reduction-function} extends to $k$-th contracting functions, where $k\geq3$. Nevertheless, the graphical interpretation is more subtle in this case. 

Let $m<n\leq 2m-2$ and $T\in\Tr^*_{n,m}$. Recalling Notation~\ref{not: specific-vertex-sets}, we write $$\mathrm{d}\equiv\degr_T(n)\geq 3,\textrm{ }E\equiv E_T(n)\textrm{ and }N_T(n)=\{l_1,...,l_\mathrm{d}\},$$ where $l_1<...<l_\mathrm{d}$. Thus, \begin{equation*}
    \W_\beta[\mathbf{w}](T)=\sum_{\mathbf{v}\in\mathscr{V}_{n-1}}\ind_{\mathbf{v}|_{[m]}=\mathbf{w}}\prod_{i=1}^{n-1}v_\beta(\mathbf{\vartheta}_i)\hspace{-.45cm}\prod_{\{j,k\}\in E(T)\setminus E}\hspace{-.45cm}e(\vartheta_j,\vartheta_k)\cdot \mathcal{R}^\mathrm{d}_\beta(\vartheta_{l_1},....,\vartheta_{l_\mathrm{d}}).
\end{equation*}We begin by noting that removing the vertex $n$ induces an ordered partition $\mathbf{P}=(Q_1,...,Q_\mathrm{d})$ of $[n-1]$ such that if $p\in[\hspace{.05cm}\mathrm{d}\hspace{.05cm}]$, then $l_p\in Q_p$ and $T|_{Q_p}\in\Tr(Q_p)$, where $E(T|_{Q_p})\coloneqq E(T)\cap\mathrm{P}(Q_p)$. Hence, \begin{equation*}
    \prod_{i=1}^{n-1}v_\beta(\vartheta_i)\hspace{-.35cm}\prod_{\{j,k\}\in E(T)\setminus E}\hspace{-.35cm}e(\vartheta_j,\vartheta_k)=\prod_{p=1}^\mathrm{d} w_\beta[T|_{Q_p}](\mathbf{v}|_{Q_p}).
\end{equation*}If $\eta_\mathrm{d}>0$, then, by Hypothesis~\ref{hyp: kth-contracting-function}, there exists $\beta_0=\beta_0(\mathrm{d},\eta_\mathrm{d})$ such that $\beta\geq\beta_0$ implies \begin{equation*}
    \mathcal{R}^\mathrm{d}_\beta(\vartheta_{i_1},...,\vartheta_{i_\mathrm{d}})\leq \eta_\mathrm{d}\sum_{\varphi\in\mathrm{Sym}(\mathrm{d})}v'_\beta(\vartheta_{l_{\varphi(\mathrm{d})}})\prod_{p=1}^{\mathrm{d}-1}e(\vartheta_{l_{\varphi(p)}},\vartheta_{l_{\varphi(p+1)}}).
\end{equation*} Combining this fact with Hypothesis~\ref{hyp: vertex-edge-function-hypotheses} and $v_\beta(\cdot)v'_\beta(\cdot)\leq v_{\beta/2}(\cdot)$, we obtain \begin{equation*}
    \W_\beta[\mathbf{w}](T)\leq \eta_\mathrm{d}\sum_{\varphi\in\mathrm{Sym}(\mathrm{d})}\sum_{\mathbf{v}\in\mathscr{V}_{n-1}}\prod_{p=1}^\mathrm{d} w_{\beta/2}[T|_{Q_p}](\mathbf{v}|_{Q_p})\prod_{q=1}^{\mathrm{d}-1}e(\vartheta_{l_{\varphi(q)}},\vartheta_{l_{\varphi(q+1)}}).
\end{equation*} That is, for every $\varphi\in\mathrm{Sym}(\mathrm{d})$, there exists a (unique) tree $T_\varphi\in\Tr_{n-1}$ such that if $\beta>0$ and $\mathbf{v}\in\mathscr{V}_{n-1}$, then\begin{equation*}
    w_\beta[T_\varphi](\mathbf{v})=\prod_{p=1}^{\mathrm{d}}w_{\beta}[T|_{Q_p}](\mathbf{v}|_{Q_p})\prod_{q=1}^{\mathrm{d}-1}e(\vartheta_{l_{\varphi(q)}},\vartheta_{l_{\varphi(q+1)}}).
\end{equation*}More precisely, 
\[
    E(T_\varphi)=(E(T)\setminus E)\cup \{\{l_{\varphi(p)},l_{\varphi(p+1)}\}:p\in[\mathrm{d}-1]\}.
\]
Therefore, if $\varphi_1,\varphi_2\in\mathrm{Sym}(\mathrm{d})$, then $T_{\varphi_1}=T_{\varphi_2}$ if, and only if, \begin{equation*}
    \{\{{\varphi_1(p)},{\varphi_1(p+1)}\}:p\in[\mathrm{d}-1]\}=\{\{{\varphi_2(p)},{\varphi_2(p+1)}\}:p\in[\mathrm{d}-1]\}.
\end{equation*}It is not hard to see that either $\varphi_1=\varphi_2$ or $\varphi_1=\varphi_2\circ \varphi_{\mathrm{d},r}$, where $\varphi_{\mathrm{d},r}\in\mathrm{Sym}(\mathrm{d})$ denotes the \textit{reversal permutation}; that is, $\varphi_{\mathrm{d},r}(i)=\mathrm{d}-i$ for $i\in[\hspace{.05cm}\mathrm{d}\hspace{.05cm}]$. There is a final, and crucial, observation to be made: if $T\in\Tr_n$ and $\varphi\in\mathrm{Sym}(\mathrm{d})$, then $\degr_T(i)\leq \degr_{T_\varphi}(i)$ for all $i\in[n-1]$. That is, if $m<n\leq 2m-2$ and $T\in\Tr^*_{n,m}$, then $T_\varphi\in\Tr^*_{n-1,m}$ for all $\varphi\in\mathrm{Sym}(\mathrm{d})$. We conclude that if $m<n\leq 2m-2$, $T\in\Tr^*_{n,m}$ and $\beta\geq\beta_0(\mathrm{d},\eta)$, then\begin{equation}\label{eq: W-beta-T-star}
    \W_\beta[\mathbf{w}](T)\leq \eta_\mathrm{d}\sum_{\varphi\in\mathrm{Sym}(\mathrm{d})}\W_{\beta/2}[\mathbf{w}](T_\varphi)\leq 2\eta_\mathrm{d}\sum_{\tilde{T}\in\Tr^*_{n-1,m}}\W_{\beta/2}[\mathbf{w}](\tilde{T})=2\eta_\mathrm{d}\W_{\beta/2,n-1}^*(\mathbf{w}).
\end{equation}Figure~\ref{fig: kthcontractingfunction} provides us with a useful visual description of the operation on trees associated with the $k$-th contracting function, and how it relates to Hypothesis~\ref{hyp: kth-contracting-function}. Although these graph operations are very complex (resolving $\mathcal{R}_{\beta}^k$ generates $k!/2$ distinct trees), this does not pose us any problems. That is because of Lemma~\ref{lem: tree-size}.

\begin{figure}[hbt!]
    \centering
    \input{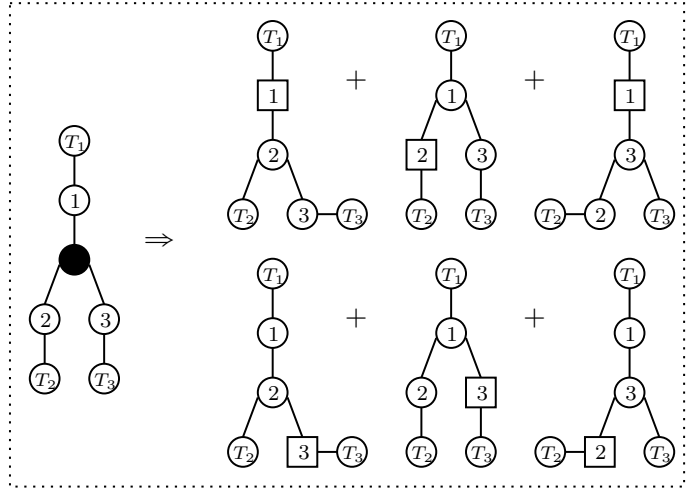}
    \caption{Graphs arising after removing the black vertex. Squares indicate the vertices carrying the extra $v'_\beta(\cdot)$ term. Note that this graph operation preserves the set of leaves. Note that the degree of every remaining vertex does not decrease.}
    \label{fig: kthcontractingfunction}
\end{figure}

\begin{theorem}\label{theo: m-point-bound}
    Let $m\geq 1$. There exists $\beta_0\coloneqq\beta_0(m)>0$ such that if $\beta\geq\beta_0$ and $\mathbf{w}\in\mathscr{V}_m$, then \begin{equation}\label{eq: m-point-bound}
       \W_\beta(\mathbf{w})\leq \sum_{T\in\Tr_m}\prod_{i=1}^mv_{2^{-m}\beta}(\varOmega_i)\hspace{-0.35cm}\prod_{\{j,k\}\in E(T)}\hspace{-0.35cm}e(\varOmega_j,\varOmega_k).
    \end{equation}
\end{theorem}

\begin{proof}
    We assume that $m\geq 3$. The starting point is Equation~\eqref{eq: W'-definition}. Our strategy is to show that if $m\leq n\leq n+1\leq 2m-2$, then \begin{equation}\label{eq: iterative-step-m-point-bound}
        \W^*_{\beta,n+1}(\mathbf{w})\leq \frac{1}{2}\W^*_{\beta/2,n}(\mathbf{w}).
    \end{equation}This is indeed sufficient: assume that Equation~\eqref{eq: iterative-step-m-point-bound} holds. An easy induction argument proves that if $0\leq k\leq m-2$, then  $\W^*_{\beta,m+k}(\mathbf{w})\leq 2^{-k}\W^*_{2^{-k}\beta,m}(\mathbf{w})$. Therefore, \begin{equation}\label{eq: consequence-of-iterative-step-m-point-bound}
        \W_\beta(\mathbf{w})\leq\sum_{n=m}^{2m-2}\W^*_{\beta/4,n}(\mathbf{w})= \sum_{k=0}^{m-2}\W^*_{\beta/4,m+k}(\mathbf{w})\leq\sum_{k=0}^{m-2}\frac{1}{2^k}\W^*_{2^{-2-k}\beta,m}(\mathbf{w}).
    \end{equation}Since $\Tr^*_{m,m}=\Tr_m$,  we have \begin{equation*}
        \W^*_{2^{-2-k}\beta,m}(\mathbf{w})=\sum_{T\in\Tr_m}\prod_{i=1}^mv_{2^{-2-k}\beta}(\varOmega_i)\hspace{-0.35cm}\prod_{\{j,k\}\in E(T)}\hspace{-0.35cm}e(\varOmega_j,\varOmega_k).
    \end{equation*}Thus, if $0\leq k<m-2$, then, by Hypothesis~\ref{hyp: vertex-edge-function-hypotheses}, we may suppose that \begin{equation}\label{eq: W'-same-index-decreasing-beta-bound}
        \W^*_{2^{-2-k}\beta,m}(\mathbf{w})\leq\frac{1}{2}\sum_{T\in\Tr_m}\prod_{i=1}^mv_{2^{-m}\beta}(\varOmega_i)\hspace{-0.35cm}\prod_{\{j,k\}\in E(T)}\hspace{-0.35cm}e(\varOmega_j,\varOmega_k)=\frac{1}{2}\W^*_{2^{-m}\beta,m}(\mathbf{w}).
    \end{equation}Inserting Equation~\eqref{eq: W'-same-index-decreasing-beta-bound} into Equation~\eqref{eq: consequence-of-iterative-step-m-point-bound}, we obtain \begin{equation*}
        \W_\beta(\mathbf{w})\leq\left(\sum_{k=0}^{m-3}\frac{1}{2^{k+1}}+\frac{1}{2^{m-2}}\right)\W^*_{2^{-m}\beta,m}(\mathbf{w})=\sum_{T\in\Tr_m}\prod_{i=1}^mv_{2^{-m}\beta}(\varOmega_i)\hspace{-.35cm}\prod_{\{j,k\}\in E(T)}\hspace{-.35cm}e(\varOmega_j,\varOmega_k).
    \end{equation*}To finish the proof, let us show that Equation~\eqref{eq: iterative-step-m-point-bound} holds. 

    By Hypothesis~\ref{hyp: kth-contracting-function}, there exists $\beta_0$ such that if $\beta\geq\beta_0$, then \begin{equation*}
        \mathcal{R}^k_\beta(\mathbf{v})\leq (4\lvert\Tr_{2m-2}\rvert)^{-1}\sum_{\varphi\in\mathrm{Sym}(k)}v'_\beta(\vartheta_{\varphi(k)})\prod_{i=1}^{k-1}e(\vartheta_{\varphi(i)},\vartheta_{\varphi(i+1)}),
    \end{equation*}for $3\leq k\leq 2m-3$ and all $\mathbf{v}\in\mathscr{V}_k$. Hence, recalling Equation~\eqref{eq: W-beta-T-star}, we obtain \begin{equation*}
        \W^*_{\beta,n+1}(\mathbf{w})=\sum_{T\in\Tr^*_{n+1,m}}\W_\beta[\mathbf{w}](T)\leq \sum_{T\in\Tr^*_{n+1,m}}(2\lvert\Tr_{2m-2}\rvert)^{-1}\W^*_{\beta/2,n}(\mathbf{w})\leq\frac{1}{2}\W^*_{\beta/2,n}(\mathbf{w}),
    \end{equation*}where the final inequality holds since $T\in\Tr^*_{n+1,m}$ implies $\degr_{T}(n+1)\in [3,2m-3]$ and $m\leq n\leq 2m-2$ implies $\lvert\Tr^*_{n,m}\rvert\leq\lvert\Tr_{2m-2}\rvert$. The theorem has thus been proved.    
\end{proof}

\section{Convergence of the polymer gas expansion}\label{sec: convergence}

In this Section, we prove that the formal series in Equation~\eqref{eq: pressure-polymer-gas} converges absolutely provided $\beta$ is sufficiently large. Employing a graph-tree bound, we transform the sum over graphs into a sum over trees, which allows us to prove absolute convergence by applying twice the framework developed in Section~\ref{sec: sums-over-trees}.

\subsection{Implementing the program}\label{subsec: the-program}  First, we recall the following graph-tree bound (see \cite{Penrose1963}, or Proposition 5 in \cite{Fernndez2007}, or the generalist approach in \cite{Scott2005}): \begin{equation}\label{eq: graph-tree-bound}
    \lvert\phi_n(\mathbf{\Gamma})\rvert\leq \sum_{T\in\Tr_n}\ind_{E(T)\subset E(\mathbf{\Gamma})}.
\end{equation} Therefore, recalling Equation~\eqref{eq: pressure-polymer-gas} and Lemma~\ref{lem: upper-activity}, to establish the absolute convergence of the pressure, it is sufficient to show that \begin{equation*}\label{eq: upper-pressure}
    \overline{P}_{\beta,h}(\Lambda)\coloneqq\frac{1}{\lvert \Lambda\rvert}\sum_{n=1}^\infty\frac{1}{n!}\sum_{T\in\Tr_n}\sum_{\mathbf{\Gamma}\in\mathscr{P}^n_\Lambda}\prod_{i=1}^n\bar{z}_{\beta,h}(\Gamma_i)\hspace{-0.25cm}\prod_{\{j,k\}\in E(T)}\hspace{-0.25cm}\ind_{\Gamma_j\not\sim\Gamma_k}
\end{equation*}is uniformly bounded with respect to $\Lambda\in\mathrm{P}(\mathbb{Z})$. Since we are interested in applying the cluster expansion to compute the correlations of the systems, we limit ourselves to dealing with the case where $h$ is non-trivial in a finite (fixed) set of lattice sites. For this reason, we may suppose (see the proof of Theorem~\ref{theo: convergence-cluster-expansion}) that $h\equiv 0$, which we shall do for the rest of this section.

We are interested in two types of sums of trees: trees of polymers and trees of contours. Recalling Definition~\ref{def: families-of-sums}, in the case of polymers, we are interested in the family of \textit{locally} compatible sums $\mathfrak{\overline{W}}$ associated to the triple \begin{equation*}
    (\mathscr{P},\{\bar{z}_\beta(\cdot)\},\ind_{\cdot\not\sim\cdot}).
\end{equation*} In the case of contours, we are interested in the family of \textit{globally} compatible sums $\mathfrak{W}$ associated to the triple (recall Definition~\ref{def: positive-partitions}) \begin{equation*}
    (\mathscr{C},\{e^{-\beta H(\cdot)}\},\Phi(\cdot,\cdot)\ind_{\cdot\parallel\cdot}).
\end{equation*} Since each sum is of a different type, we may establish the following convention:

\begin{convention}\label{conv: tree-of-polymers-trees-of-contours}
    Whenever we write $\overline{\W}_\beta(\cdot)$, we mean a (locally compatible) sum over trees of \textit{polymers}, and whenever we write $\W_\beta(\cdot)$, we mean a (globally compatible) sum over trees of \textit{contours}.
\end{convention}

One should readily recognize that \begin{equation}\label{eq: pressure-tree-function}
    \overline{P}_\beta(\Lambda)\leq\frac{1}{\lvert\Lambda\rvert}\sum_{\Gamma\in\mathscr{P}_\Lambda}\overline{\W}_{\beta}(\Gamma).
\end{equation}In order to apply Lemma~\ref{lem: reducing-sums-over-trees}, which shall be sufficient to prove the absolute convergence of the cluster expansion, we must show that Hypothesis~\ref{hyp: vertex-edge-function-hypotheses} holds for $\{\bar{z}_\beta(\cdot)\}$ and the appropriate leaf pruning function satisfies Hypothesis~\ref{hyp: leaf-pruning-operator}. 

\begin{convention}\label{conv: leaf-pruning-operators}
    In order to distinguish between the leaf pruning function for polymers and for contours, we will denote the first by $\overline{\mathcal{R}^1}_\beta$ and the second by $\mathcal{R}^1_\beta$.
\end{convention}

The first condition is very easy to check. For any $\Gamma\in\mathscr{P}$,\begin{align}\label{eq: v-prime-function-polymers}
    \beta\mapsto\log\frac{\bar{z}_{\beta/2}(\Gamma)}{\bar{z}_{\beta}(\Gamma)}=\sum_{\gamma\in\Gamma}\frac{1}{4}\beta H(\gamma)\textrm{ is increasing and}\inf_{\Gamma\in\mathscr{P}}\log\frac{\bar{z}_{\beta/2}(\Gamma)}{\bar{z}_{\beta}(\Gamma)}\geq\frac{1}{2}\beta\xrightarrow{\beta\rightarrow\infty}\infty.
\end{align}The second condition is the content of Lemma~\ref{lem: notjoin-leaf-operator-satisfies-hypothesis}. It is reasonable that it should be proved later, for its proof rests on the application of the machinery of Section~\ref{sec: sums-over-trees} to trees of contours: the content of the next subsection.

\subsection{Sums of trees of contours}\label{subsec: sums-of-trees-of-contours}

Assume that $\overline{\mathcal{R}^1}_\beta$ satisfies Hypothesis~\ref{hyp: leaf-pruning-operator}. In this case, choosing $\eta=1$ in Lemma~\ref{lem: reducing-sums-over-trees}, for $\beta\geq \beta_0$ sufficiently large, we have \begin{align}\label{eq: pressure-intermediary-bound}
    \overline{P}_\beta(\Lambda)\leq \frac{1}{\lvert\Lambda\rvert}\sum_{\Gamma\in\mathscr{P}_\Lambda}\bar{z}_{\beta/2}(\Gamma).
\end{align}

\begin{remark}
    Although Equation~\eqref{eq: pressure-intermediary-bound} is only valid if $\overline{\mathcal{R}^1}_\beta$ satisfies Hypothesis~\eqref{hyp: leaf-pruning-operator}, every estimate involving sums over trees of contours shown in this section is true regardless. In fact, as mentioned earlier, the proof of Lemma~\ref{lem: notjoin-leaf-operator-satisfies-hypothesis} is based on such estimates.
\end{remark}

\begin{notation}\label{not: calC_k}
    We write \begin{equation*}
        \mathscr{C}_k=\{(\gamma_1,...,\gamma_k)\in\mathscr{C}^k:i\neq j\Rightarrow\gamma_i\parallel\gamma_j\}.
    \end{equation*}
\end{notation}

Recalling Equation~\eqref{eq: upper-activity-no-field}, if $m\in\mathbb{N}$ and $\mathbf{w}=(\gamma_1,...,\gamma_n)\in\mathscr{C}_m$, then\begin{equation}\label{eq: bar-zeta-compared-to-W}
    \sum_{\Gamma\in\mathscr{P}}\bar{z}_\beta(\Gamma)\prod_{i=1}^m\ind_{\gamma_i\in\Gamma}= \W_{\beta/2}(\mathbf{w}),
\end{equation}which justifies the choice of the triple $(\mathscr{C},\{e^{-\beta H(\cdot)}\},\Phi(\cdot,\cdot)\ind_{\cdot\parallel\cdot})$. Similarly to the case of polymers, we must show that $\{e^{-\beta H(\cdot)}\}$ satisfies Hypothesis~\ref{hyp: vertex-edge-function-hypotheses}. If $\gamma\in\mathscr{C}$, then \begin{equation}\label{eq: modified-vertex-function-contours}
    \log\frac{e^{-\frac{1}2\beta H(\gamma)}}{e^{-\beta H(\gamma)}}=\frac{1}{2}\beta H(\gamma)\textrm{ is increasing in $\beta$ and}\inf_{\gamma\in\mathscr{C}}\log\frac{e^{-\frac{1}{2}\beta H(\gamma)}}{e^{-\beta H(\gamma)}}\geq\beta\xrightarrow{\beta\rightarrow\infty}\infty.
\end{equation}
In order to apply Lemma~\ref{lem: reducing-sums-over-trees}, we must now show that the appropriate leaf pruning function $\mathcal{R}^1_\beta:\mathscr{C}\rightarrow \mathbb{R}$ satisfies Hypothesis~\ref{hyp: leaf-pruning-operator}. This is the content of the following lemma:

\begin{lemma}\label{lem: leaf-cutter-operator}
    There exists $\beta_0>0$ such that, if $\beta\geq\beta_0$ and $\gamma_1\in\mathscr{C}$, then \begin{equation*}
        \mathcal{R}^1_\beta(\gamma_1)\leq 2e^{-c_2\beta/2} H(\gamma_1).
    \end{equation*}
\end{lemma}

\begin{proof}
     By definition, we have \begin{equation*}
         \mathcal{R}^1_\beta(\gamma_1)=\sum_{\gamma\in\mathscr{C}}e^{-\beta H(\gamma)}\Phi(\gamma,\gamma_1)\ind_{\gamma\parallel\gamma_1}=\sum_{[\gamma\hspace{0.015cm}]\in\bar{\mathscr{C}}}e^{-\beta H([\gamma])}\sum_{\gamma\in[\gamma\hspace{.015cm}]}\Phi(\gamma,\gamma_1)\ind_{\gamma\parallel\gamma_1},
     \end{equation*}where we introduce the equivalence classes of contours under integer translations. 
     
     Let us estimate the sum over an equivalence class $[\gamma\hspace{0.025cm}]$. We have \begin{equation*}
         \sum_{\substack{\gamma\in[\gamma\hspace{.015cm}]}}\Phi(\gamma,\gamma_1)\ind_{\gamma\parallel\gamma_1}=\sum_{\substack{x\in\I_-(\gamma_1)\\y\in\I_-(\gamma_1)^c}}4J_{xy}\sum_{\gamma\in[\gamma\hspace{.015cm}]}\ind_{y\in\I_-(\gamma)}\ind_{\gamma\parallel\gamma_1}.
     \end{equation*} Since $\ind_{\gamma_1\parallel\gamma}\leq 1$ and at most $\lvert\I_-(\gamma)\rvert\leq \mathrm{diam}(\gamma)$ elements $\gamma\in[\gamma\hspace{.025cm}]$ satisfy $y\in\I_-(\gamma)$, we obtain \begin{equation*}
         \sum_{\substack{\gamma\in[\gamma\hspace{0.015cm}]}}\Phi(\gamma,\gamma_1)\ind_{\gamma\parallel\gamma_1}\leq 2\mathrm{diam}(\gamma)H(\gamma_1).
     \end{equation*} Finally, by Hypotheses~\ref{hyp: diameter-hypothesis} and \ref{hyp: phase-transition-estimate}, we obtain \begin{align*}
         \mathcal{R}^1_\beta(\gamma)\leq 2\sum_{[\gamma_1]\in\overline{\mathscr{C}}}e^{-(\beta-c_0)H([\gamma_1])} H(\gamma)\leq 2e^{-c_2\beta/2} H(\gamma),
     \end{align*} where the last inequality follows by choosing $\beta_0\geq 2c_0$ (recall Hypothesis~\ref{hyp: phase-transition-estimate}). The lemma has thus been proved.
\end{proof}

\begin{corollary}\label{cor: polymers-reduced-to-contours}
    There exists $\beta_0>0$ such that, if $\beta\geq\beta_0$ and $\gamma_1\in\mathscr{C}$, then \begin{equation*}
        \W_\beta(\gamma_1)\leq e^{-\frac{1}{2}\beta H(\gamma_1)}.
    \end{equation*}
\end{corollary}

\begin{proof}
    This is a direct consequence of Lemmas~\ref{lem: reducing-sums-over-trees} and \ref{lem: leaf-cutter-operator}.
\end{proof}

\begin{remark}
    With Corollary~\ref{cor: polymers-reduced-to-contours} proven, the only missing piece needed to prove the absolute convergence of pressure is Lemma~\ref{lem: notjoin-leaf-operator-satisfies-hypothesis}.
\end{remark}

Let us now estimate $\W_\beta(\mathbf{w})$, where $\mathbf{w}=(\gamma_1,...,\gamma_m)\in\mathscr{C}_m$ and $m\geq 2$. That is, let us tackle the problem of multiple contour estimates. We shall need these estimates in the proof of Lemma~\ref{lem: notjoin-leaf-operator-satisfies-hypothesis}. They are also crucial to computing the decay of correlations.

\begin{definition}
    Let $A,B\Subset\mathbb{Z}$ be disjointed and not empty. Define \begin{equation*}
        \Phi(A,B)\coloneqq \sum_{\substack{x\in A\\y\in B}}4J(\lvert x-y\rvert)=\sum_{\substack{x\in A\\y\in B}}\frac{4}{\lvert x-y\rvert^\alpha}.
    \end{equation*}
\end{definition}

The following Lemma is a generalization of Lemma 3.5 in \cite{cluster}.

\begin{lemma}\label{lem: phi_abc_bound}
				Let $k\geq 2$ and $\{A_i\in\mathrm{P}(\mathbb{X}):i\in[k]\}$ be a family of not empty pairwise disjoint finite subsets. Let $\emptyset\neq B\Subset\mathbb{Z}$ such that $A_i\cap B=\emptyset$, for all $i\in [k]$. Then 
				\begin{equation*}\label{eq: phi_abc_bound}
					\prod_{i=1}^k\Phi(A_i,B)\leq 4^{2k-1} (\mathrm{diam}(B)+1)^{3k-2}\hspace{-.25cm}\sum_{\varphi\in \mathrm{Sym}(k)}\frac{1}{\mathrm{dist}( A_{\varphi(k)},B)^\alpha}\prod_{i=1}^{k-1}\Phi(A_{\varphi(i)},A_{\varphi(i+1)}).
				\end{equation*}
			\end{lemma}
			
			\begin{proof}
				Since $\alpha\in(1,2]$, if $x,y,z\in\mathbb{R}$, then 
				\begin{align}\label{eq: reverse-holder}
					\lvert x-y\rvert^{\alpha}&\leq (\lvert x-z\rvert+\lvert z-x\rvert)^\alpha\leq 2\lvert x-z\rvert^\alpha+2\lvert z-y\rvert^\alpha.
				\end{align}
				Write $\mathbf{A}\coloneqq\times_{i\in[k]}A_i$ and $\mathbf{B}=B^k$. Define $E:\mathrm{Sym}(k)\rightarrow\mathrm{P}(\mathbf{A}\times\mathbf{B})$ by 
				\begin{align*}
					E(\varphi)=\{(\mathbf{a},\mathbf{b})\in\mathbf{A}\times\mathbf{B}:i\leq j\in [k]\implies\lvert a_{\varphi(i)}-b_{\varphi(i)}\rvert\geq\lvert a_{\varphi(j)}-b_{\varphi(j)}\rvert\}.
				\end{align*} 
				Hence, if $(\mathbf{a},\mathbf{b})\in E(\varphi)$ and $i\in[k-1]$, then, by Equation~\eqref{eq: reverse-holder}, 
				\begin{equation*}
					\begin{split}
						\lvert a_{\varphi(i)}-a_{\varphi(i+1)}\rvert^\alpha&\leq(\lvert{a_{\varphi(i)}-b_{\varphi(i)}}\rvert+\lvert{b_{\varphi(i)}-b_{\varphi(i+1)}}\rvert+\lvert{b_{\varphi(i+1)}-a_{\varphi(i+1)}}\rvert)\\
						&\leq \left(2\lvert a_{\varphi(i)}-b_{\varphi(i)}\rvert+\mathrm{diam}(B)\right)^\alpha\\
						&\leq 8\lvert a_{\varphi(i)}-b_{\varphi(i)}\rvert^\alpha+2[1+\mathrm{diam}(B)]^\alpha\\
						&\leq 16[1+\diam(B)]^2\lvert{a_{\varphi(i)}-b_{\varphi(i)}}\rvert^\alpha,
					\end{split}
				\end{equation*}
				and thus
				\begin{align*}
					\prod_{i=1}^{k-1}\lvert a_{\varphi(i)}-a_{\varphi(i+1)}\rvert^\alpha\leq [4(1+\mathrm{diam}(B)]^{2(k-1)}\prod_{i=1}^{k-1}\lvert a_{\varphi(i)}-b_{\varphi(i)}\rvert^\alpha.
				\end{align*}
				Therefore, if $(\mathbf{a},\mathbf{b})\in E(\varphi)$, then 
				\begin{align*}
					\frac{\prod_{i=1}^{k-1}\lvert a_{\varphi(i)}-a_{\varphi(i+1)}\rvert^{\alpha}}{\prod_{i=1}^k\lvert a_i-b_i\rvert^{\alpha}}\leq \frac{[4(1+\mathrm{diam}(B))]^{2(k-1)}}{\lvert a_{\varphi(k)}-b_{\varphi(k)}\rvert^\alpha}.
				\end{align*}
				Since $\{E(\varphi):\varphi\in\mathrm{Sym}(k)\}$ covers $\mathbf{A}\times\mathbf{B}$, 
				\begin{equation*}
					\begin{split}
						\prod_{i\in[k]}\Phi(A_i,B)&\leq\sum_{\varphi\in\mathrm{Sym}(k)}\sum_{(\mathbf{a},\mathbf{b})\in E(\varphi)}\prod_{i\in[k]}\frac{4 }{\lvert{a_i-b_i}\rvert^\alpha}\\
						&\leq\sum_{\varphi\in\mathrm{Sym}(k)}\sum_{(\mathbf{a},\mathbf{b})\in E(\varphi)}\frac{[16 (1+\diam(B))^2]^k}{\lvert{a_{\varphi(k)}-b_{\varphi(k)}}\rvert^\alpha}\prod_{i\in[k-1]}\frac{1}{\vert{a_{\varphi(i)}-a_{\varphi(i+1)}}\rvert^\alpha}\\
						&\leq\sum_{\varphi\in\mathrm{Sym}(k)}\frac{[4(1+\diam(B))]^{2(k-1)}(4 \cdot\# B)^k}{\dist(A_{\varphi(k)},B)^\alpha}\sum_{\mathbf{a}\in \mathbf{A}}\prod_{i\in[k-1]}\frac{1}{\lvert{a_{\varphi(i)}-a_{\varphi(i+1)}}\rvert^\alpha}\\
						&\leq 4^{2k-1} (1+\diam(B))^{3k-2}\sum_{\varphi\in\mathrm{Sym}(k)}\frac{1}{\dist(A_{\varphi(k)},B)^\alpha}\prod_{i\in[k-1]}\Phi(A_{\varphi(i)},A_{\varphi(i+1)}).
					\end{split}
				\end{equation*}
				The lemma has thus been proved.
			\end{proof}
			
			\begin{corollary}\label{cor: reduction-operator-contours}
				If $k\geq2$, then there exists $\beta_0=\beta_0(\alpha,k)>0$ such that if $\beta\geq\beta_0$ and $\mathbf{w}=(\gamma_1,...,\gamma_k)\in\mathscr{C}_k$, then \begin{equation*}
					 \mathcal{R}^k_\beta(\mathbf{w})\leq \frac{4^{2k}\alpha}{\alpha-1}e^{-\frac{1}{2}c_2\beta}\sum_{\varphi\in\mathrm{Sym}(k)}v'_\beta\left(\gamma_{\varphi(k)}\right)\cdot\prod_{i\in[k-1]}\Phi\left(\gamma_{\varphi(i)},\gamma_{\varphi(i+1)}\right).
				\end{equation*}.
			\end{corollary}
			
			\begin{proof}
				Introducing translation-invariant equivalence classes as in the proof of Lemma~\ref{lem: leaf-cutter-operator}, 
				\begin{equation*}
					 \mathcal{R}^k_\beta(\mathbf{w})=\sum_{[\gamma]\in\overline{\mathscr{C}}}e^{-\beta H([\gamma])}\sum_{\gamma\in [\gamma]}\prod_{i=1}^k\Phi(\gamma_i,\gamma)\ind_{\gamma_i\parallel\gamma}.
				\end{equation*} 
				By Hypothesis~\ref{hyp: diameter-hypothesis}, if $\gamma\in[\gamma]$, then
				\[4^{2k-1}(\diam(\I_-(\gamma))+1)^{3k-2}\leq 4^{2k-1}\diam([\gamma])^{3k-2}\leq 4^{2k-1}e^{(3k-2)c_0H([\gamma])}.\]
				Therefore, by Lemma~\ref{lem: phi_abc_bound},
				\begin{equation*}
					\prod_{i=1}^k\Phi(\gamma_i,\gamma)\ind_{\gamma_i\parallel\gamma}\leq e^{(3k-2)c_0H([\gamma])}\sum_{\varphi\in \mathrm{Sym}(k)}\frac{4^{2k-1} \ind_{\gamma_{\varphi(k)}\parallel\gamma}}{\mathrm{dist}\left(\I_-(\gamma_{\varphi(k)}),\I_-(\gamma)\right)^\alpha}\prod_{j=1}^{k-1}\Phi\left(\gamma_{\varphi(j)},\gamma_{\varphi(j+1)}\right).
				\end{equation*}
				If $l\in[k]$ and $n\in\mathbb{N}$, then there are at most $2 \cdot\#\gamma_l$ elements of $[\gamma\hspace{0.025cm}]$ with $\mathrm{dist}(\I_-(\gamma_l),\I_-(\gamma))=n$. Hence, 
				\begin{equation*}
					\sum_{\gamma\in[\gamma]}\frac{ \ind_{\gamma_l\parallel\gamma}}{\mathrm{dist}(\I_-(\gamma_l),\I_-(\gamma))^\alpha}\leq \#\gamma_l\cdot 2 \sum_{n=1}^\infty\frac{1}{n^\alpha}\leq\frac{\alpha}{\alpha-1}H(\gamma_l).
				\end{equation*}
				Finally,
				\begin{align*}
					 \mathcal{R}^k_\beta(\mathbf{w})&\leq\frac{4^{2k-1}\alpha}{\alpha-1}\sum_{\varphi\in \mathrm{Sym}(k)}\prod_{j=1}^{k-1}\Phi(\gamma_{\varphi(j)},\gamma_{\varphi(j+1)})\sum_{[\gamma]\in\overline{\mathscr{C}}}e^{-[\beta-(3k-2)c_0]H([\gamma\hspace{.025cm}])} H\left(\gamma_{\varphi(k)}\right)\\
					&\leq \frac{4^{2k-1}\alpha}{\alpha-1}e^{-\frac{1}{2}c_2\beta}\sum_{\varphi\in \mathrm{Sym}(k)} H\left(\gamma_{\varphi(k)}\right)\prod_{j=1}^{k-1}\Phi\left(\gamma_{\varphi(j)},\gamma_{\varphi(j+1)}\right)\\
					&\leq \frac{4^{2k-1}\alpha}{\alpha-1}e^{-\frac{1}{2}c_2\beta}\sum_{\varphi\in\mathrm{Sym}(k)}v'_\beta\left(\gamma_{\varphi(k)}\right)\cdot\prod_{i\in[k-1]}\Phi\left(\gamma_{\varphi(i)},\gamma_{\varphi(i+1)}\right),
				\end{align*}
				where the last inequality is a consequence of Hypothesis~\ref{hyp: phase-transition-estimate} and $\beta\geq 2(3k-2)c_0$. The corollary has thus been proved.
			\end{proof}

\begin{corollary}\label{cor: two-contour-estimate}
    There exists $\beta_0>0$ such that if $\beta\geq\beta_0$ and $(\gamma_1,\gamma_2)\in\mathscr{C}_2$, then \begin{equation*}
        \W_\beta(\gamma_1,\gamma_2)\leq\Phi(\gamma_1,\gamma_2)e^{-\frac{1}{4}\beta [H(\gamma_1)+H(\gamma_2)]}.
    \end{equation*}
\end{corollary}

\begin{proof}
    This is a direct consequence of Corollaries~\ref{cor: reduction-operator-sums} and \ref{cor: reduction-operator-contours}.
\end{proof}

\begin{corollary}\label{cor: m-contour-estimate}
    Let $m\geq 3$. Then, there exists $\beta_0=\beta_0(\alpha,m)>0$ such that if $\beta\geq\beta_0$ and $\mathbf{w}\in\mathscr{C}_m$, then \begin{equation*}
        \W_\beta(\mathbf{w})\leq \sum_{T\in\Tr_m}\prod_{i=1}^me^{-\beta H(\gamma_i)/2^{m}}\hspace{-.25cm}\prod_{\{j,k\}\in E(T)}\hspace{-.5cm}\Phi(\gamma_j,\gamma_k)
    \end{equation*}
\end{corollary}

\begin{proof}
    This is a direct consequence of Theorem~\ref{theo: m-point-bound} and Corollary~\ref{cor: reduction-operator-contours}.
\end{proof}

\subsection{Convergence of the pressure}\label{subsec: convergence-of-the-pressure} 

In this subsection, we prove the last lemma necessary to obtain the absolute convergence of the pressure. It is equivalent to Proposition 5.3 in \cite{Cassandro.Merola.Picco.Rozikov.14} and Proposition 2.20 in \cite{cluster}. Given the myriad of contour estimates proved in the previous subsections, we translate polymer incompatibility in terms of contours. We finish it with the proof of Theorem~\ref{theo: convergence-cluster-expansion}.

\begin{definition}\label{def: incompatibility-distance-contours-polymers}
    We define \begin{equation*}\label{eq: incompatibility-distance-contours-polymers}
        \mathscr{A}\coloneqq\left\{(\gamma_1,\gamma_2)\in\mathscr{C}^2:\mathrm{dist}(\I_-(\gamma_1),\I_-(\gamma_2))\leq 2M\min\{\mathrm{diam}(\gamma_1),\mathrm{diam}(\gamma_2)\}^{3/2}\right\}.
    \end{equation*}
\end{definition}

\begin{lemma}\label{lem: incompatibility-distance-contours-polymers}
    Let $\Gamma_1,\Gamma_2\in\mathscr{P}$. If $\Gamma_1\not\sim\Gamma_2$, then there exists $(\gamma_1,\gamma_2)\in(\Gamma_1\times\Gamma_2)\cap\mathscr{A}$. 
\end{lemma}

\begin{proof}
    By Remark~\ref{rem: characteristic-function-polymer}, there are two cases to consider. First, if there exist $\gamma_1\in\Gamma_1$ and $\gamma_2\in\Gamma_2$ such that $\gamma_1\not\sim\gamma_2$, then \begin{equation*}
        \begin{split}
            \mathrm{dist}(\I_-(\gamma_1),\I_-(\gamma_2))&\leq 1+\mathrm{dist}(\gamma_1,\gamma_2)\leq 1+M\min\{\mathrm{diam}(\gamma_1),\mathrm{diam(\gamma_2)}\}^{3/2}\\&\leq 2M\min\{\mathrm{diam}(\gamma_1),\mathrm{diam(\gamma_2)}\}^{3/2}.
        \end{split}
    \end{equation*}That is, $(\gamma_1,\gamma_2)\in\mathscr{A}$. If $\Gamma_1\not\sim\Gamma_2$ and $(\gamma_1,\gamma_2)\in\Gamma_1\times\Gamma_2\Rightarrow \gamma_1\sim\gamma_2$, then there exist $\gamma_1\in\Gamma_1$ and $\gamma_2\in\Gamma_2$ such that $\I_-(\gamma_1)\cap\I_-(\gamma_2)\neq\emptyset$. That is, $(\gamma_1,\gamma_2)\in\mathscr{A}$. The lemma has thus been proved.
\end{proof}

\begin{lemma}\label{lem: notjoin-leaf-operator-satisfies-hypothesis}
    There exists $\beta_0>0$ such that if $\beta\geq\beta_0$ and $\Gamma_1\in\mathscr{P}$, then \begin{equation}\label{eq: notjoin-leaf-operator-satisfies-hypothesis}
        \overline{\mathcal{R}}^1_\beta(\Gamma_0)=\sum_{\Gamma\in\mathscr{P}}\ind_{\Gamma_0\not\sim\Gamma}\bar{z}_\beta(\Gamma)\leq Me^{-\frac{1}{8}c_4\beta}\sum_{\gamma_0\in\Gamma_0}H(\gamma_0).
    \end{equation}
\end{lemma}

\begin{proof}
				By Equation~\eqref{eq: bar-zeta-compared-to-W}, Lemma~\ref{lem: incompatibility-distance-contours-polymers} and Corollary~\ref{cor: polymers-reduced-to-contours},
				\[
					\begin{split}
						\overline{\mathcal{R}}^1_\beta(\Gamma_0)&=\sum_{\Gamma\in\mathscr{P}}\bar{z}_\beta(\Gamma)\ind_{\Gamma\not\sim\Gamma_0}\leq\sum_{\gamma_0\in\Gamma_0}\sum_{\gamma\in\mathscr{C}}\ind_{\mathscr{A}}(\gamma_0,\gamma)\sum_{\Gamma\in\mathscr{P}}\bar{z}_{\beta}(\Gamma)\ind_{\gamma\in\Gamma}\\
						&=\sum_{\gamma_0\in\Gamma_0}\sum_{\gamma\in\mathscr{C}}\ind_{\mathscr{A}}(\gamma_0,\gamma)\W_{\beta/2}(\gamma)\leq\sum_{\gamma_0\in\Gamma_0}\sum_{\gamma\in\mathscr{C}}\ind_{\mathscr{A}}(\gamma_0,\gamma)e^{-\frac{1}{4}\beta H(\gamma)}.
					\end{split}
				\]
				It is time to introduce translation-invariant equivalence classes. 
				
				Note that if $\gamma_0\in\Gamma_0$, then $\lvert\partial\I_-(\gamma_0)\rvert=\lvert\gamma_0\rvert$. Hence, if $[\gamma]\in\overline{\mathscr{C}}$, then there are at most $2M\diam([\gamma])^{3/2}\cdot\lvert\gamma_0\rvert$ elements $\gamma\in[\gamma]$ such that $\ind_{\mathscr{A}}(\gamma_0,\gamma)=1$. Hence, by Hypothesis~\ref{hyp: diameter-hypothesis} and choosing $\beta_0\geq 1+12 c_0$,
				\[
					\mathcal{R}_\beta(\Gamma_0)\leq\sum_{\gamma_0\in\Gamma_0}\lvert\gamma_0\rvert\cdot\sum_{[\gamma]\in\overline{\mathscr{C}}}2M\diam([\gamma])^{3/2}e^{-\frac{1}{4}\beta H([\gamma])}\leq 2M\sum_{\gamma_0\in\Gamma_0}\lvert\gamma_0\rvert\cdot\sum_{[\gamma]\in\overline{\mathscr{C}}}e^{-\frac{1}{8}\beta H([\gamma])}.
				\]
				Finally, by Hypothesis~\ref{hyp: phase-transition-estimate} and $H(\gamma_0)\geq 2\lvert\gamma_0\rvert$,
				\[
					\mathcal{R}_\beta(\Gamma_0)\leq Me^{-\frac{1}{8}c_4\beta}\sum_{\gamma_0\in\Gamma_0}H(\gamma_0).
				\]
				The lemma has just been proved.			
			\end{proof}

\begin{corollary}\label{cor: sum-tree-of-polymers}
    For all $\eta>0$, there exists $\beta_0=\beta_0(\eta)>0$ such that, if $\beta\geq\beta_0$ and $\varGamma\in\mathscr{P}$, then \begin{equation*}
        \overline{\W}_\beta(\varGamma)\leq \eta\bar{z}_{\beta/2}(\varGamma). 
    \end{equation*}
\end{corollary}

\begin{proof}
    This is a direct consequence of Lemmas~\ref{lem: reducing-sums-over-trees} and \ref{lem: notjoin-leaf-operator-satisfies-hypothesis}.
\end{proof}

\begin{proof}[Proof of Theorem~\ref{theo: convergence-cluster-expansion}]
    First, let us deal with the matter of convergence in the absence of an external field. By Equation~\eqref{eq: pressure-tree-function}, it is sufficient to apply the machinery developed in Section~\ref{sec: sums-over-trees}. More precisely, with Lemma~\ref{lem: notjoin-leaf-operator-satisfies-hypothesis}, we transform the sum over trees of polymers into a sum over trees of contours. Once again, we apply the machinery developed in Section~\ref{sec: sums-over-trees}; for all the details, see Section~\ref{subsec: sums-of-trees-of-contours}. That is, \begin{align*}
        \lvert \Lambda\rvert P_\beta(\Lambda)\leq\sum_{\Gamma\in\mathscr{P}_\Lambda}\overline{\W}_\beta(\Gamma)\leq\sum_{\Gamma\in\mathscr{P}_\Lambda}\bar{z}_{\beta/2}(\Gamma)\leq\sum_{\gamma\in\mathscr{C}_\Lambda}\W_{\beta/4}(\gamma)\leq \sum_{\gamma\in\mathscr{C}_\Lambda}e^{-\frac{1}{8}\beta H(\gamma)}.
    \end{align*} The first inequality is a consequence of Equation~\eqref{eq: pressure-tree-function}. The second inequality is a consequence of Corollary~\ref{cor: sum-tree-of-polymers} ($\eta=1$). The third inequality is a consequence of Equation~\eqref{eq: bar-zeta-compared-to-W}. Finally, the fourth inequality is a consequence of Corollary~\ref{cor: polymers-reduced-to-contours}.

    Let $[\gamma\hspace{+0.05cm}]$ be an equivalence class of contours associated to the action of $\mathbb{Z}$ by translation. There are at most $\lvert\Lambda\rvert$ elements of $[\gamma\hspace{+0.05cm}]$ in $\mathscr{C}_\Lambda$. Hence, \begin{align*}
        P_\beta(\Lambda)\leq \sum_{[\gamma\hspace{+0.025cm}]\in\overline{\mathscr{C}}}e^{-\frac{1}{8}\beta H([\gamma\hspace{+0.025cm}])}\leq e^{-\frac{1}{8}c_2\beta},
    \end{align*}where the second inequality is a consequence of Hypothesis~\ref{hyp: phase-transition-estimate} (recall Equation~\eqref{eq: hypothesis-3-modified}). The first part of Theorem~\ref{theo: convergence-cluster-expansion} has thus been proved.

    The second part of Theorem~\ref{theo: convergence-cluster-expansion} is a direct application of Lemma~\ref{lem: upper-activity}. We thus conclude the proof of Theorem~\ref{theo: convergence-cluster-expansion}.
\end{proof}

\section{Lattice site estimates}\label{sec: point-estimates}

In this section, we establish some estimates involving lattice sites necessary to compute correlations of $n$-points. This is done because we transform polymers into contours, and then contours into lattice sites. The approach developed in this section can be seen as an extension of the one employed by Imbrie in \cite{Imbrie.82} to estimate the two-point correlations for the exponent $\alpha=2$. We begin with an auxiliary lemma.

\begin{lemma}\label{lem: permutation-trick}
    Let $\mathbb{L}$ be a set and $C:\mathbb{L}^2\rightarrow\mathbb{R}_{>0}$ be a function such that $C(a,a)=1$ for all $a\in\mathbb{L}$. If $k\in\mathbb{N}$, $\mathbf{x}=(x_1,...,x_k)\in\mathbb{L}^k$ and $y\in\mathbb{L}$, then  \begin{equation*}
        \sum_{\varphi\in\mathrm{Sym}(l)}C(\bar{x}_{\varphi(l)},y)\prod_{i=1}^{l-1}C(\bar{x}_{\varphi(i)},\bar{x}_{\varphi(i+1)})\leq \sum_{\varphi\in\mathrm{Sym}(k)}C(x_{\varphi(k)},y)\prod_{j=1}^{k-1}C(x_{\varphi(j)},x_{\varphi(j+1)}),
    \end{equation*}where all the entries in $\bar{\mathbf{x}}=(\bar{x}_1,...,\bar{x}_l)$ are distinct and $\{\bar{x}_i:i\in[l]\}=\{x_j:j\in [k]\}$.
\end{lemma}

\begin{proof}
    Let $\mathbf{x}$ and $\bar{\mathbf{x}}$ satisfy the hypotheses of the lemma. We shall construct an injection from $\mathrm{Sym}(l)$ to $\mathrm{Sym}(k)$, denoted by $\varphi\mapsto\tilde\varphi$, such that $\bar{x}_{\varphi(l)}=x_{\tilde{\varphi}(k)}$ and \begin{equation}\label{eq: C-x-x-C-bar-x-bar-x}
        \prod_{i=1}^{l-1}C(\bar{x}_{\varphi(i)},\bar{x}_{\varphi(i+1)})=\prod_{j=1}^{k-1}C(x_{\tilde{\varphi}(j)},x_{\tilde{\varphi}(j+1)}).
    \end{equation}Since $C$ is positive, the existence of such an injection proves the lemma.

    For $i\in[l]$, define $K_i\coloneqq\{j\in[k]:x_j=\bar{x}_i\}$. 
    
    Denote by $\phi_i:[\lvert K_i\rvert]\rightarrow K_i$ the bijection that orders $K_i$; that is, $j<j'\implies \phi_i(j)<\phi_i(j')$. 
    
    Let $\varphi\in\mathrm{Sym}(l)$. It is convenient to write \[k[\varphi](i)\coloneqq\sum_{p=1}^i\lvert K_{\varphi(p)}\rvert,\] with the convention $k[\varphi](0)=0$. Note that $k[\varphi](l)=k$ and that $i\mapsto k[\varphi](i)$ is strictly increasing. In order to define $\tilde{\varphi}$, we proceed in blocks: if $j\in[k]$, then there exists a unique $i(j)\in[l]$ such that $k[\varphi](i(j)-1)<j\leq k[\varphi](i(j))$. Define \begin{equation*}
        \tilde{\varphi}(j)\coloneqq \phi_{\varphi(i(j))}\left(j-k[\varphi](i(j))\right).
    \end{equation*}While this may seem rather involved, the graphical description is rather straightforward, see Figure~\ref{fig: permutation-trick}.

    \begin{figure}[hbt!]
        \centering
        \tikzset{every picture/.style={line width=0.75pt}} 

\begin{tikzpicture}[x=0.75pt,y=0.75pt,yscale=-1,xscale=1]

\draw   (120,20) -- (135,20) -- (135,35) -- (120,35) -- cycle ;
\draw   (135,20) -- (150,20) -- (150,35) -- (135,35) -- cycle ;
\draw   (150,20) -- (165,20) -- (165,35) -- (150,35) -- cycle ;
\draw   (165,20) -- (180,20) -- (180,35) -- (165,35) -- cycle ;
\draw   (180,20) -- (195,20) -- (195,35) -- (180,35) -- cycle ;
\draw   (195,20) -- (210,20) -- (210,35) -- (195,35) -- cycle ;
\draw   (210,20) -- (225,20) -- (225,35) -- (210,35) -- cycle ;
\draw   (225,20) -- (240,20) -- (240,35) -- (225,35) -- cycle ;
\draw   (40,20) -- (55,20) -- (55,35) -- (40,35) -- cycle ;
\draw   (55,20) -- (70,20) -- (70,35) -- (55,35) -- cycle ;
\draw   (70,20) -- (85,20) -- (85,35) -- (70,35) -- cycle ;
\draw   (40,65) -- (55,65) -- (55,80) -- (40,80) -- cycle ;
\draw   (55,65) -- (70,65) -- (70,80) -- (55,80) -- cycle ;
\draw   (70,65) -- (85,65) -- (85,80) -- (70,80) -- cycle ;
\draw   (120,65) -- (135,65) -- (135,80) -- (120,80) -- cycle ;
\draw   (135,65) -- (150,65) -- (150,80) -- (135,80) -- cycle ;
\draw   (150,65) -- (165,65) -- (165,80) -- (150,80) -- cycle ;
\draw   (165,65) -- (180,65) -- (180,80) -- (165,80) -- cycle ;
\draw   (180,65) -- (195,65) -- (195,80) -- (180,80) -- cycle ;
\draw   (195,65) -- (210,65) -- (210,80) -- (195,80) -- cycle ;
\draw   (210,65) -- (225,65) -- (225,80) -- (210,80) -- cycle ;
\draw   (225,65) -- (240,65) -- (240,80) -- (225,80) -- cycle ;
\draw   (40,40) -- (85,40) -- (62.5,60) -- cycle ;
\draw   (120,40) -- (240,40) -- (180,60) -- cycle ;
\draw  [dash pattern={on 0.84pt off 2.51pt}] (30,10) -- (250,10) -- (250,90) -- (30,90) -- cycle ;

\draw (127.5,27.5) node  [font=\small]  {$a$};
\draw (142.5,27.5) node  [font=\small]  {$b$};
\draw (157.5,27.5) node  [font=\small]  {$c$};
\draw (172.5,27.5) node  [font=\small]  {$a$};
\draw (187.5,27.5) node  [font=\small]  {$a$};
\draw (202.5,27.5) node  [font=\small]  {$b$};
\draw (217.5,27.5) node  [font=\small]  {$b$};
\draw (232.5,27.5) node  [font=\small]  {$c$};
\draw (62.5,27.5) node  [font=\small]  {$a$};
\draw (77.5,27.5) node  [font=\small]  {$b$};
\draw (47.5,27.5) node  [font=\small]  {$c$};
\draw (62.5,72.5) node  [font=\small]  {$a$};
\draw (77.5,72.5) node  [font=\small]  {$b$};
\draw (47.5,72.5) node  [font=\small]  {$c$};
\draw (157.5,72.5) node  [font=\small]  {$a$};
\draw (217.5,72.5) node  [font=\small]  {$b$};
\draw (127.5,72.5) node  [font=\small]  {$c$};
\draw (187.5,72.5) node  [font=\small]  {$a$};
\draw (172.5,72.5) node  [font=\small]  {$a$};
\draw (202.5,72.5) node  [font=\small]  {$b$};
\draw (232.5,72.5) node  [font=\small]  {$b$};
\draw (142.5,72.5) node  [font=\small]  {$c$};
\draw (62.5,47) node  [font=\scriptsize]  {$\varphi $};
\draw (180,47) node  [font=\scriptsize]  {$\tilde{\varphi }$};

\end{tikzpicture}
        \caption{Here, $\mathbf{x}=(a,b,c,a,a,b,b,c)$, $\bar{\mathbf{x}}=(c,a,b)$ and $\varphi\in\mathrm{Sym(3)}$ denotes the trivial permutation.}
        \label{fig: permutation-trick}
    \end{figure}

    Let us first show that $\varphi\mapsto\tilde\varphi$ is an injection. If $\varphi_1\neq\varphi_2\in\mathrm{Sym}(l)$, then there exists a \emph{smallest} $i\in[l]$ such that $\varphi_1(i)\neq\varphi_2(i)$. Thus, if $j=k[\varphi_1](i-1)+1=k[\varphi_2](i-1)+1$, then $\tilde{\varphi}_1(j)\in K_{\varphi_1(i)}$ and $\tilde\varphi_2(j)\in K_{\varphi_2(i)}$. Since $K_{\varphi_1(i)}\cap K_{\varphi_2(i)}=\emptyset$, we conclude that $\varphi\mapsto\tilde\varphi$ is an injection.

    Finally, let us now show that Equation~\eqref{eq: C-x-x-C-bar-x-bar-x} holds. If $i\in[l]$, then \begin{equation*}
        \prod_{j=k[\varphi](i-1)+1}^{k[\varphi](i)-1}C(x_{\tilde{\varphi}(j)},x_{\tilde{\varphi}(j+1)})=\prod_{j=k[\varphi](i-1)+1}^{k[\varphi](i)-1}C(\bar{x}_{\varphi(i)},\bar{x}_{\varphi(i)})=1.
    \end{equation*}Therefore, we obtain \begin{equation*}
        \prod_{j=1}^{k-1}C(x_{\tilde\varphi(j)},x_{\tilde\varphi(j+1)})=\prod_{i=1}^{l-1}C(x_{\tilde{\varphi}(k[\varphi](i))},x_{\tilde{\varphi}(k[\varphi](i)+1)})=\prod_{i=1}^{l-1}C(\bar{x}_{\varphi(i)},\bar{x}_{\varphi(i+1)}).
    \end{equation*}The reader may readily check that $\bar{x}_{\varphi(l)}=x_{\tilde{\varphi}(k)}$. The lemma has thus been proved.    
\end{proof}

Of course, the interest is not an arbitrary function $C$. Following Imbrie in \cite{Imbrie.82}, we define\begin{equation*}
    C(x,y)=\begin{cases}
        \lvert x-y\rvert^{-\alpha}, &\textrm{if }x\neq y\\
        1, &\textrm{if }x=y.
    \end{cases}
\end{equation*}In general, the natural definition is $C(x,y)=\ind_{x=y}+\ind_{x\neq y}J_{x,y}$.

Let us begin by proving an analogous result of Lemma~\ref{lem: phi_abc_bound} for lattice sites.

\begin{lemma}\label{lem: C-xyz-bound}
    Let $k\geq 2$. Let $\mathbf{x}=(x_1,...,x_k)\in\mathbb{Z}^k$ and $y\in\mathbb{Z}$, then \begin{equation}\label{eq: C-xyz-bound}
        \prod_{i=1}^kC(x_i,y)\leq 4^{k}\sum_{\varphi\in \mathrm{Sym}(k)}C(x_{\varphi(k)},y)\prod_{i=1}^{k-1}C(x_{\varphi(i)},x_{\varphi(i+1)}).
    \end{equation}In particular, we have \begin{equation}\label{eq: triple-c-estimate}
        \sum_{y}\prod_{i=1}^kC(x_i,y)\leq 4^{k}(1+2\zeta(\alpha))\sum_{\varphi\in \mathrm{Sym}(k)}\prod_{i=1}^{k-1}C(x_{\varphi(i)},x_{\varphi(i+1)}).
    \end{equation}
\end{lemma}

\begin{remark}
    If $\{x_i\}$ and $y$ are all distinct, this is simply Lemma~\ref{lem: phi_abc_bound}. 
\end{remark}

\begin{proof}
    Let $\bar{\mathbf{x}}=(\bar{x}_1,...,\bar{x}_l)$ such that $i\neq j\Rightarrow\bar{x}_i\neq\bar{x}_j$ and $\{x_i:i\in[k]\}=\{\bar{x}_i:i\in[l]\}$. Then \begin{align*}
        \prod_{i=1}^kC(x_i,y)=\prod_{j=1}^lC(\bar{x}_j,y)^{k_j},
    \end{align*}where $k_j=\lvert\{i\in[k]:x_i=\bar{x}_j\}\rvert$. Hence, \begin{equation*}
        \prod_{i=1}^kC(x_i,y)\leq\prod_{j=1}^lC(\bar{x}_j,y)=\prod_{j=1}^l\frac{\ind_{\bar{x}_j\neq y}}{\lvert \bar{x}_j-y\rvert^\alpha}+\sum_{i=1}^l\ind_{\bar{x}_i=y}\prod_{j\neq i}\frac{\ind_{\bar{x}_j\neq y }}{\lvert y-\bar{x}_j\rvert^\alpha}.
    \end{equation*}By Lemma~\ref{lem: phi_abc_bound}, we bound the first term by \begin{equation*}
        \prod_{j=1}^l\frac{\ind_{\bar{x}_j\neq y}}{\lvert \bar{x}_j-y\rvert^\alpha}\leq 4^{l}\prod_{i=1}^l\ind_{\bar{x}_i\neq y}\sum_{\varphi\in \mathrm{Sym}(l)}C(\bar{x}_{\varphi(l)},y)\prod_{j=1}^{l-1}C(\bar{x}_{\varphi(j)},\bar{x}_{\varphi(j+1)}).
    \end{equation*}Let us now bound the second term. For $i\in[l]$, define $\mathbf{z}^i=(\bar{x}_1,...,\bar{x}_{i-1},\bar{x}_{i+1},...,\bar{x}_l)$. Thus, the same Lemma~\ref{lem: phi_abc_bound} implies that \begin{equation*}
        \ind_{\bar{x}_i=y}\prod_{j=1}^{l-1}\frac{\ind_{z^i_j\neq y}}{\lvert z^i_j-y\rvert^{\alpha}}\leq 4^{l-1}\ind_{\bar{x}_i=y}\prod_{j\neq i}\ind_{\bar{x}_j\neq y}\hspace{-.25cm}\sum_{\varphi\in\mathrm{Sym}(l-1)}\prod_{j=1}^{l-2}C(z^{i}_{\varphi(j)},z^i_{\varphi(j+1)})\cdot C(z^i_{\varphi(l-1)},\bar{x}_i).
    \end{equation*}There is a canonical bijection $\mathrm{Sym}(l-1)\ni\varphi\leftrightarrow\tilde{\varphi}\in\mathrm{Sym}(l)$ given by \begin{equation*}
        \tilde{\varphi}(j)\coloneqq\begin{cases}
            i, &\textrm{if }j=l;\\
            \varphi(j), &\textrm{if }\varphi(j)<i;\\
            \varphi(j)+1, &\textrm{if }\varphi(j)\geq i.
        \end{cases}
    \end{equation*}Recalling that $\ind_{\bar{x}_i=y}=\ind_{\bar{x}_i=y}C(\bar{x}_i,y)$, the right-hand side can be written as\begin{equation*}
        4^{l-1}\ind_{\bar{x}_i=y}\prod_{j\neq i}\ind_{\bar{x}_i\neq\bar{x}_j}\sum_{i=1}^l\sum_{\substack{\varphi\in\mathrm{Sym}(l)\\\varphi(l)=i}}C(\bar{x}_{\varphi(l)},y)\prod_{j=1}^{l-1}C(\bar{x}_{\varphi(j)},\bar{x}_{\varphi(j+1)}).
    \end{equation*}Finally, since \[\prod_{i=1}^l\ind_{\bar{x}_i\neq y}+\sum_{i=1}^l\ind_{\bar{x}_i=y}\prod_{j\neq i}\ind_{\bar{x}_i\neq \bar{x}_j}\leq 1,\]we obtain that\begin{equation*}
        \prod_{i=1}^kC(x_i,y)\leq 4^{l}\sum_{\varphi\in\mathrm{Sym}(l)}C(\bar{x}_{\varphi(l)},y)\prod_{i=1}^{l-1} C(\bar{x}_{\varphi(i)},\bar{x}_{\varphi(i+1)}).
    \end{equation*}By Lemma~\ref{lem: permutation-trick}, the proof is finished.
\end{proof}

Consider the family of locally compatible sums of trees of lattice sites associated to the triple
\[
    \left(\mathbb{Z},\left\{v_\beta(\cdot)\equiv e^{-\beta}\right\},C(\cdot,\cdot)\right).
\]
Iterating Equation~\eqref{eq: triple-c-estimate}, if $n\geq2$ and $x_1,x_2\in\mathbb{Z}$, then
\begin{equation}\label{eq: total-edge-lattice-sites}
    \sum_{\mathbf{y}\in\mathbb{Z}^n}\ind_{y_1=x_1}\ind_{y_n=x_2}\prod_{i\in[n-1]}C(y_i,y_{i+1})\leq [32(1+2\zeta(\alpha))]^{(n-2)}C(x_1,x_2).
\end{equation}
Although we are dealing with locally compatible sums of trees of lattice sites, Lemma~\ref{lem: C-xyz-bound} and Equation~\eqref{eq: total-edge-lattice-sites} imply that the results for globally compatible trees are true in this case. Hence, if $m\in\mathbb{N}$, then there exists $\beta_0$ such that if $\mathbf{x}\in\mathbb{Z}^m$ and $\beta\geq\beta_0$, then
\begin{equation}
    \sum_{n=m}^{\infty}\frac{e^{-n\beta}}{(n-m)!}\sum_{T\in\Tr_n}\sum_{\mathbf{y}\in\mathbb{Z}^n}\ind_{\mathbf{y}|_{[m]}=\mathbf{w}}\prod_{\{j,k\}\in E(T)}C(y_j,y_k)\leq e^{-\frac{m\beta}{2^{m}}}\sum_{T\in\Tr_m}\prod_{\{j,k\}\in E(T)}C(x_j,x_k).
\end{equation}

Let us now prove the estimates that allow us to translate estimates in terms of contours into estimates in terms of lattice sites.

\begin{lemma}\label{lem: contour-to-point-bound}
    Let $x,y\in\mathbb{Z}$. 
    \begin{enumerate}
        \item If $\gamma\in\mathscr{C}$ and $x,y\in\I_-(\gamma)$, then $1\leq C(x,y)e^{2c_0 H(\gamma)}$.
        \item If $(\gamma_1,\gamma_2)\in\mathscr{C}_2$, $x\in\I_-(\gamma_1)$ and $y\in\I_-(\gamma_2)$, then $\Phi(\gamma_1,\gamma_2)\leq 2^6\cdot e^{3 c_0[H(\gamma_1)+H(\gamma_2)]}C(x,y)$.
        \item If $(\gamma_1,\gamma_2)\in\mathscr{A}$, $x\in\I_-(\gamma_1)$ and $y\in\I_-(\gamma_2)$, then $1\leq (4M)^2 e^{2 c_0[H(\gamma_1)+H(\gamma_2)]}C(x,y).$
    \end{enumerate}
\end{lemma}

\begin{proof}
    The first bound is immediate from Hypothesis~\ref{hyp: diameter-hypothesis}. In general, we have \begin{equation}\label{eq: trivial-phi-bound}
        \Phi(\gamma_1,\gamma_2)\leq\frac{4\mathrm{diam}(\gamma_1)\mathrm{diam}(\gamma_2)}{\mathrm{dist}(\gamma_1,\gamma_2)^\alpha}.
    \end{equation}On the other hand, by the triangular inequality, we have \begin{equation*}
        \lvert x-y\rvert\leq \mathrm{diam}(\gamma_1)+\mathrm{dist}(\gamma_1,\gamma_2)+\mathrm{diam}(\gamma_2).
    \end{equation*}Dividing and multiplying by $2$ and using $x,y\geq 2\Rightarrow x+y\leq xy$, we arrive at \begin{equation}\label{eq: crude-x-y-bound}
        \lvert x-y\rvert^\alpha\leq 16\cdot\mathrm{dist}(\gamma_1,\gamma_2)^\alpha\mathrm{diam}(\gamma_1)^2\mathrm{diam}(\gamma_2)^2.
    \end{equation} Combining Hypothesis~\ref{hyp: diameter-hypothesis} with Equations~\eqref{eq: trivial-phi-bound} and \eqref{eq: crude-x-y-bound}, we obtain the second bound. 
    
    Let us prove the third bound. Let $(\gamma_1,\gamma_2)\in\mathscr{A}$, $x_1\in\I_-(\gamma_1)$ and $x_2\in\I_-(\gamma_2)$. Then \begin{equation*}
        \lvert x-y\rvert\leq \mathrm{diam}(\gamma_1)+\mathrm{dist}(\I_-(\gamma_1),\I_-(\gamma_2))+\mathrm{diam}(\gamma_2).
    \end{equation*}By the definition of $\mathscr{A}$, \begin{equation*}
        \mathrm{dist}(\I_-(\gamma_1),\I_-(\gamma_2))\leq 2M\mathrm{diam}(\gamma_1)^{3/4}\mathrm{diam}(\gamma_2)^{3/4}.
    \end{equation*} Since $M\mathrm{diam}(\gamma_1)^{3/4}\mathrm{diam}(\gamma_2)^{3/4}+\mathrm{diam}(\gamma_1)\leq 2M\mathrm{diam}(\gamma_1)\mathrm{diam}(\gamma_2)^{3/4}$,  \begin{equation*}
        \begin{split}
            \lvert x-y\rvert\leq M\mathrm{diam}(\gamma_1)^{3/4}\mathrm{diam}(\gamma_2)^{3/4}(2\mathrm{diam}(\gamma_1)^{1/4}+2\mathrm{diam}(\gamma_2)^{1/4})\leq 4M\diam(\gamma_1)\diam(\gamma_2),
        \end{split}
    \end{equation*}which combined with Hypothesis~\ref{hyp: diameter-hypothesis} yields the second bound. The lemma has thus been proved.
\end{proof}

All the ingredients needed to give an upper bound to the decay rate of the $n$-point correlation functions have now been described.

\section{Correlations of n-points}\label{sec: correlations}

In this section, we apply the cluster expansion developed in Sections~\ref{sec: polymers} and \ref{sec: convergence} and the site estimates proved in Section~\ref{sec: point-estimates} to obtain an upper bound to the $n$-point correlations of the long-range one-dimensional ferromagnetic Ising model with a homogeneous boundary condition. That is, we are interested in estimating \begin{equation}\label{eq: n-point-correlations}
    \langle :\sigma_A:\rangle^+_{\beta}\coloneqq\lim_{\Lambda\uparrow\mathbb{Z}}\frac{1}{(-2\beta)^{\lvert A\rvert}}\frac{\partial}{\partial h_A}|_{h=0}\log Z^+_{\Lambda,\beta,h},
\end{equation}where $A$ is an arbitrary subset of $\Lambda$.
Note that if $A=\{x,y\}$, then
\[
\langle:\sigma_A:\rangle^+_\beta=\langle \sigma_x\sigma_y\rangle^+_\beta-\langle\sigma_x\rangle^+_\beta\langle\sigma_y\rangle^+_\beta.
\]
Let $\mathbf{A}\in \Pt(A,p)\coloneqq\{\mathbf{A}\in \mathrm{P}(A)^p: \{A_i:i\in[p]\}\in\Pt(P)\}$, i.e., the set of labeled partitions of $A$ with $p\geq 2$ elements. Inspired by Lemma~\ref{lem: many-point-arbitrary-sums}, define
\begin{equation}
    \overline{\W}^*_\beta(\mathbf{A})\coloneqq\sum_{n=p}^{2p-2}\frac{1}{(n-p)!}\sum_{T\in\Tr^*_{n,p}}\sum_{\mathbf{\Gamma}\in\mathscr{P}^n}\prod_{q\in[p]}\ind_{A_q\subset\I_-(\Gamma_q)}\prod_{i\in[n]}\overline{\W}_\beta(\Gamma_i)\prod_{\{j,k\}\in E(T)}\overline{\E}_\beta(\Gamma_j,\Gamma_k),
\end{equation}
where I recall that
\begin{equation}\label{eq: total-edge-function-for-polymers}
    \overline{\E}_\beta(\varGamma_1,\varGamma_2)\coloneqq\sum_{n=2}^\infty\sum_{\mathbf{\Gamma}\in\mathscr{P}^n}\ind_{\Gamma_1=\varGamma_1}\ind_{\Gamma_n=\varGamma_2}\prod_{i=2}^{n-1}\overline{\W}_\beta(\Gamma_j)\prod_{j\in[n-1]}\ind_{\Gamma_j\not\sim\Gamma_{j+1}}.
\end{equation}
Thus, by Lemma~\ref{lem: many-point-arbitrary-sums}, Theorem~\ref{theo: convergence-cluster-expansion} and Cauchy's estimates for holomorphic functions, 
\begin{equation}
\begin{split}
    \frac{\lvert\langle :\sigma_A:\rangle^+_{\beta}\rvert}{(6\lvert A\rvert)^{\lvert A\rvert}}&\leq
    \sum_{\Gamma\in\mathscr{P}}\ind_{A\subset\I_-(\Gamma)}\overline{\W}_\beta(\Gamma)+\sum_{p=2}^{\lvert A\rvert}\frac{1}{p!}\sum_{\mathbf{A}\in\Pt(A,p)}\overline{\W}^*_\beta(\mathbf{A}).
\end{split}
\end{equation}
We restrict ourselves to two-point functions. The computations for three(+)-point functions are quite complicated and not particularly novel technique-wise. The interested reader will be able to find a detailed account in Corsini's PhD thesis \cite{Corsini-tese}.

Henceforth, assume that $A=\{x,y\}$. Thus,
\begin{equation}
    \frac{1}{144}\left\lvert\langle:\sigma_A:\rangle^+_\beta\right\rvert\leq\sum_{\Gamma\in\mathscr{P}}\ind_{x,y\in\I_-(\Gamma)}\overline{\W}_\beta(\Gamma)+\sum_{\mathbf{\Gamma}\in\mathscr{P}^2}\ind_{x\in\I_-(\Gamma_1)}\ind_{y\in\I_-(\Gamma_2)}\overline{\W}_\beta(\Gamma_1)\overline{\E}_\beta(\Gamma_1,\Gamma_2)\overline{\W}_\beta(\Gamma_2). 
\end{equation}
Let us begin with the first term. By Corollary~\ref{cor: sum-tree-of-polymers} and Equation~\eqref{eq: bar-zeta-compared-to-W},
\[
    \sum_{\Gamma\in\mathscr{P}}\ind_{x,y\in\I_-(\Gamma)}\overline{\W}_\beta(\Gamma)\leq\sum_{\gamma\in\mathscr{C}}\ind_{x,y\in\I_-(\gamma)}\W_{\beta/4}(\gamma)+\sum_{\mathbf{w}\in\mathscr{C}_2}\ind_{x\in\I_-(\gamma_1)}\ind_{y\in\I_-(\gamma_2)}\W_{\beta/4}(\mathbf{w}).
\]
By Hypothesis~\ref{hyp: diameter-hypothesis}, there exists $c_6>0$ such that
\[
    \sum_{\gamma\in\mathscr{C}}\ind_{x,y\in\I_-(\gamma)}e^{-\beta H(\gamma)}\leq 2^{-c_6\beta}\lvert x-y\rvert^{-c_6\beta }\leq 2^{-c_6\beta} C(x,y),
\]
provided $\beta\geq\beta_0$ is sufficiently large. The reader will find the proof of such statement in \cite{Corsini-tese}, Section 5.2.2. On the other hand, by Corollary~\ref{cor: two-contour-estimate}, Lemma~\ref{lem: contour-to-point-bound} and Hypothesis~\ref{hyp: phase-transition-estimate},
\[
    \sum_{\mathbf{w}\in\mathscr{C}_2}\ind_{x\in\I_-(\gamma_1)}\ind_{y\in\I_-(\gamma_2)}\W_{\beta/4}(\mathbf{w})\leq 16\cdot\left(e^{-\left(\frac{1}{16}\beta-3c_0\right)c_2}\right)^{2}C(x,y)\leq e^{-\frac{1}{16}c_2\beta }C(x,y).
\]
Thus, the final piece of the puzzle is simply bounding $\overline{\E}_\beta
(\cdot,\cdot)$, which is the object of the next subsection.

\subsection{Total edge function bounds}\label{subsec: total-edge-function-bounds}

Let us now bound the total edge function for polymer trees (see Equation~\eqref{eq: total-edge-function-for-polymers}). Let $\Gamma_1,\Gamma_2,\Gamma_3$ be polymers such that $\Gamma_1\not\sim\Gamma_2$ and $\Gamma_2\not\sim\Gamma_3$. By Lemma~\ref{lem: incompatibility-distance-contours-polymers}, there exist $(\gamma_{1,2},\gamma_{2,1})\in(\Gamma_1\times\Gamma_2)\cap\mathscr{A}$ and $(\gamma_{2,3},\gamma_{3,2})\in(\Gamma_2\times\Gamma_3)\cap\mathscr{A}$. There are two cases to consider: $\gamma_{2,1}=\gamma_{2,3}$ or $\gamma_{2,1}\parallel\gamma_{2,3}$, that is, $(\gamma_{2,1},\gamma_{2,3})\in\mathscr{C}_2$. For a graphical description, see Figure~\ref{fig: type-of-incompatibility-edges}.

\begin{figure}[hbt!]
    \centering
    \tikzset{every picture/.style={line width=0.75pt}} 

\begin{tikzpicture}[x=0.75pt,y=0.75pt,yscale=-1,xscale=1]

\draw   (70,67) .. controls (70,58.72) and (76.72,52) .. (85,52) .. controls (93.28,52) and (100,58.72) .. (100,67) .. controls (100,75.28) and (93.28,82) .. (85,82) .. controls (76.72,82) and (70,75.28) .. (70,67) -- cycle ;

\draw   (130,67) .. controls (130,58.72) and (136.72,52) .. (145,52) .. controls (153.28,52) and (160,58.72) .. (160,67) .. controls (160,75.28) and (153.28,82) .. (145,82) .. controls (136.72,82) and (130,75.28) .. (130,67) -- cycle ;

\draw   (190,67) .. controls (190,58.72) and (196.72,52) .. (205,52) .. controls (213.28,52) and (220,58.72) .. (220,67) .. controls (220,75.28) and (213.28,82) .. (205,82) .. controls (196.72,82) and (190,75.28) .. (190,67) -- cycle ;

\draw    (100,67) -- (130,67) ;
\draw    (160,67) -- (190,67) ;
\draw   (40,145) .. controls (40,136.72) and (46.72,130) .. (55,130) .. controls (63.28,130) and (70,136.72) .. (70,145) .. controls (70,153.28) and (63.28,160) .. (55,160) .. controls (46.72,160) and (40,153.28) .. (40,145) -- cycle ;

\draw   (100,145) .. controls (100,136.72) and (106.72,130) .. (115,130) .. controls (123.28,130) and (130,136.72) .. (130,145) .. controls (130,153.28) and (123.28,160) .. (115,160) .. controls (106.72,160) and (100,153.28) .. (100,145) -- cycle ;
\draw   (160,145) .. controls (160,136.72) and (166.72,130) .. (175,130) .. controls (183.28,130) and (190,136.72) .. (190,145) .. controls (190,153.28) and (183.28,160) .. (175,160) .. controls (166.72,160) and (160,153.28) .. (160,145) -- cycle ;
\draw   (220,145) .. controls (220,136.72) and (226.72,130) .. (235,130) .. controls (243.28,130) and (250,136.72) .. (250,145) .. controls (250,153.28) and (243.28,160) .. (235,160) .. controls (226.72,160) and (220,153.28) .. (220,145) -- cycle ;
\draw    (70,145) -- (100,145) ;
\draw    (130,145) -- (160,145) ;
\draw    (190,145) -- (220,145) ;
\draw   (40,204.5) .. controls (40,196.22) and (46.72,189.5) .. (55,189.5) .. controls (63.28,189.5) and (70,196.22) .. (70,204.5) .. controls (70,212.78) and (63.28,219.5) .. (55,219.5) .. controls (46.72,219.5) and (40,212.78) .. (40,204.5) -- cycle ;

\draw   (100,204.5) .. controls (100,196.22) and (120.15,189.5) .. (145,189.5) .. controls (169.85,189.5) and (190,196.22) .. (190,204.5) .. controls (190,212.78) and (169.85,219.5) .. (145,219.5) .. controls (120.15,219.5) and (100,212.78) .. (100,204.5) -- cycle ;
\draw   (220,204.5) .. controls (220,196.22) and (226.72,189.5) .. (235,189.5) .. controls (243.28,189.5) and (250,196.22) .. (250,204.5) .. controls (250,212.78) and (243.28,219.5) .. (235,219.5) .. controls (226.72,219.5) and (220,212.78) .. (220,204.5) -- cycle ;
\draw    (70,204.5) -- (100,204.5) ;
\draw    (190,204.5) -- (220,204.5) ;
\draw    (146,89.33) -- (146,117.2) ;
\draw [shift={(146,119.2)}, rotate = 270] [color={rgb, 255:red, 0; green, 0; blue, 0 }  ][line width=0.75]    (10.93,-3.29) .. controls (6.95,-1.4) and (3.31,-0.3) .. (0,0) .. controls (3.31,0.3) and (6.95,1.4) .. (10.93,3.29)   ;
\draw  [dash pattern={on 0.84pt off 2.51pt}] (30,40) -- (260,40) -- (260,230) -- (30,230) -- cycle ;

\draw (197,57.9) node [anchor=north west][inner sep=0.75pt]    {$\Gamma _{3}$};
\draw (137,57.9) node [anchor=north west][inner sep=0.75pt]    {$\Gamma _{2}$};
\draw (77,57.9) node [anchor=north west][inner sep=0.75pt]    {$\Gamma _{1}$};
\draw (103.5,45.4) node [anchor=north west][inner sep=0.75pt]    {$\not{\sim }$};
\draw (163.5,45.4) node [anchor=north west][inner sep=0.75pt]    {$\not{\sim }$};
\draw (103,138.4) node [anchor=north west][inner sep=0.75pt]    {$\gamma _{2,1}$};
\draw (163,138.4) node [anchor=north west][inner sep=0.75pt]    {$\gamma _{2,3}$};
\draw (223,138.4) node [anchor=north west][inner sep=0.75pt]    {$\gamma _{3,2}$};
\draw (43,138.4) node [anchor=north west][inner sep=0.75pt]    {$\gamma _{1,2}$};
\draw (74.5,125.9) node [anchor=north west][inner sep=0.75pt]    {$\mathscr{A}$};
\draw (194.5,125.9) node [anchor=north west][inner sep=0.75pt]    {$\mathscr{A}$};
\draw (134.5,125.92) node [anchor=north west][inner sep=0.75pt]    {$\mathscr{C}_{2}$};
\draw (113,197.9) node [anchor=north west][inner sep=0.75pt]    {$\gamma _{2,1} =\gamma _{2,3}$};
\draw (223,197.9) node [anchor=north west][inner sep=0.75pt]    {$\gamma _{3,2}$};
\draw (74.5,185.4) node [anchor=north west][inner sep=0.75pt]    {$\mathscr{A}$};
\draw (194.5,185.4) node [anchor=north west][inner sep=0.75pt]    {$\mathscr{A}$};
\draw (43,197.9) node [anchor=north west][inner sep=0.75pt]    {$\gamma _{1,2}$};
\draw (140,166.07) node [anchor=north west][inner sep=0.75pt]  [font=\scriptsize]  {$or$};

\end{tikzpicture}
    \caption{Graphical description of the transitive character of polymer incompatibility in terms of contours.}
    \label{fig: type-of-incompatibility-edges}
\end{figure}
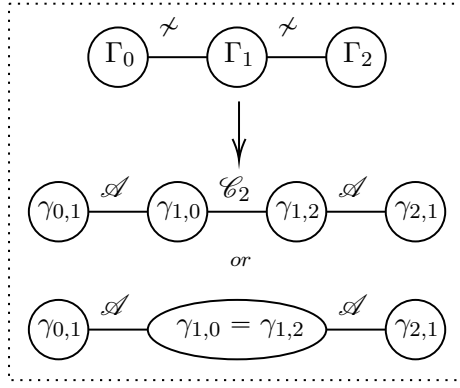

\begin{notation}
    Let $n\geq 2$.
    \begin{enumerate}
        \item Write $\mathcal{I}_n\coloneqq\{(i,j)\in[n]^2:\lvert x-y\rvert=1\}$.
        \item Write $\mathcal{J}_n\coloneqq\mathrm{P}([2,n-1])$. Note that $\mathcal{J}_2=\{\emptyset\}$.
        \item Write $\mathrm{J}_n\in\mathcal{J}_n$ and $\mathrm{J}_n^c\coloneqq[2,n-1]\setminus\mathrm{J}_n$.
    \end{enumerate}
\end{notation}

Thus,
\begin{equation*}
    \ind_{\Gamma_1\not\sim\Gamma_2}\ind_{\Gamma_2\not\sim\Gamma_3}\leq\sum_{\Upsilon\in\mathscr{C}^{\mathcal{I}_3}}\ind_\mathscr{A}(\gamma_{1,2},\gamma_{2,1})\tilde{\ind}(\gamma_{2,1},\gamma_{2,3})\ind_{\mathscr{A}}(\gamma_{2,3},\gamma_{3,2})\ind_{\gamma_{1,2}\in\Gamma_1}\ind_{\gamma_{2,1},\gamma_{2,3}\in\Gamma_2}\ind_{\gamma_{3,2}\in\Gamma_3},
\end{equation*}where $\tilde{\ind}:\mathscr{C}^2\rightarrow\{0,1\}$ is given by \begin{equation}\label{eq: tilde-indicator-definition}
    \tilde{\ind}(\gamma_1,\gamma_2)\coloneqq\ind_{\gamma_1=\gamma_2}+\ind_{\gamma_1\parallel\gamma_2}.
\end{equation}This generalizes rather easily to linear trees: let $(\Gamma_1,\Gamma_2,...,\Gamma_{n})\in\mathscr{P}^{n}$, where $n\geq 3$, then \begin{equation*}\label{eq: first-try-linear-incompatibilities}
    \prod_{i=1}^{n-1}\ind_{\Gamma_i\not\sim\Gamma_{i+1}}\leq\sum_{\Upsilon\in\mathscr{C}^{\mathcal{I}_n}}\prod_{i=1}^{n-1}\ind_{\mathscr{A}}(\gamma_{i,i+1},\gamma_{i+1,i})\prod_{j=2}^{n-1}\tilde{\ind}(\gamma_{j,j-1},\gamma_{j,j+1})\prod_{(k,l)\in\mathcal{I}_k}\ind_{\gamma_{k,l}\in\Gamma_k}.
\end{equation*}Writing \begin{equation*}\label{eq: J-tilde-indicator-function}
    \prod_{j=2}^{n-1}\tilde{\ind}(\gamma_{j,j-1},\gamma_{j,j+1})=\hspace{-0.25cm}\sum_{\mathrm{J}_n\in\mathcal{J}_n}\prod_{j_1\in \mathrm{J}_n}\ind_{\gamma_{j_1-1,j_1}\parallel\gamma_{j_1,j_1+1}}\hspace{-.25cm}\prod_{j_2\in\mathrm{J}_n^c}\hspace{-.25cm}\ind_{\gamma_{j_2,j_2-1}=\gamma_{j_2,j_2+1}}\eqqcolon\sum_{J_n\in\mathcal{J}_n}\tilde{\ind}_{J_n}(\Upsilon),
\end{equation*}we obtain \begin{equation}\label{eq: linear-incompatibilities-final-formula}
    \prod_{i=1}^{n-1}\ind_{\Gamma_i\not\sim\Gamma_{i+1}}\leq \sum_{\mathrm{J}_n\in\mathcal{J}_n}\sum_{\Upsilon\in\mathscr{C}^{\mathcal{I}_n}}\tilde{\ind}_{\mathrm{J}_n}(\Upsilon)\tilde{\ind}_{\mathscr{A}}(\Upsilon)\prod_{(k,l)\in\mathcal{I}_n}\ind_{\gamma_{k,l}\in\Gamma_k},
\end{equation}where \begin{equation*}
    \tilde{\ind}_\mathscr{A}(\Upsilon)\coloneqq\prod_{i=1}^{n-1}\ind_{\mathscr{A}}(\gamma_{i,i+1},\gamma_{i+1,i}).
\end{equation*}

\begin{lemma}\label{lem: concrete-overline-E-beta-bound}
    There exists $\beta_0>0$ such that if $\beta\geq\beta_0$ and $\varGamma_1,\varGamma_2\in\mathscr{P}$, then \begin{equation}\label{eq: concrete-overline-E-beta-bound}
        \overline{\E}_\beta(\varGamma_1,\varGamma_2)\leq (8M)^\alpha \sum_{\substack{\gamma_1\in\varGamma_1\\\gamma_2\in\varGamma_2}}e^{2c_0H(\gamma_1)}e^{2c_0H(\gamma_2)}\sum_{x_1,x_2}\ind_{x_1\in\I_-(\gamma_1)}\ind_{x\in\I_-(\gamma_2)}.
    \end{equation}
\end{lemma}

\begin{remark}
    This is very similar to the type of bound announced in Equation~\eqref{eq: total-edge-function-bound}.
\end{remark}

\begin{proof}
    By Equation\eqref{eq: linear-incompatibilities-final-formula} and Corollary~\ref{cor: sum-tree-of-polymers}, we have \begin{equation*}
    \begin{split}
        \overline{\E}_{\beta}(\varGamma_1,\varGamma_2)&\leq\sum_{n=2}^\infty\sum_{J_n\in\mathcal{J}_n}\sum_{\Upsilon\in\mathscr{C}^{\mathcal{I}_n}}\tilde{\ind}_{\mathrm{J}_n}(\Upsilon)\tilde{\ind}_\mathscr{A}(\Upsilon)\sum_{\mathbf{\Gamma}\in\mathscr{P}^n}\ind_{\Gamma_1=\varGamma_1}\ind_{\Gamma_n=\varGamma_2}\prod_{(k,l)\in\mathcal{I}_n}\ind_{\gamma_{k,l}\in\Gamma_k}\prod_{j=2}^{n-1}\bar{z}_{\beta/2}(\Gamma_j)\\
        &\eqqcolon\sum_{n=2}^\infty\sum_{J_n\in\mathcal{J}_n}\sum_{\Upsilon\in\mathscr{C}^{\mathcal{I}_n}}\tilde{\ind}_{\mathrm{J}_n}(\Upsilon)\tilde{\ind}_\mathscr{A}(\Upsilon)\ind_{\gamma_{1,2}\in\varGamma_1}\ind_{\gamma_{n,n-1}\in\varGamma_2}\prod_{j=2}^{n-1}\mathcal{Z}(\{\gamma_{j,j-1},\gamma_{j,j+1}\}).
    \end{split}
    \end{equation*}
    By Corollaries~\ref{cor: polymers-reduced-to-contours} and \ref{cor: sum-tree-of-polymers} and Equation~\eqref{eq: bar-zeta-compared-to-W}, if $j\in\mathrm{J}_n$, then
    \[
        \mathcal{Z}(\{\gamma_{j,j-1},\gamma_{j,j+1}\})\leq\W_{\beta/4}(\gamma_{j,j-1},\gamma_{j,j+1})\leq e^{-\frac{1}{16}\beta H(\gamma_{j,j-1})}\Phi(\gamma_{j,j-1},\gamma_{j,j+1})e^{-\frac{1}{16}\beta H(\gamma_{j,j+1})}.
    \]
    If $j\in\mathrm{J}_n^c$, then 
    \[
        \mathcal{Z}(\{\gamma_{j,j-1},\gamma_{j,j+1}\})\leq\ind_{\gamma_{j,j-1}=\gamma_{j,j+1}}\W_{\beta/4}(\gamma_{j,j-1},\gamma_{j,j+1})\leq e^{-\frac{1}{16}\beta H(\gamma_{j,j-1})}e^{-\frac{1}{16}\beta H(\gamma_{j,j+1})}.
    \]
    Finally, it is useful to write $\ind_n(\varGamma_1,\varGamma_2)\coloneqq \ind_{\gamma_{1,2}\in\varGamma_1}\ind_{\gamma_{n,n-1}\in\varGamma_2}$ and
    \[
        \mathcal{Z}^{\mathrm{J}_n}_{\beta}(\Upsilon)\coloneqq\prod_{j=2}^{n-1}e^{-\beta H(\gamma_{j,j-1})}e^{-\beta H(\gamma_{j,j+1})}\prod_{j\in\mathrm{J}_n}\Phi(\gamma_{j,j-1},\gamma_{j,j+1})
    \]
    Therefore,
    \[
        \begin{split}
            \overline{\E}_\beta(\varGamma_1,\varGamma_2)\leq\sum_{n=2}^\infty\sum_{\mathrm{J}_n\in\mathcal{J}_n}\sum_{\Upsilon\in\mathscr{C}^{\mathcal{I}_n}}&\tilde{\ind}_{\mathrm{J}_n}(\Upsilon)\tilde{\ind}_{\mathscr{A}}(\Upsilon)\ind_n(\varGamma_1,\varGamma_2)\mathcal{Z}^{\mathrm{J}_n}_{\beta/16}(\Upsilon).
        \end{split}
    \]
    In order to proceed, Lemma~\ref{lem: contour-to-point-bound} is necessary. Hence, we introduce lattice sites:
    \[
        \begin{split}
            \overline{\E}_\beta(\varGamma_1,\varGamma_2)&\leq \sum_{n=2}^\infty\sum_{\mathrm{J}_n\in\mathcal{J}_n}\sum_{\mathbf{x}\in\mathbb{Z}^{\mathcal{I}_n}}\sum_{\Upsilon\in\mathscr{C}^{\mathcal{I}_n}}\prod_{(k,l)\in\mathcal{I}_n}\ind_{x_{k,l}\in\I_-(\gamma_{k,l})}\cdot\tilde{\ind}_{\mathrm{J}_n}(\Upsilon)\tilde{\ind}_{\mathscr{A}}(\Upsilon)\ind_n(\varGamma_1,\varGamma_2)\mathcal{Z}^{\mathrm{J}_n}_{\beta/16}(\Upsilon)\\
            &\coloneqq\sum_{n=2}^\infty\sum_{\mathrm{J}_n\in\mathcal{J}_n}\sum_{\mathbf{x}\in\mathbb{Z}^{\mathcal{I}_n}}\sum_{\Upsilon\in\mathscr{C}^{\mathcal{I}_n}}\ind_n(\mathbf{x},\Upsilon)\tilde{\ind}_{\mathrm{J}_n}(\Upsilon)\tilde{\ind}_{\mathscr{A}}(\Upsilon)\ind_n(\varGamma_1,\varGamma_2)\mathcal{Z}^{\mathrm{J}_n}_{\beta/16}(\Upsilon).
        \end{split}
    \]
    First, if $\ind_n(\mathbf{x},\Upsilon)=1$, then
    \[
        \tilde{\ind}_{\mathscr{A}}(\Upsilon)\leq \prod_{i=1}^{n-1}(4M)^2e^{2c_0 H(\gamma_{i,i+1})}e^{2c_0H(\gamma_{i+1,i})}C(x_{i,i+1},x_{i+1,i}).
    \]
    If $j\in\mathrm{J}_n$ and $\ind_n(\mathbf{x},\Upsilon)=1$, then
    \[
        \Phi(\gamma_{j,j-1},\gamma_{j,j+1})\leq 2^6e^{3c_0H(\gamma_{j,j-1})}e^{3c_0H(\gamma_{j,j+1})}C(x_{j,j-1},x_{j,j+1}).
    \]
    Finally, if $j\in\mathrm{J}_n^c$ and $\ind_n(\mathbf{x},\Upsilon)=1$, then
    \[
        1\leq e^{c_0H(\gamma_{j,j-1})}e^{c_0H(\gamma_{j,j+1})}C(x_{j,j-1},x_{j,j+1}).
    \]
    Therefore, by writing $\mathcal{I}_n'\coloneqq\{(j,k)\in\mathcal{I}_n:2\leq j\leq n-1\}$,
    \[
        C(\mathbf{x})\coloneqq C(x_{1,2},x_{2,1})C(x_{2,1},x_{2,3})...C(x_{n-1,n-2},x_{n-1,n})C(x_{n-1,n},x_{n,n-1}).
    \]
    and
    \[
        \tilde{\ind}_n(\varGamma_1,\varGamma_2)\coloneqq\sum_{\gamma_1,\gamma_2\in\mathscr{C}}\ind_{x_{1,2}\in\I_-(\gamma_1)}\ind_{x_{n,n-1}\in\I_-(\gamma_2)}\ind_{\gamma_1\in\varGamma_1}\ind_{\gamma_2\in\varGamma_2}e^{2c_0H(\gamma_1)}e^{2c_0H(\gamma_2)},
    \]
    inputing the estimates arisen from Lemma~\ref{lem: contour-to-point-bound}, we obtain
    \[
        \begin{split}
            \overline{\E}_\beta(\varGamma_1,\varGamma_2)\leq (4M)^2\sum_{n=2}^\infty\sum_{\mathrm{J}_n\in\mathcal{J}_n}\sum_{\mathbf{x}\in\mathbb{Z}^{\mathcal{I}_n}}&\tilde{\ind}_n(\varGamma_1,\varGamma_2)C(\mathbf{x})\left(2^6(4M)^2\right)^{n-2}\\
            &\cdot\prod_{(k,l)\in\mathcal{I}'_n}\sum_{\gamma_{k,l}\in\mathscr{C}}\ind_{x_{k,l}\in\I_-(\gamma_{k,l})}e^{\left(5c_0-\frac{1}{16}\beta\right)H(\gamma_{k,l})}.
        \end{split}
    \]
    If $\beta\geq 160c_0$, then, by Hypothesis~\ref{hyp: phase-transition-estimate} and $\lvert\mathcal{I}'_n\rvert=2(n-2)$,
    \[
        \overline{\E}_\beta(\varGamma_1,\varGamma_2)\leq (4M)^2\sum_{n=2}^\infty\left(2^{10}M^2e^{-\frac{1}{16}c_2\beta}\right)^{n-2}\sum_{\mathrm{J}_n\in\mathcal{J}_n}\sum_{\mathbf{x}\in\mathbb{Z}^n}\tilde{\ind}_n(\varGamma_1,\varGamma_2)C(\mathbf{x}).
    \]
    The next step is to apply Equation~\eqref{eq: total-edge-lattice-sites} and sum over $\mathrm{J}_n\in\mathcal{J}_n$, which produces a combinatorial factor $2^{n-2}$. Thus,
    \[
        \sum_{\mathbf{x}\in\mathbb{Z}^n}\tilde{\ind}_n(\varGamma_1,\varGamma_2)C(\mathbf{x})\leq[32(1+2\zeta(\alpha))]^{2(n-2)}\sum_{\substack{\gamma_1\in\varGamma_1\\\gamma_2\in\varGamma_2}}e^{2c_0H(\gamma_1)}e^{2c_0H(\gamma_2)}\sum_{\substack{x_1\in\I_-(\gamma_1)\\x_2\in\I_-(\gamma_2)}}C(x_1,x_2).
    \]
    Finally, summing the geometric series,
    \[
        \overline{\E}_\beta(\varGamma_1,\varGamma_2)\leq \frac{(4M)^2}{1-2^{20}(1+2\zeta(\alpha))^2M^2e^{-\frac{1}{16}c_2\beta}}\sum_{\substack{\gamma_1\in\varGamma_1\\\gamma_2\in\varGamma_2}}e^{2c_0H(\gamma_1)}e^{2c_0H(\gamma_2)}\sum_{\substack{x_1\in\I_-(\gamma_1)\\x_2\in\I_-(\gamma_2)}}C(x_1,x_2),
    \]
    which concludes the proof.
\end{proof}

\subsection{Two-point correlations}\label{subsec: two-point-correlations}

In this subsection, we employ Lemma~\ref{lem: concrete-overline-E-beta-bound} to obtain an upper bound to the two-point correlation function.

\begin{proof}[Proof of Theorem~\ref{theo: two-point-correlations}]
    Given the discussion right before Subsection~\ref{subsec: total-edge-function-bounds}, it is sufficient to bound
    \[
        \overline{\W}^*_\beta(x,y)\coloneqq\sum_{\mathbf{\Gamma}\in\mathscr{P}^2}\ind_{x\in\I_-(\Gamma_1)\ind_-(\Gamma_2)}\overline{\W}_\beta(\Gamma_1)\overline{\E}_\beta(\Gamma_1,\Gamma_2)\overline{\W}_\beta(\Gamma_2).
    \]
    Write $\mathcal{I}\coloneqq\{(1,0),(1,2),(2,1),(2,0)\}$. By Corollary~\ref{cor: sum-tree-of-polymers} and Lemma~\ref{lem: concrete-overline-E-beta-bound},
    \[
        \begin{split}
            \overline{\W}^*_\beta(x,y)\leq &\sum_{\mathbf{x}\in\mathbb{Z}^\mathcal{I}}\ind_{x_{1,0}=x}\ind_{x_{2,0}=y}C(x_{1,2},x_{2,1})\\
            &\cdot\sum_{\Upsilon\in\mathscr{C}^\mathcal{I}}\prod_{\mathbf{i}\in\mathcal{I}}\ind_{x_{\mathbf{i}}\in\I_-(\gamma_{\mathbf{i}})}\cdot\tilde{\ind}(\gamma_{1,0},\gamma_{1,2})\tilde{\ind}(\gamma_{2,1},\gamma_{2,0})e^{2c_0H(\gamma_{1,0})}e^{2c_0H(\gamma_{2,0})}\\
            &\cdot\sum_{\mathbf{\Gamma}\in\mathscr{P}^2}\ind_{\gamma_{1,0},\gamma_{1,2}\in\Gamma_1}\ind_{\gamma_{2,1},\gamma_{2,0}\in\Gamma_2}\bar{z}_{\beta/2}(\Gamma_1)\bar{z}_{\beta/2}(\Gamma_2).
        \end{split}
    \]
    There is no need to go into all the details of a computation which was already done in the previous subsection. Essentially, we separate the four possible cases arising from $\tilde{\ind}(\gamma_{1,0},\gamma_{1,2})\tilde{\ind}(\gamma_{2,1},\gamma_{2,0})$. This allows us to remove the sum over polymer via Equation~\eqref{eq: bar-zeta-compared-to-W} and Corollaries~\ref{cor: polymers-reduced-to-contours} and \ref{cor: two-contour-estimate}. Then, it is sufficient to input Lemma~\ref{lem: contour-to-point-bound}, which allows us to eliminate the sums over contours. The final step is simply a triple iteration of Lemma~\ref{lem: C-xyz-bound}, which concludes the proof.
\end{proof}

\section{Concluding Remarks}

This work successfully continues the program of obtaining results for the entire region of phase transition $1<\alpha\leq 2$, whenever they are true, and erasing the hypothesis that the interaction between first-neighbors should be arbitrarily large. 

The matter of phase-separation remains open in the case $\alpha\in (1,3-3\log2]\sqcup\{2\}$: the case $1<\alpha <3-3\log2$ should be rather easy given \cite{Cassandro.Merola.Picco.17} and this work. We point out that another relatively natural avenue of further research is the study of the decay of correlations in the presence of a random field when $\alpha\in (1,3/2)$, especially given recent work by Ding, Huang, and Maia \cite{Ding}. 

\section*{Acknowledgements}

RB was supported by CNPq grants 311658/2025-3 and 407527/2025-7. This study was financed in part by the Coordena\c{c}\~ao de Aperfei\c{c}oamento de Pessoal de N\'ivel Superior - Brasil (CAPES) - Finance Code 001. HC was supported by CAPES grant 160295/2024-6 and by CNPq grant 140800/2022-0. We would like to thank Lucas Affonso, Jo\~ao Maia, Pierre Picco, and Kelvyn Welsch for many discussions about the long-range Ising model, and the contours and cluster expansion arguments employed in studying it. Without their extensive prior work, this paper would not have been possible.

\bibliographystyle{habbrv}
\bibliography{refs}

\end{document}